\documentclass[11pt]{article}
\usepackage{amsmath,amssymb,amsfonts,latexsym,xspace,amsthm,array,float}
\usepackage[utf8]{inputenc}
\usepackage{a4wide,hyperref}
\usepackage[lined,boxed,commentsnumbered]{algorithm2e}
\usepackage{tikz}
\usepackage{wrapfig}
\usetikzlibrary{patterns}
\usepackage[final]{showkeys}

\newenvironment{pf}{{\em \noindent Proof:}}{ \hfill \qed\smallskip}

\tikzstyle{Grisfill}=[fill=black!10]
\colorlet{darkgreen}{green!50!black}
\colorlet{darkred}{red!50!black}

\tikzstyle{H}=[red]
\tikzstyle{Hfill}=[fill=red!20]
\tikzstyle{Hpoint}=[fill=darkred,darkred]

\tikzstyle{V}=[green]
\tikzstyle{Vfill}=[fill=green!20]
\tikzstyle{Vpoint}=[darkgreen,fill=darkgreen]

\definecolor{vert}{rgb}{0,0.75,0.25}
\newcommand{\R}{\ensuremath{\color{red}R}\xspace}
\newcommand{\G}{\ensuremath{\color{vert}G}\xspace}

\newcommand{\RRR}{\ensuremath{\color{red}132}\xspace}
\newcommand{\GGR}{\ensuremath{{\color{vert}1}{\color{red}X}{\color{vert}2}}\xspace}
\newcommand{\RRG}{\ensuremath{{\color{vert}2}/{\color{red}1}{\color{red}3}}\xspace}
\newcommand{\GGG}{\ensuremath{\color{vert}213}\xspace}

\makeatletter
\newcommand{\rmnum}[1]{\romannumeral #1}
\newcommand{\Rmnum}[1]{\expandafter\@slowromancap\romannumeral #1@}
\makeatother

\newcommand{\Vpoint}[2]{\draw (#1,#2) [darkgreen,fill=darkgreen] circle (3pt);}
\newcommand{\Hpoint}[2]{\draw (#1,#2) [darkred,fill=darkred] circle (3pt);}
\newcommand{\zoneD}[1]{\draw [#1,#1fill, very thick] (4,0) -- (5,0) -- (5,1) -- (4,1);}
\newcommand{\zoneE}[1]{\draw [#1,#1fill, very thick] (5,0) -- (6,0) -- (6,1) -- (5,1) -- (5,0);}
\newcommand{\zoneB}[1]{\draw [#1,#1fill, very thick] (5,-1) -- (5,0) -- (6,0) -- (6,-1);}
\newcommand{\zoneF}[1]{\draw [#1,#1fill, very thick] (7,1) -- (6,1) -- (6,0) -- (7,0);}
\newcommand{\zoneH}[1]{\draw [#1,#1fill, very thick] (5,2) -- (5,1) -- (6,1) -- (6,2);}
\newcommand{\zoneA}[1]{\draw [#1,#1fill, very thick] (5,-1) -- (5,0) -- (4,0);}

\newcommand{\Hzone}{\begin{tikzpicture}[scale=.5]
\draw [Hfill] (0,0) rectangle (1,1); 
\end{tikzpicture}~}
\newcommand{\Vzone}{\begin{tikzpicture}[scale=.5]
\draw [Vfill] (0,0) rectangle (1,1); 
\end{tikzpicture}~}

\newcommand{\zoneRG}[3]{
\draw [very thick,H,Hpoint] (#1,#2) -- +(-#3,0);
\draw [very thick,V,Vpoint] (#1,#2) -- +(0,#3);
\draw [Hfill] (#1,#2) -- +(-#3,#3) -- +(-#3,0);
\draw [Vfill] (#1,#2) -- +(-#3,#3) -- + (0,#3);
}
\newcommand{\zoneGR}[3]{
\draw [very thick,H,Hpoint] (#1,#2) -- +(-#3,0);
\draw [very thick,V,Vpoint] (#1,#2) -- +(0,#3);
\draw [Vfill] (#1,#2) -- +(-#3,#3) -- +(-#3,0);
\draw [Hfill] (#1,#2) -- +(-#3,#3) -- + (0,#3);
}

\newcommand{\zoneI}[1]{\draw [#1,#1fill, very thick] (6,2) -- (6,1) -- (7,1);}
\newcommand{\etiquette}[1]{\draw (2.5,1) node {(\rmnum{#1})};}

\newtheorem{thm}{Theorem}[section]
\newtheorem{prop}[thm]{Proposition}
\newtheorem{rem}[thm]{Remark}
\newtheorem{defn}[thm]{Definition}
\newtheorem{lem}[thm]{Lemma}

\newcommand{\patternV}{\ensuremath{|12|\ |}\xspace}
\newcommand{\patternH}{\ensuremath{|\ |132|}\xspace}
\newcommand{\patternVH}{\ensuremath{|2|13|}\xspace}

\newcommand{\pushall}{$2$-stack pushall sortable\xspace}
\newcommand{\ssi}{if and only if\xspace}
\newcommand{\ascentRG}{increasing sequence $RG$\xspace}
\newcommand{\ascentGR}{increasing sequence $GR$\xspace}
\newcommand{\ascentRR}{increasing sequence $RR$\xspace}
\newcommand{\ascentGG}{increasing sequence $GG$\xspace}
\newcommand{\ascentsRG}{increasing sequences $RG$\xspace}
\newcommand{\ascentsGR}{increasing sequences $GR$\xspace}

\newcounter{indice}

\newcommand{\permutation}[1]{
\setcounter{indice}{0};
\foreach \i in {#1} 
\addtocounter{indice}{1};

\addtocounter{indice}{1}
\draw [help lines] (1,1) grid (\theindice,\theindice);

\setcounter{indice}{1};

\foreach \i in { #1 } {
\draw (\theindice+.5,\i+.5) [fill] circle (.2);
\addtocounter{indice}{1};
}
\addtocounter{indice}{-1};
}

\title{\pushall permutations \footnote{This work was completed with the support of the ANR
   project ANR BLAN-0204\_07  MAGNUM}}
\author{Adeline Pierrot \and Dominique Rossin}

\begin{document}
\maketitle
\begin{abstract}
In the 60's, Knuth introduced stack-sorting and serial compositions of stacks.
In particular, one significant question arise out of the work of Knuth: 
how to decide efficiently if a given permutation is sortable with $2$ stacks in series?
Whether this problem is polynomial or \mbox{NP-complete} is still unanswered yet.
In this article we introduce $2$-stack pushall permutations which form a subclass of $2$-stack sortable permutations 
and show that these two classes are closely related. 
Moreover, we give an optimal ${\mathcal O}(n^{2})$ algorithm to decide if a given permutation of size $n$ is \pushall and describe all its sortings.
This result %aims at being TODO: preciser que ça a aussi un interet en soi ?
is a step to the solve the general $2$-stack sorting problem in polynomial time.
% Moreover, we characterize permutations that have a polynomial number of ways to be sorted and the others. 
%TODO: Demander à Dominique ce qu'il a voulu dire et s'il veut garder cette phrase.
\end{abstract}

\section{Introduction}

In the 60's, Knuth introduced the problem of stack-sorting \cite{Knuth68} and 
then serial compositions of stacks \cite{Knuth73}.
To answer the one-stack case, he introduced both the pattern-containment relation on permutations and permutation classes, two new fields of combinatorics. 
Stack-sorting was further generalized to sorting networks by Tarjan \cite{Tarjan72}
while several variants appear by either considering other types of combinatorial structures or by changing rules \cite{Pratt73,EvenItai71,AAL10}.

In this article, we focus on sorting with two stacks in series. 
More precisely, if $\sigma$ is a permutation, we consider $\sigma$ as a sequence of integers $\sigma_1,\sigma_2,\ldots,\sigma_n$ that we take as input
and at each step we have three possibilities as described in Figure~\ref{fig:rholambdamu} (p.\pageref{fig:rholambdamu}):
\begin{enumerate}
\item[$\rho$:] Get the next element of $\sigma$ and push its value on top of the first stack denoted $H$.
\item[$\lambda$:] Pop the topmost element of stack $H$ and push this value on top of the second stack $V$.
\item[$\mu$:] Pop the topmost element of $V$ and write it to the output.
\end{enumerate}

We iterate over these three possibilities until all elements have been output. 
If there is a sequence of operations that leads to identity on the output, 
then we say that the permutation is $2$-stack sortable. 
Three natural questions among others arise:
\begin{enumerate}
\item Decision: what is the complexity of the problem consisting of deciding whether a given permutation is sortable or not?
\item Characterization: can one characterize permutations that are sortable?
\item Counting: establish the generating function of sortable permutations.
\end{enumerate}

For the one-stack case these three problems were solved by Knuth in \cite{Knuth68}. 
A greedy algorithm allows to answer the decision problem in linear time. 
Moreover he characterized sortable permutations by introducing the $231$-avoiding permutations class, 
whose generating function is the Catalan series. 
Since this article, the more general question of sorting with multiple stacks in series or in parallel has been widely studied.
Knuth \cite{Knuth68}, Tarjan \cite{Tarjan72} and Pratt \cite{Pratt73} noted that the permutations sortable by the various configurations 
could be described by forbidding certain patterns to occur in the permutations.

Regarding $t$ parallel stacks,
the decision problem can be answered in time ${\cal O}(n \log n)$ for $t=1,2,3$,
while for $t>3$ this is NP-complete (this is proved by a reduction in~\cite{EvenItai71} to a problem solved in~\cite{Unger92}).
The characterization problem is studied in~\cite{Pratt73}:
for $t>1$, the basis of the class of permutations sortable with $t$ stacks in parallel is infinite.
Finally, about the counting problem,
when $t=2$ the generating function is described in \cite{ABM}, but by an infinite system of equations.

% For $t$ stacks in parallel, an efficient algorithm (complexity ${\cal O}(n \log n)$) exists to decide whether a given permutation 
% is sortable when $t = 1,2,3$,
% while for $t>3$ this is NP-complete (this is proved by combining results of \cite{EvenItai71} and \cite{Unger92}).
% For $t>1$, the basis of the class of permutations sortable with $t$ stacks in parallel is infinite \cite{Pratt73}.
% When $t=2$, the generating function of sortable permutations is described by an infinite system of equations \cite{ABM}. 

For stacks in series, it has been shown in \cite{Knuth68} % TODO: verifier que c'est bien dans ce volume (c'est ce que dit Bona en tout cas)
that every permutation of size $n$ can be sorted by $\log_2(n)$ stacks in series. 
But none of the above three questions has been answered for more than one stack in series. 
For two stacks, Murphy \cite{Murphy02} proved that the basis of the class of sortable permutations is infinite. 
In his Phd thesis, he also studies the problem of deciding whether a given permutation is sortable with $2$ stacks in series.
He reduced this problem to a $3$-SAT problem;
at the same time he reduced a $2$-SAT instance to the decision problem,
and hoped than one of both reduction was actually an equivalence.
But none of those results has been proved or disproved. 
In \cite{Bona02}, B\'ona gives an overview of advances in sorting networks and mentions this problem as possibly NP-complete.
% even recently, this problem has a central place in the book. Lequel ? Pas celui \cite{Kitaev11} de Kitaev en tout cas !
More surprising, both conjectures exist:
in \cite{AMR02b}, the authors conjectured that the decidability problem is NP-complete,
while Murphy in his PhD (\cite{Murphy02} Conjecture 260) conjectured that it is in $P$.
Several weaker variants of this problem have been studied. 
First, West considered permutations sortable with two consecutive greedy passes through a stack in \cite{West90, West93}. 
He conjectured the enumeration formula which was proved after by Zeilberger \cite{Zeilberger92}. 
For more than two passes, few results are known~\cite{BM00, Ulfarsson11}. 
Another variant studied in~\cite{AMR02b} is to consider decreasing stacks (i.e.~elements in the stack must be decreasing from bottom to top)
instead of general stacks.
In this article we define a new restriction of $2$-stacks sorting, namely $2$-stacks pushall sorting,
and prove that the decidability problem in this case is polynomial.

\bigskip
Throughout this article we usually write permutations as $\sigma = \sigma_1\sigma_2\ldots\sigma_n$ 
where $n$ is the size of $\sigma$, denoted by $|\sigma|$, and $\sigma_i$ is the image of $i$ for all $i \leq n$. 
A permutation $\pi = \pi_1\pi_2\ldots\pi_k$ is a {\em pattern} of $\sigma$ if and only if there exist 
$1 \leq i_1 \leq i_2 \leq i_3\leq \ldots\leq i_k \leq n$ such that $\sigma_{i_1} \sigma_{i_2} \ldots \sigma_{i_k}$ is order isomorphic to $\pi$. 
We note $Av(B)$ the set of permutations avoiding $B$, i.e.~not having any permutation of $B$ as a pattern.
A {\em permutation class} $\mathcal{C}$ is a set of permutations downward-closed for the pattern relation: 
if $\sigma$ belongs to $\mathcal{C}$, then every pattern of $\sigma$ belongs to $\mathcal{C}$. 
Note that for any set $B$, $Av(B)$ is a class.
A permutation class ${\cal C}$ can be defined by its minimal set $B$ such that ${\cal C}=Av(B)$.
% of excluded permutations. % -- minimal for the pattern relation.
This minimal set is called the {\em basis} of the class. 
For example, Knuth proved that $1$-stack sortable permutations are those that belong to the class $Av(231)$.
Unfortunately, the basis can be infinite. 
For $2$-stack sortable permutations, as stated above, it has been proved in~\cite{Murphy02} that the basis is infinite.

A permutation can also be represented by its
{\it diagram}, consisting in the set of points at coordinates $(i,\sigma_i)$ drawn in the plane
(see two examples in Figure~\ref{fig:graphicalRepresentation}).
An {\em interval} in a permutation is a consecutive range of elements, consecutive both in indices and values. 
For example in the permutation $4\,7\,9\,6\,8\,1\,3\,2\,5$, the elements $7\,9\,6\,8$ form an interval: 
they are consecutive in the permutation and the values span the whole integer interval $[6\ldots 9]$. 
In the diagram, note that an interval is a square
which is itself a diagram of a permutation (if translated to the origin).
 % in which is represented a permutation 
In particular, no point outside this square has the same x or y coordinate than any cell of the square
 % lies in line with it
(see the yellow stripes of Figure~\ref{fig:graphicalRepresentation}).
A permutation where all intervals are trivial --either a singleton or the whole permutation-- is called a {\em simple} permutation. 
For instance, $2\,4\,1\,3$ and $3\,1\,4\,2$ are the two simple permutations of size $4$.
An {\em inflation} of an element $\sigma_i$ in $\sigma$ by a permutation $\pi$ is the permutation obtained 
by replacing $\sigma_i$ by $\pi$ and renormalizing the resulting permutation. 
For example if we inflate $3$ in $2\,3\,1\,4$ by the permutation $4\,1\,5\,2\,3$, 
we obtain the permutation $2\,{\bf 6\,3\,7\,4\,5}\,1\,8$
(see the second diagram of Figure~\ref{fig:graphicalRepresentation}).
Notice that in an inflation by $\pi$, elements corresponding to $\pi$ form an interval in the resulting permutation.

\begin{figure}[H]
\begin{center}
\begin{tikzpicture}[scale=.3]
\useasboundingbox (0,0) rectangle (11,9);
\fill [color=yellow!20!white] (0,5) rectangle +(1,4);
\fill [color=yellow!20!white] (1,0) rectangle +(4,5);
\fill [color=yellow!20!white] (5,5) rectangle +(4,4);
\draw[thin] (0,0) grid (9,9);
\draw[fill] (.5,3.5) circle (6pt);
\draw[fill] (1.5,6.5) circle (6pt);
\draw[fill] (2.5,8.5) circle (6pt);
\draw[fill] (3.5,5.5) circle (6pt);
\draw[fill] (4.5,7.5) circle (6pt);
\draw[fill] (5.5,0.5) circle (6pt);
\draw[fill] (6.5,2.5) circle (6pt);
\draw[fill] (7.5,1.5) circle (6pt);
\draw[fill] (8.5,4.5) circle (6pt);
\draw [very thick] (1,5) rectangle +(4,4);
\end{tikzpicture}
\begin{tikzpicture}[scale=.3]
\draw[thin] (0,0) grid (8,8);
\draw[fill] (.5,1.5) circle (6pt);
\draw[fill] (1.5,5.5) circle (6pt);
\draw[fill] (2.5,2.5) circle (6pt);
\draw[fill] (3.5,6.5) circle (6pt);
\draw[fill] (4.5,3.5) circle (6pt);
\draw[fill] (5.5,4.5) circle (6pt);
\draw[fill] (6.5,0.5) circle (6pt);
\draw[fill] (7.5,7.5) circle (6pt);
\draw [very thick] (1,2) rectangle +(5,5);
\end{tikzpicture}
\caption{Diagram of $4\,7\,9\,6\,8\,1\,3\,2\,5$ and the inflation of $3$ in $2\,3\,1\,4$ by $4\,1\,5\,2\,3$
\label{fig:graphicalRepresentation}}
\end{center}
\end{figure}
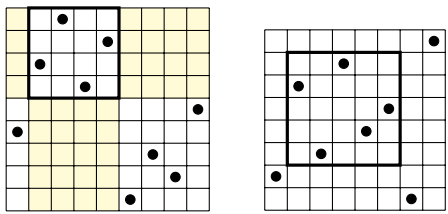

We denote inflations by $\sigma = \tau[\pi^{(1)},\pi^{(2)},\ldots,\pi^{(k)}]$ where $\tau$ is a permutation of size $k$ and $\tau_i$ is inflated by $\pi^{(i)}$ for all $i$.
When $\tau$ is the identity $1\, 2 \dots k$ (resp. the decreasing permutation $k \dots 1$) we write $\sigma = \oplus[\pi^{(1)}, \pi^{(2)}, \ldots, \pi^{(k)}]$ (resp. $\ominus[\pi^{(1)}, \pi^{(2)}, \ldots, \pi^{(k)}]$).

A permutation $\sigma$ is $\oplus$-decomposable (resp. $\ominus$-decomposable) if it can be written
$\sigma = \oplus[\pi^{(1)},\pi^{(2)},\pi^{(3)},\ldots,\pi^{(k)}]$
(resp.  $\ominus[\pi^{(1)},\pi^{(2)},\pi^{(3)},\ldots,\pi^{(k)}]$) with $k > 1$.
Otherwise $\sigma$ is $\oplus$-indecomposable (resp. $\ominus$-indecomposable)

A decomposition theorem \cite{AA05} states that any permutation $\sigma\neq1$ can be written in a unique way as either:
\begin{itemize}
\item $\sigma = \oplus[\pi^{(1)},\pi^{(2)},\pi^{(3)},\ldots,\pi^{(k)}]$ where $k \geq 2$ and the $\pi^{(i)}$ are $\oplus$-indecomposable.
\item $\sigma = \ominus[\pi^{(1)},\pi^{(2)},\pi^{(3)},\ldots,\pi^{(k)}]$ where $k \geq 2$ and the $\pi^{(i)}$ are $\ominus$-indecomposable.
\item$\sigma = \tau[\pi^{(1)},\pi^{(2)},\pi^{(3)},\ldots,\pi^{(k)}]$ where $k\geq 4$ and $\tau$ is simple.
\end{itemize}

In the next section we study $2$-stack sorting and $2$-stack pushall sorting and show the close correlation between these two models. 
This combinatorial study concludes on some partial characterization of both classes in terms of permutations they contain or permutations in the basis.
The key idea is to use the block-decomposition of permutations given in the above theorem. 

Then in section~\ref{sec:coloring} we prove that $2$-stack pushall sorting can be expressed as a $2$-color problem on the diagram of permutations.
Moreover we characterize diagram of permutations that can be colored. 
This characterization leads to a polynomial algorithm to check whether a permutation is $2$-stack pushall sortable by finding all colorings for its diagram. 

Section~\ref{sec:optimalAlgo} refines the results of section~\ref{sec:coloring} by limiting the number of colorings to test.
This leads to an optimal algorithm computing in quadratic time a linear representation of all pushall sortings of a given permutation,
which thus decides whether a permutation is \pushall.

To conclude, we give in section~\ref{sec:conclusion} some natural continuations of our work.

\section{$2$-stack sorting vs $2$-stack pushall sorting}

In this section we define pushall sorting and point out the close link between $2$-stack sorting and $2$-stack pushall sorting. 
Moreover, for each of these sorting problems we exhibit some recursive necessary and sufficient conditions for a permutation to be sortable depending on the root of its block-decomposition.

%\subsection{A word approach}
In $2$-stack sorting, three different operations are allowed as pictured in Figure~\ref{fig:rholambdamu}. 
Each of this operation can be encoded with a letter (see Figure~\ref{fig:rholambdamu}).
For example, whenever an element is popped from stack $H$ and pushed in stack $V$, we write $\lambda$. 
A sequence of operations is encoded by a word whose length is the number of operations performed.

\begin{figure}[H]
\begin{center}
\begin{tikzpicture}
\draw[very thick] (0,2) -- (0,0) -- (1,0) -- (1,2);
\draw (2.5,-0.5) node {$H$};
\draw[very thick] (2,2) -- (2,0) -- (3,0) -- (3,2);
\draw (0.5,-0.5) node {$V$};
\draw[very thick, dashed, ->] (4,2.5) -- node[above] {$\rho$} (2.66,2.5) -- (2.66,1.7);
\draw (4.8,2.5) node {{\small INPUT}};
\draw[very thick,dashed, ->] (2.33,1.7) -- (2.33,2.5) -- node[above] {$\lambda$} (0.66,2.5) -- (0.66,1.7);
\draw[very thick, dashed, <-] (-1,2.5) -- node[above] {$\mu$} (0.33,2.5) -- (0.33,1.7);
\draw (-2.5,2.5) node {{\small OUTPUT}};
\end{tikzpicture}
\caption{Sorting with two stacks in series \label{fig:rholambdamu}}
\end{center}
\end{figure}
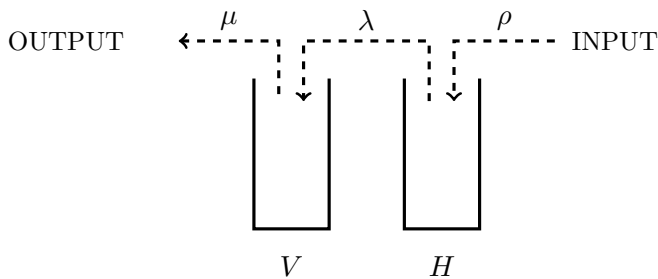

\begin{defn}
A {\em stack word} $w$ is a word over the alphabet $\{ \rho, \lambda, \mu \}$ such that $|w|_\rho = |w|_\lambda = |w|_\mu$ and for all prefix $v$ of $w$, $|w|_\rho \geq |w|_\lambda \geq |w|_\mu$. 
Intuitively it's a word which describes a sequence of appropriate stack operations which take a permutation through two stacks in series (without necessarily sorting it).
A permutation $\sigma = \sigma_1\ldots \sigma_n$ is $2$-stack sortable if and only if there exists a stack word of length $3n$ ($n$ times each letter $\rho,\lambda, \mu$) which leads to the identity in the output with $\sigma$ as input. 
Such a word is called a {\em valid} stack word for $\sigma$.
\end{defn}

There are several valid stack word for a given permutation: for example, permutation $2341$ admits either $\rho\rho\rho\rho\lambda\mu\lambda\lambda\lambda\mu\mu\mu$ or $\rho\rho\rho\lambda\lambda\lambda\rho\lambda\mu\mu\mu\mu$ as valid words.
Note also that $\rho$ and $\mu$ commutes: if $w$ is a valid stack word for $\sigma$ an $w'$ is obtained for $w$ by exchanging adjacent letters $\rho$ and $\mu$, then $w'$ is a valid stack word for $\sigma$. 
In his Phd \cite{Murphy02}, Murphy studied $2$-stack sorting by studying stack words. 
This presentation of $2$-stack sorting allow us to define formally $2$-stack pushall sorting.

\begin{defn}
A {\em pushall} stack word is a stack word such that the first occurrence of letter $\mu$ is after the last occurence of the letter $\rho$.
A permutation $\sigma$ of size $n$ is $2$-stack pushall sortable if and only it admits a valid pushall stack word.

More informally, $2$-stack pushall sortable permutations are those which can be sorted by pushing all elements in the stacks before writing any element to the output.
\end{defn}

For example $2431$ is $2$-stack pushall sortable as the word $\rho\rho\rho\rho\lambda\mu\lambda\lambda\lambda\mu\mu\mu$ respects the required condition (as does $\rho\rho\rho\lambda\lambda\lambda\rho\lambda\mu\mu\mu\mu$). 

\begin{rem}\label{rem:decompoNonUnique}
A stack word $w$ is a pushall stack word if and only if it can be written as $w = uv$ with $u \in \{\rho,\lambda\}^{*}$ and $v \in \{\lambda,\mu\}^{*}$. This decomposition is not unique.
In the preceding example, the word $w = \rho\rho\rho\rho\lambda\mu\lambda\lambda\lambda\mu\mu\mu$ admits two decompositions: $w = (\rho\rho\rho\rho)(\lambda\mu\lambda\lambda\lambda\mu\mu\mu)$ and $w = (\rho\rho\rho\rho\lambda)(\mu\lambda\lambda\lambda\mu\mu\mu)$.
\end{rem}

The previous definition of $2$-stack pushall sortable permutations implies that they form a subset of $2$-stack sortable permutations. 
Moreover it is easy to check that $2$-stack pushall sorting is stable by pattern relation: if $\sigma$ is $2$-stack pushall sortable then every pattern $\pi$ of $\sigma $ is $2$-stack pushall sortable: choose an occurrence of $\pi$ in $\sigma$ and a valid pushall stack word $w$ of $\sigma$. 
To obtain a valid pushall stack word of $\pi$, delete letters of $w$ that correspond to elements of $\sigma$ not involved in the occurrence of $\pi$. 
The same reasoning holds for general $2$-stack sorting.

\begin{prop}
$2$-stack pushall permutations form a subclass of $2$-stack sortable permutation class.
\end{prop}

Although we do not know the ratio between these two classes, there exists a close correlation between them and solving $2$-stack pushall sorting is a prerequisite for the more general case. 
We first study the possible configurations of the stacks during a sorting procedure. 
This will help us to obtain properties of stack sorting permutations thanks to their decomposition. 
In a last subsection, we study the basis of $2$-stack sortable permutation class and show how it is correlated to the $2$-stack pushall one.

\subsection{Stack configurations}

At each step of a sorting procedure, some elements of the permutation lie in the stacks. 
We call a {\em stack configuration} the position of these elements in stacks $H$ and $V$. 
In this section, we exhibit a necessary condition on stack configurations to be part of a sorting procedure.
First we define formally stack configurations.

\begin{defn}
A {\em stack configuration} is a pair of two vectors of positive integers $(\overrightarrow{V},\overrightarrow{H})$ of arbitrairy (and maybe different) sizes, such that all coordinates are distincts.
A stack configuration may be empty (if both vectors are of size zero).
Vector $\overrightarrow{V}$ (resp. $\overrightarrow{H}$) represents elements that are in stack $V$ (resp. $H$) given from bottom to top, so we can apply to stack configurations moves $\lambda$ and $\mu$, and move $\rho$ if we know what is the next integer in the input.

Let $\sigma$ be a permutation, a stack configuration of $\sigma$ is a stack configuration in which coordinates are bounded by $|\sigma|$.
\end{defn}

\begin{defn}
To each stack word $w$ of size $3n$ and permutation $\sigma$ of size $n$ we associate a sequence of $3n+1$ stack configurations $\big(c_k(w,\sigma)\big)$ describing how the sequence of moves $w= w_1 \dots w_{3n}$ take $\sigma$ through the stacks: $c_1(w,\sigma)$ is empty and we obtain $c_{k+1}(w,\sigma)$ from $c_k(w,\sigma)$ by doing operation $w_k$ with $\sigma$ as input at the beginning.
\end{defn}

\begin{defn}
Let $\sigma$ be a permutation.
A stack configuration $c$ is {\em reachable} for $\sigma$ if it exists a stack word $w$ and an integer $k$ such that $c = c_k(w,\sigma)$.
A stack configuration $c$ is {\em total} for $\sigma$ if all integers from $1$ to $|\sigma|$ appear in $c$ (this notion depends only on $|\sigma|$, we don't ask $c$ to be reachable for $\sigma$).
\end{defn}

\begin{rem}\label{rem:TotalConfigurations}
Let $w$ be a stack word of size $3n$ and $\sigma$ a permutation of size $n$.
Then $w$ is a pushall stack word if and only if at least one of the stack configurations $\big(c_k(w,\sigma)\big)$ is total.
\end{rem}

During a sorting procedure, stack configurations have constraints so that all elements can be popped out in increasing order. 
Recall that in one-stack sorting, the stack must be in decreasing order (from bottom to top). 
For two-stack sorting, we have the same decreasing constraint on stack $V$ but other constraints appear that can be represented as stack patterns.

\begin{defn}\label{def:unsortableStackPatterns}
We call {\em unsortable stack-patterns} the following three patterns, denoted respectively \patternV, \patternH and \patternVH:
\begin{center}
\begin{tikzpicture}[scale=.5]
\draw [dotted] (-1,-1) rectangle +(5,3.5);
\draw[very thick] (0,2) -- (0,0) -- (1,0) -- (1,2);
\draw (2.5,-0.5) node {$H$};
\draw[very thick] (2,2) -- (2,0) -- (3,0) -- (3,2);
\draw (0.5,-0.5) node {$V$};
\draw (0.5,0.5) node {$1$};
\draw (0.5,1.2) node {$2$};
\end{tikzpicture}
\begin{tikzpicture}[scale=.5]
\draw [dotted] (-1,-1) rectangle +(5,3.5);
\draw[very thick] (0,2) -- (0,0) -- (1,0) -- (1,2);
\draw (2.5,-0.5) node {$H$};
\draw[very thick] (2,2) -- (2,0) -- (3,0) -- (3,2);
\draw (0.5,-0.5) node {$V$};
\draw (2.5,0.5) node {$1$};
\draw (2.5,1.2) node {$3$};
\draw (2.5,1.9) node {$2$};
\end{tikzpicture}
\begin{tikzpicture}[scale=.5]
\draw [dotted] (-1,-1) rectangle +(5,3.5);
\draw[very thick] (0,2) -- (0,0) -- (1,0) -- (1,2);
\draw (2.5,-0.5) node {$H$};
\draw[very thick] (2,2) -- (2,0) -- (3,0) -- (3,2);
\draw (0.5,-0.5) node {$V$};
\draw (2.5,0.5) node {$1$};
\draw (2.5,1.2) node {$3$};
\draw (0.5,0.5) node {$2$};
\end{tikzpicture}
\end{center}
More precisely pattern \patternV means that there is in stack $V$ one element which has a smaller element below it. 
Pattern \patternH means that there is in stack $H$ one element which has a greater element below it and a smaller element more below. 
Pattern \patternVH is somehow special as the pattern is divided in both stacks. 
It means that there are elements $a,b,c$ such that $b \in V$, $a,c \in H$, $a < b < c$ and $c$ is above $a$ in stack $H$. 
\end{defn}

\begin{thm}\label{thm:popable}
A stack configuration can be popped out in increasing order if and only if it avoids each unsortable stack-pattern.
\end{thm}

\begin{proof}
Notice that if a stack configuration contains any of the $3$ unsortable stack-patterns, then elements involved in the pattern cannot be popped out in increasing order.

For the converse, we prove by induction on the number of elements in the stacks that a configuration which avoids the $3$ unsortable stack-patterns can be popped out in increasing order. 
Suppose that the result has been proved for all stack configurations with at most $k$ elements. 
Note that the result is trivially true for $k \leq 2$. 
Let $c$ be a stack configuration with $k+1$ elements avoiding the $3$ unsortable stack-patterns and $m$ the smallest element of this configuration.
We show that $m$ can be popped out so that the stack configuration of the $k$ remaining elements still avoids the $3$ unsortable stack-patterns.
Without loss of generality assume $m = 1$.

Suppose that $1$ lies in stack $V$. 
As $c$ avoids pattern \patternV, $V$ is in decreasing order so $1$ is at the top of it. 
It can be popped out and there remains $k$ elements still avoiding the $3$ unsortable stack-patterns. 
Thus they can be all popped out in increasing order by induction. 

Suppose now that $1$ lies in stack $H$. 
As $c$ avoids pattern \patternH and $1$ is the smallest element, all elements above $1$ are in increasing order (from $1$ to top). 
All these elements can be pushed onto stack $V$ so that stack $V$ remains in decreasing order. 
Indeed as $c$ avoids pattern \patternVH , the top of stack $V$ is greater than the top of stack $H$. 
When all elements greater than $1$ and above $1$ in stack $H$ are transferred onto stack $V$, then $1$ can be popped out both stacks $H$ and $V$ and the remaining configuration still avoids the $3$ unsortable stack-patterns (as $c$ avoids pattern \patternH , no pattern \patternVH has been created) and we can apply the induction.
\end{proof}

\begin{rem}\label{rem:uniquePopOut}
There is at most one way to pop out in increasing order elements from a stack configuration. 
Indeed to pop out we only use moves $\mu$ and $\lambda$, and if we want to pop out in increasing order we have to perform move $\mu$ if and only if the smallest element lies in the top of $V$.
\end{rem}

\begin{algorithm}[H]
  \SetAlgoLined
\LinesNumbered
  \KwData{$\sigma$ a permutation and $c$ a total stack configuration of $\sigma$.}
  \KwResult{True if $c$ can be popped out in increasing order.}
$i \longleftarrow 1$\;
\While{$i \leq |\sigma|$}{
\eIf{$top(V) = i$}{
pop out $top(V)$ from stack $V$ and let $i \longleftarrow i+1$ 
}{
\eIf{$H$ is non empty and $top(H) < top(V)$}{
pop $top(H)$ from stack $H$ and push it into $V$\;}{
Return false\;
}
}% prendre en compte le cas où V est vide (?)
}
Return true\;
 \caption{Pop out in increasing order\label{algo:popOut}}
\end{algorithm}

%\marginpar{Mettre ici l'algo PopOut}

\begin{prop}\label{prop:AlgoPopOut}
Let $c$ be a total stack configuration of a permutation $\sigma$.
Then Algorithm~\ref{algo:popOut} applied to $c$ returns $true$ \ssi $c$ can be popped out in increasing order.
Moreover Algorithm~\ref{algo:popOut} runs in linear time w.r.t. $|\sigma|$.
\end{prop}

\begin{pf}
At each step, Algorithm~\ref{algo:popOut} performs either a move $\mu$ or a move $\lambda$. 
As at most $|\sigma|$ moves $\mu$ and $|\sigma|$ moves $\lambda$ can be done, it runs in linear time w.r.t. $|\sigma|$.
We conclude using Remark~\ref{rem:uniquePopOut}.
\end{pf}

Theorem~\ref{thm:popable} ensures that a stack configuration can be popped out in increasing order. 
Conditions of this theorem must be verified at each step of a sorting procedure. 
This is formalised in the following proposition:

\begin{prop}\label{prop:eachConfigAvoidUnsortablePattern}
If $w$ is a valid stack word for the permutation $\sigma$, 
then each stack configuration of $\big(c_k(w,\sigma)\big)$ avoids the $3$ unsortable stack-patterns.
\end{prop}

The converse is not true: let $w = (\rho \lambda \mu)^n$ then for all permutation $\sigma$ of size $n$ 
each stack configuration of $\big(c_k(w,\sigma)\big)$ avoids the $3$ unsortable stack-patterns (as it has at most one element in the stacks).
But if $\sigma$ is not the identity, $w$ is not a valid stack word for $\sigma$.

For $2$-stack pushall sorting, however, it is sufficient to check whether the stack configuration obtained 
just after the last element of $\sigma$ has been pushed onto $H$ avoids the $3$ unsortable stack-patterns.

\begin{prop}\label{prop:pushallIffConfigurationEvitePatterns}
A permutation $\sigma$ is \pushall if and only if there is a way to put all its elements in the stacks 
so that the total stack configuration obtained avoids the three unsortable patterns.
\end{prop}

\begin{pf}
If $\sigma$ is \pushall we conclude using Proposition~\ref{prop:eachConfigAvoidUnsortablePattern} and Remark~\ref{rem:TotalConfigurations}. 
The converse is a consequence of Theorem~\ref{thm:popable}.
\end{pf}

\subsection{Decomposition and stack sorting}

In this part we exhibit conditions for a permutation $\sigma$ to be $2$-stack sorted depending on its decomposition. 

\paragraph{$\ominus$-decomposable permutations}: 
\begin{prop}\label{prop:2stacksMoinsDecomposable}
A permutation $\sigma = \ominus[\pi^{(1)},\pi^{(2)},\ldots, \pi^{(k)}]$ is $2$-stack sortable if and only if every $\pi^{(i)}$ for $i \in \{ 1 \ldots k-1\}$ is $2$-stack pushall sortable and $\pi^{(k)}$ is $2$-stack sortable.
\end{prop}
\begin{proof}
Suppose that $\sigma$ is $2$-stack sortable. Let $w_{\sigma}$ be a valid stack word of $\sigma$.
For $i \in \{ 1 \ldots k\}$, consider the subword $w_{\pi^{(i)}}$ of $w_{\sigma}$ by taking letters corresponding to an element of $\pi^{(i)}$. 
This word is of size $3|\pi^{(i)}|$ and has equal number of occurrences of the letters $\rho, \lambda, \mu$. 
Moreover, it is a valid stack word for $\pi^{(i)}$ as the relative order of elements of $\pi^{(i)}$ under the action of $w_{\pi^{(i)}}$ will be the same as the action of $w_{\sigma}$ on $\sigma$. 
Furthermore, as the element $1$ in $\sigma$ belongs to the last block $\pi^{(k)}$, all elements of $\pi^{(i)}$ are pushed into the stacks before the first pop. 
Hence $\pi^{(i)}$ is $2$-stack pushall sortable. 
$2$-stack sortable permutations form a permutation class, so that $\pi^{(k)}$ must be $2$-stack sortable.

Conversely, if every $\pi^{(i)}$ for $i \in \{ 1 \ldots k-1\}$ is $2$-stack pushall sortable and $\pi^{(k)}$ is $2$-stack sortable, let $w_i$ ($1 \leq i \leq k-1$) be a pushall stack word for $\pi^{(i)}$ and $w_k$ be a stack word for $\pi^{(k)}$. 
Then each $w_i$ ($1 \leq i \leq k-1$) can be written as $w'_iw''_i$ where $w'_i$ contains no occurrence of $\mu$ and $w''_i$ no occurrence of $\rho$.
It is easy to check that the word $w'_1w'_2\ldots w'_{k-1} w_k w''_{k-1} w''_{k-2} \ldots w''_1$ is a valid stack word for $\sigma$, hence $\sigma$ is $2$-stack sortable.
\end{proof}

With a similar proof, we have the following result when restricting to \pushall permutations:

\begin{prop}\label{prop:pushallMoinsDecomposable}
A permutation $\sigma = \ominus[\pi^{(1)},\pi^{(2)},\ldots, \pi^{(k)}]$ is \pushall if and only if every $\pi^{(i)}$ for $i \in \{ 1 \ldots k\}$ is $2$-stack pushall sortable.
\end{prop}

\paragraph{$\oplus$-decomposable permutations}

The case where $\sigma$ is $\oplus$-decomposable is a bit different as each block of the decomposition can be popped out as soon as they are pushed into the stacks. 
So the only condition is given in the following proposition.

\begin{prop}
If $\sigma = \oplus[\pi^{(1)},\ldots, \pi^{(k)}]$ then $\sigma$ is $2$-stack sortable if and only if each $\pi^{(i)}$ is $2$-stack sortable.
\end{prop}

For \pushall permutations, $\oplus$-decomposable permutations are harder to handle. 
As no element can be popped out before all elements have been pushed, the element $1$ which belongs to the first block must remain in the stacks until every element is pushed. 
This induces several constraints which are proved in the following propositions. 
All these propositions aim at proving Theorem~\ref{thm:+pt} which fully characterizes $\oplus$-decomposable \pushall permutations.

\begin{thm}
\label{thm:+pt}
Let $\sigma$ be a $\oplus$-decomposable permutation. Then $\sigma$ is \pushall if and only if $\sigma$ avoids
\begin{eqnarray*}
B_+ &=& \{ 132465, 135246, 142536, 142635, 143625, 153624, 213546, 214365, 214635, 215364, \\
&&241365, 314265, 315246, 315426, 351426, 1354627, 1365724, 1436527, 1473526 ,1546273, \\
&&1573246, 1624357, 1627354, 1632547, 1632574, 1642573, 1657243, 2465137, 2631547, \\
&&2635147, 3541627, 4621357, 4652137, 5136427, 5162437, 21687435, 54613287 \}
\end{eqnarray*}
\end{thm}

The proof proceeds step by step in Propositions~\ref{CS} to \ref{prop:+[1,x,1]pt}.

\begin{prop}\label{CS}
Let $\sigma$ be a permutation such that either:
\begin{itemize}
\item $\sigma \in Av(132)$
\item $\sigma \in Av(213)$
\item $\sigma \in \oplus[Av(132),Av(213)]$
\item $\sigma \in \oplus[Av(213),Av(132)]$
\end{itemize}
Then $\sigma$ is \pushall.
\end{prop}
\begin{proof}
We show that we can put all elements of $\sigma$ in the stacks so that they avoid patterns of Theorem~\ref{thm:popable} (p.\pageref{thm:popable}).
In the first case, just push every element in stack $H$. 
For the second case, we know from Knuth \cite{Knuth73} that each permutation avoiding $231$ can be sort in increasing order with one stack. 
So each permutation avoiding $213$ can be sort in deacreasing order with one stack. 
Hence we can use stack $H$ to push all elements of $\sigma$ in decreasing order onto stack $V$. 
For the last two cases, we push the elements in corresponding stacks $H$ for $Av(132)$ and $V$ for $Av(213)$. 
In each case, the stack configuration respect conditions of Theorem~\ref{thm:popable}.
\end{proof}

Note that Proposition~\ref{CS} give a sufficient condition which is not necessary: the permutation $143652$ is \pushall but do not belong to one of the preceding cases. 
In this proposition, an important role is given to classes $Av(213)$ and $Av(132)$. 
These indeed are exactly the classes of permutations that can be pushall sorted with a stack configuration where all elements lie in one {\em single} stack ($V$ for $Av(213)$ and $H$ for $Av(132)$). 
Thus the only difficult case is whenever a permutation contains both pattern $132$ and $213$. 
This is characterized by the following proposition:

\begin{prop}\label{prop:132et213}
A permutation $\sigma$ contains both patterns $213$ and $132$ if and only if it contains one of the following patterns:
$1324, 2143, 2413, 3142, 465213$ and $546132$.

\begin{figure}[H]
\begin{center}
\begin{tikzpicture}[scale=.30]
\useasboundingbox (1.5,0.5) rectangle (7,7);
\permutation{1,3,2,4}
\end{tikzpicture}
\begin{tikzpicture}[scale=.30]
\useasboundingbox (1,0.5) rectangle (7,7);
\permutation{2,1,4,3}
\end{tikzpicture}
\begin{tikzpicture}[scale=.30]
\useasboundingbox (1,0.5) rectangle (7,7);
\permutation{2,4,1,3}
\end{tikzpicture}
\begin{tikzpicture}[scale=.30]
\useasboundingbox (1,0.5) rectangle (7,7);
\permutation{3,1,4,2}
\end{tikzpicture}
\begin{tikzpicture}[scale=.30]
\useasboundingbox (1,0.5) rectangle (9,7);
\permutation{4,6,5,2,1,3}
\end{tikzpicture}
\begin{tikzpicture}[scale=.30]
\useasboundingbox (1,0.5) rectangle (8,7);
\permutation{5,4,6,1,3,2}
\end{tikzpicture}
\caption{Minimal permutations containing patterns $132$ and $213$.}
\label{fig:132et213}
\end{center}
\end{figure}
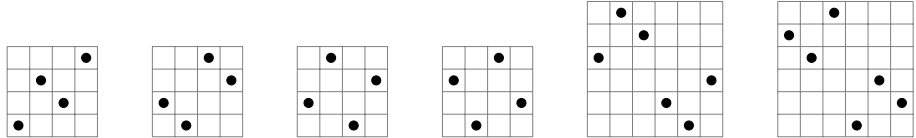
\end{prop}

\begin{proof}
Minimal permutations that contain both $132$ and $213$ are exactly permutations of the basis of $Av(132) \bigcup Av(213)$. 
By minimality of the elements of the basis those permutations are at most of size $6$ and a comprehensive study ends the proof.
\end{proof}

To prove a complete characterization of $\oplus$-decomposable \pushall permutations, we deal first with permutations whose decomposition contains non-trivial block -i.e. blocks not reduced to a singleton-.

\begin{prop}
\label{prop:B1}
Suppose $\sigma = \oplus[\alpha_1 \dots \alpha_r]$ with $r\geq 2$, each $\alpha_i$ $\oplus$-indecomposable and blocks $\alpha_1$ and $\alpha_r$ are non-trivial. 
Then $\sigma$ is \pushall if and only if $\sigma$ avoids every pattern of $B_1 = \{ 132465, 213546, 214365, 214635, 215364, \linebreak[1] 241365, 314265, 1657243, 4652137, \linebreak 21687435, 54613287 \}$.
\end{prop}
\begin{proof}

%\marginpar{Attention I ou J ne peut pas contenir a la fois 1 et r car sinon la permutation est non pushtriable}

We state by checking each pushall stack word of the right size that permutations of $B_1$ are not \pushall. 
Hence if $\sigma$ is \pushall it avoids $B_1$. 
Conversely, let $\sigma$ be a permutation avoiding every pattern of $B_1$. 
As $\alpha_1$ and $\alpha_r$ are non-trivial and $\oplus$-indecomposable, they contain $21$ as a pattern. 
But $\sigma$ avoids $214365$ so that blocks $\alpha_i$ with $2 \leq i \leq r-1$ are trivial. 
Let $I = \{i\ |\ \alpha_i$ contains pattern $132 \}$ and $J = \{j\ |\ \alpha_j$ contains pattern $213 \}$. 
These sets are included in $\{1, r\}$ and not equal to $\{1, r\}$ as $\sigma$ avoids $132465$ and $213546$.
\begin{itemize}
\item If $I = J = \O{}$, then $\alpha_1 \in Av(132)$ and $\oplus[\alpha_2 \dots \alpha_r] \in Av(213)$ so $\sigma \in \oplus[Av(132), Av(213)]$ and $\sigma$ is \pushall by Proposition~\ref{CS}.
\item If $I = \O{}$ and $J = \{ j_0\}$, then $j_0 \in \{1, r\}$.
If $j_0 = 1$ then $\alpha_1 \in Av(132)$ and $\oplus[\alpha_2 \dots \alpha_r] \in Av(213)$ hence $\sigma \in \oplus[Av(132), Av(213)]$ and $\sigma$ is \pushall by Proposition~\ref{CS}.
If $j_0 = r$, as $\sigma$ avoids $213546$ then $r = 2$, but $\alpha_1 \in Av(213)$ and $\alpha_r \in Av(132)$ hence $\sigma \in \oplus[Av(213), Av(132)]$. 
So $\sigma$ is \pushall by Proposition~\ref{CS}.

\item If $I = \{ i_0\}$ and $J = \O{}$, then $i_0 \in \{1, r\}$.
If $i_0 = 1$, as $\sigma$ avoids $132465$ then $r = 2$, but $\alpha_1 \in Av(213)$ and $\alpha_r \in Av(132)$ hence $\sigma \in \oplus[Av(213), Av(132)]$. 
So $\sigma$ is \pushall by Proposition~\ref{CS}.
If $i_0 = r$ then $\alpha_1 \in Av(132)$ and $\oplus[\alpha_2 \dots \alpha_r] \in Av(213)$ hence $\sigma \in \oplus[Av(132), Av(213)]$ and $\sigma$ is \pushall by Proposition~\ref{CS}.

\item If $I = \{ i_0\} \neq J = \{ j_0\}$.
If $i_0 = 1$ then $j_0 = r$ and $r = 2$ as $\sigma$ avoids $132465$. 
But $\alpha_1 \in Av(213)$ and $\alpha_r \in Av(132)$ hence $\sigma \in \oplus[Av(213), Av(132)]$ and $\sigma$ is \pushall by Proposition~\ref{CS}.
If $i_0 = r$ then $j_0 = 1$, $\alpha_1 \in Av(132)$ and $\oplus[\alpha_2 \dots \alpha_r] \in Av(213)$ hence $\sigma \in \oplus[Av(132), Av(213)]$. 
So $\sigma$ is \pushall by Proposition~\ref{CS}.

\item If $I = J = \{ i_0\}$, then by Proposition~\ref{prop:132et213}, $\alpha_{i_0}$ contains either $1324, 2143, 2413, 3142$, $465213$ or $546132$. 
We prove that $\sigma$ contains a pattern of $B_1$. 
If $\alpha_{i_0}$ contains $1324$, either $i_0 < r$, and $\sigma$ would contain $132465$ or $i_0 = r$, and $\sigma$ would contain $213546$. 
Similarly if $\alpha_{i_0}$ contains $2143$, $\sigma$ would contain $214365$. 
The same goes for $\alpha_{i_0}$ containing $2413$, $3142$, $465213$ or $546132$. 
Hence the case $I = J = \{ i_0\}$ cannot occur.
\end{itemize}
\end{proof}

Given two permutation classes ${\mathcal C}$ and ${\mathcal C'}$, their horizontal juxtaposition  $[{\mathcal C}\ {\mathcal C'}]$ consists of all permutations $\sigma$ that can be written as a concatenation $[\pi, \tau]$ 
where $\pi$ is order isomorphic to a permutation in ${\mathcal C}$ and $\tau$ is order-isomorphic to a permutation in ${\mathcal C'}$. 
In other words, a diagram of a permutation $\sigma \in [{\mathcal C}\ {\mathcal C'}]$ can be divided by  a vertical line into two parts, 
such that the left one is order-isomorphic to a permutation of ${\mathcal C}$ and the right one to a permutation of ${\mathcal C'}$. 
We can similarly define the vertical juxtaposition  $ \left[
\begin{array}{l}
{\mathcal C} \\
{\mathcal C'} \\
\end{array}
\right]$ consisting of permutations having a diagram cut by a horizontal line.

\begin{prop}
\label{prop:+[1,x]pt}
A permutation $\oplus[1,\sigma]$ is \pushall if and only if \\
$\sigma \in \big[Av(213)\ Av(132)\big]$ and there exists an associated decomposition $\sigma = [\pi ,\tau]$ 
such that there are no pattern $213$ in $\sigma$ where $2$ is in $\pi$ and $13$ is in $\tau$.
\end{prop}
%\marginpar{definir dans l intro la juxtaposition de classes}
\begin{proof}
If $\sigma = [\pi,\tau]$ with this decomposition satisfying hypothesis of the proposition, then $\oplus[1,\sigma]$ is \pushall using the following algorithm. 
Put $1$ in $H$. 
Then push elements of $\pi$ in stack $V$ in decreasing order. 
Then put $1$ at top of $V$ and finally push every element of $\tau$ onto $H$. 
As there are no pattern $213$ in $\sigma$ with $2$ in $\pi$ and $13$ in $\tau$, the stack configuration respects conditions of Theorem~\ref{thm:popable} hence can be popped out.

Conversely, suppose that $\oplus[1, \sigma]$ is \pushall and consider a stack word for this permutation. 
As $1$ is the first element, it is pushed at the bottom of $H$. 
Then some elements are pushed onto $1$ and into $V$ before $1$ is popped out from stack $H$ to stack $V$. 
The remaining elements are pushed into $H$ as they are greater than $1$. 
We consider the moment where all elements have been pushed and $1$ is at the top of $V$. 
This separates in two parts the elements of $\sigma$ taking $\tau$ as the elements in $H$ and $\pi$ the elements in $V$ apart from $1$. 
From Theorem~\ref{thm:popable} decomposition $\sigma = [\pi ,\tau]$ satisfies conditions of the statement.
\end{proof}

\begin{prop}
\label{prop:B2}
Let $E = \{ \sigma | \oplus[1,\sigma]$ is \pushall $\}$. 
Then $E$ is a finitely based permutation class whose basis is $B_2 = \{ 21354, 24135, 31425, 31524, 32514, 42513, 243516, 254613, \linebreak[1] 325416, \linebreak[1] 362415, 435162, 462135, 513246, 516243, 521436, 521463, 531462, 546132, 4652137 \}$.
\end{prop}

\begin{proof}
As \pushall permutations is a permutation class, so does $E$. 
Let $B_2$ be the basis of $E$. 
To prove that $B_2$ is finite, we first prove that every permutation in $B_2$ has size less than $9$. 
Then an comprehensive computation gives the permutations in $B_2$.

By Proposition~\ref{prop:+[1,x]pt}, $E = \{ \sigma = \pi \tau \mid \pi \in Av(213), \tau \in Av(132)$ and there are no pattern $213$ in $\sigma$ where $2$ is in $\pi$ and $13$ is in $\tau \}$. 
Let $\sigma \in B_2$. 
By definition $\sigma \not\in E$ so $\sigma \not\in Av(213)$ and $\sigma \not\in Av(132)$. 
Let $\sigma_i \sigma_j \sigma_k$ be a pattern $132$ such that $i$ is maximal and $\sigma_r \sigma_s \sigma_t$ be a pattern $213$ such that $t$ is minimal, then $r$ minimal (for $t$ fixed) and finally $s$ maximal (for $t$ and $r$ fixed).
\begin{itemize}
\item If $t < i$ then $\pi = \sigma_r \sigma_s \sigma_t \sigma_i \sigma_j \sigma_k \notin E$, hence by minimality of the basis $\sigma = \pi$ so $|\sigma| = 6$.
\item If $t = i$ then $\pi = \sigma_r \sigma_s \sigma_i \sigma_j \sigma_k \notin E$ and by minimality $\sigma = \pi$ so $|\sigma| = 5$.
\item If $t > i$, consider the pattern $\sigma_r \sigma_s \sigma_t$ (shown in Figure~\ref{fig1Preuve}). 
Minimality conditions for $t$ and $r$ and maximality condition for $s$ imply that gray zones in the diagram of $\sigma$ are empty. 
So $s = t-1$. 
As $\sigma \notin E$, there is no possible cut $\sigma = \pi \tau$ such that $\pi \in Av(213)$, $\tau \in Av(132)$ and there are no pattern $213$ in $\sigma$ where $2$ is in $\pi$ and $13$ is in $\tau$. 
Hence, all cuts in $\sigma$ are forbidden, either because they are to the left of a $132$ pattern or to the right of a $213$ pattern or between element $2$ and $1$ of a pattern $213$. 
More specially the cut between $t-1$ and $t$ is forbidden. 
This cut cannot be to the left of a pattern $132$ by maximality of $i$ ($t > i$) and cannot be to the right of a pattern $213$ by minimality of $t$. 
So this cut is between elements $2$ and $1$ of a pattern $213$. 
We consider a pattern $213$ denoted by $\sigma_x \sigma_y \sigma_z$ such that $x$ is minimal and $y$ is minimal for $x$ fixed among patterns $213$ such that $x \leq s = t-1$ and $y \geq t$.

\begin{figure}[H]
\begin{minipage}[b]{.22\linewidth}
\begin{center}
\begin{tikzpicture}[scale=.5]
\draw [fill, lightgray] (2,0) rectangle (3,4);
\draw [fill, lightgray] (0,1) rectangle (1,3);
\draw [help lines] (0,0) grid (4,4);
\draw (2,1) [fill] circle (0.2);
\node at (2.5,0.6) {{\scriptsize$\sigma_s$}};
\draw (3,3) [fill] circle (0.2);
\node at (3.5,3.4) {{\scriptsize$\sigma_t$}};
\draw (1,2) [fill] circle (0.2);
\node at (0.6,2.4) {{\scriptsize$\sigma_r$}};
\end{tikzpicture}
\caption{$\sigma_r \sigma_s \sigma_t$}
\label{fig1Preuve}
\end{center}
\end{minipage}
\begin{minipage}[b]{.23\linewidth}
\begin{center}
\begin{tikzpicture}[scale=.5]
\draw [fill, lightgray] (2,0) rectangle (3,4);
\draw [fill, lightgray] (0,1) rectangle (1,3);
\draw [fill, lightgray] (1,3) rectangle (2,4);
\draw [fill, lightgray] (3,1) rectangle (4,3);
\draw [help lines] (0,0) grid (4,4);
\draw (2,1) [fill] circle (0.2);
\node at (2.5,0.6) {{\scriptsize$\sigma_s$}};
\draw (3,3) [fill] circle (0.2);
\node at (3.5,2.6) {{\scriptsize$\sigma_t$}};
\draw (1,2) [fill] circle (0.2);
\node at (0.6,2.4) {{\scriptsize$\sigma_r$}};
\node at (3.5,3.5) {{\small$A$}};
\node at (3.5,0.5) {{\small$B$}};
\node at (0.5,3.5) {{\small$C$}};
\node at (0.5,0.5) {{\small$D$}};
\end{tikzpicture}
\caption{Cas $r > i$}
 \label{fig2Preuve}
\end{center}
\end{minipage}
\begin{minipage}[b]{.25\linewidth}
\begin{center}
\begin{tikzpicture}[scale=.5]
\draw [fill, lightgray] (2,3) rectangle (4,6);
\draw [fill, lightgray] (0,1) rectangle (1,5);
\draw [fill, lightgray] (1,1) rectangle (2,3);
\draw [fill, lightgray] (3,1) rectangle (6,3);
\draw [fill, lightgray] (3,0) rectangle (4,1);
\draw [help lines] (0,0) grid (6,6);
\draw (3,1) [fill] circle (0.2);
\node at (3.5,0.6) {{\scriptsize$\sigma_s$}};
\draw (4,3) [fill] circle (0.2);
\node at (4.5,3.4) {{\scriptsize$\sigma_t$}};
\draw (2,2) [fill] circle (0.2);
\node at (1.6,2.4) {{\scriptsize$\sigma_r$}};
\draw (1,4) [fill] circle (0.2);
\node at (0.6,4.4) {{\scriptsize$\sigma_x$}};
\draw (5,5) [fill] circle (0.2);
\node at (5.5,5.4) {{\scriptsize$\sigma_z$}};
\node at (0.5,5.5) {{\small$\gamma$}};
\node at (0.5,0.5) {{\small$\delta$}};
\end{tikzpicture}
\caption{$\sigma_x \sigma_r \sigma_s \sigma_t \sigma_z$}
 \label{fig3Preuve}
\end{center}
\end{minipage}
\begin{minipage}[b]{.25\linewidth}
\begin{center}
\begin{tikzpicture}[scale=.5]
\draw [fill, lightgray] (2,0) rectangle (4,2);
\draw [fill, lightgray] (0,1) rectangle (1,5);
\draw [fill, lightgray] (1,4) rectangle (3,6);
\draw [fill, lightgray] (2,2) rectangle (6,4);
\draw [help lines] (0,0) grid (6,6);
\draw (2,2) [fill] circle (0.2);
\node at (2.5,1.6) {{\scriptsize$\sigma_s$}};
\draw (3,4) [fill] circle (0.2);
\node at (3.5,4.4) {{\scriptsize$\sigma_t$}};
\draw (1,3) [fill] circle (0.2);
\node at (0.6,3.4) {{\scriptsize$\sigma_r$}};
\draw (4,1) [fill] circle (0.2);
\node at (4.5,0.6) {{\scriptsize$\sigma_y$}};
\draw (5,5) [fill] circle (0.2);
\node at (5.5,5.4) {{\scriptsize$\sigma_z$}};
\node at (0.5,5.5) {{\small$\gamma$}};
\node at (0.5,0.5) {{\small$\delta$}};
\end{tikzpicture}
\caption{$\sigma_r \sigma_s \sigma_t \sigma_y \sigma_z$}
 \label{fig4Preuve}
\end{center}
\end{minipage}
\end{figure}

\begin{itemize}
\item If $r \leq i$ then $\pi = \{\sigma_r \sigma_s \sigma_t \sigma_i \sigma_j \sigma_k \sigma_x \sigma_y \sigma_z\} \notin E$. 
Indeed all cuts are forbidden: 
those before $r$ by $\sigma_i \sigma_j \sigma_k$, between $r$ and $s$ by $\sigma_r \sigma_s \sigma_t$, between $s$ and $t$ by $\sigma_x \sigma_y \sigma_z$ and before $t$ by $\sigma_r \sigma_s \sigma_t$. 
So by minimality of the basis $|\sigma| \leq 9$.
\item If $r > i$, we want to prove that $x \leq i$. 
Then $\pi = \{\sigma_r \sigma_s \sigma_t \sigma_i \sigma_j \sigma_k \sigma_x \sigma_y \sigma_z\} \notin E$ since all cuts are forbidden as before and $|\sigma| \leq 9$.
As $r > i$ and $i$ maximal, gray zones added in Figure~\ref{fig2Preuve} are empty. 
As $y \geq t$, $\sigma_y$ and $\sigma_z$ lie either both in $A$, or both in $B$, or $\sigma_y$ lies in $B$ and $\sigma_z$ in $A$.
\begin{itemize}
\item If $\sigma_y$ and $\sigma_z$ lie both in $B$, then $\sigma_x$ lies in $D$ and $\sigma_x \sigma_t \sigma_z$ form the permutation $132$ and as $i$ is maximal, $x \leq i$.
\item If $\sigma_y$ and $\sigma_z$ lie both in $A$, then $\sigma_x$ lies in $C$ and by minimality of $y$ we have $y = t$. 
$x$ is minimal, so that gray zones added in Figure~\ref{fig3Preuve} are empty. 
Suppose that $x > i$. 
The cut between $i$ and $i + 1$ is forbidden as $\sigma \notin E$. 
As $i$ is maximal the cut cannot be to the left of a pattern $132$, neither to the right of a pattern $213$ by minimality of $t$. 
Hence the cut lies between element $2$ and $1$ of a pattern $213$. 
Let $\sigma_a \sigma_b \sigma_c$ be such a pattern $213$ such that $a \leq i$ and $b > i$. 
Then $a < x$ and $\sigma_a$ lies in area $\gamma$ or $\delta$ and $c \geq t$ by minimality of $t$. 
If $\sigma_a$ lies in $\gamma$ then $\sigma_a \sigma_t \sigma_c$ is the pattern $213$, which is forbidden by minimality of $x$. 
Hence $\sigma_a$ lies in $\delta$ and $b \geq t$ otherwise $\sigma_a \sigma_b \sigma_t$ is a pattern $213$ with $a \leq i < r$, which is also forbidden by minimality of $r$. 
Hence $\sigma_a \sigma_b \sigma_c$ is a pattern $213$ with $a \leq i < x \leq s$ and $b \geq t$ which is impossible by minimality of $x$.
\item If $\sigma_y$ lies in $B$ and $\sigma_z$ in $A$, by minimality of $x$, $x = r$ or $\sigma_x$ lies in $C$ or $\sigma_x$ lies in $D$. 
If $\sigma_x$ lies in $C$ then $\sigma_x \sigma_t \sigma_z$ is a pattern $213$ which contradicts the minimality of $y$. 
If $\sigma_x$ lies in $D$, $\sigma_x \sigma_r \sigma_s$ is a pattern $132$ hence $x \leq i$. 
If $x = r$, by minimality of $x$ then $y$, gray zones added in Figure~\ref{fig4Preuve} are empty. 
The cut between $i$ and $i + 1$ is forbidden as $\sigma \notin E$. 
As before the cut lies between elements $2$ and $1$ of a pattern $213$. 
Let $\sigma_a \sigma_b \sigma_c$ such a pattern $213$ such that $a \leq i$ and $b > i$. 
Then $a < r$ and $\sigma_a$ lies in $\gamma$ or $\delta$ and $c \geq t$ by minimality of $t$. 
If $\sigma_a$ lies in $\gamma$ then $\sigma_a \sigma_t \sigma_c$ is a pattern $213$ and by minimality of $x$, $x \leq a \leq i$. 
If $\sigma_a$ lies in $\delta$ then $b \geq t$ otherwise $\sigma_b \sigma_t \sigma_y$ is a pattern $132$ with $b > i$, which is forbidden by maximality of $i$. 
But $\sigma_a \sigma_b \sigma_c$ is a pattern $213$ with $a \leq i$ and $b \geq t$, so by minimality of $x$, $x \leq i$.
\end{itemize}
\end{itemize}
\end{itemize}

\end{proof}

\begin{prop}
\label{prop:+[x,1]pt}
A permutation $\oplus[\sigma,1]$ is \pushall if and only if $\sigma \in \left[
\begin{array}{l}
Av(132) \\
Av(213) \\
\end{array}
\right]$ and there exists an associated decomposition $\sigma = \left[
\begin{array}{l}
\pi \\
\tau \\
\end{array}
\right] $ such that there is no pattern $132$ in $\sigma$ where element $3$ is in $\pi$ and elements $1$ and $2$ are in $\tau$.
\end{prop}

\begin{proof}
Let $n = |\sigma|+1$. 
Consider a pushall sorting of $\oplus[\sigma,1]$. 
This permutation has $n$ as last element, so that we consider the configuration of the stacks just after the insertion of $n$. 
By Theorem~\ref{thm:popable}, it must avoid the pattern \patternVH, so that all elements in $H$ -under $n$- are greater than those of $V$.
Hence we can write $\sigma = \left[
\begin{array}{l}
\pi \\
\tau \\
\end{array}
\right] $ where $\tau$ contains elements of $V$ and $\pi$ those in $H$ -except $n$-.
Then from Theorem~\ref{thm:popable} $\pi \in Av(132)$ and $\tau \in Av(213)$ and that there are no pattern $132$ in $\sigma$ where element $3$ is in $\pi$ and elements $1$ and $2$ are in $\tau$.

Conversely, suppose that there exists a decomposition $\sigma = \left[
\begin{array}{l}
\pi \\
\tau \\
\end{array}
\right] $ respecting the previous conditions then we have a pushall sorting of the permutation $\oplus[\sigma,1]$ using the following algorithm.
While the input is not empty, if stack $H$ is empty or if the top of $H$ belongs to $\pi$, we push the next element of the input onto $H$. 
If $\sigma_i$, the top of $H$ belongs to $\tau$, and if the next element of the input $\sigma_j$ belongs to $\tau$ and is greater than $\sigma_i$, we push $\sigma_j$ onto $H$, otherwise we pop $\sigma_i$ from $H$ and push it onto $V$. 
At each step we verify conditions of Theorem~\ref{thm:popable} so that all elements can be popped out in increasing order at the end.
\end{proof}

\begin{prop}
\label{prop:B3}
Let $F = \{ \sigma | \oplus[\sigma,1]$ is \pushall $\}$. 
$F$ is a finitely based permutation class whose basis is $B_3 = \{13524, 14253, 21354, 31524, 31542, 35142, 135462, 143652, \\162435, 163254, 246513, 263154, 263514, 354162, 462135, 465213, 513642, 516243, 1657243 \}$.
\end{prop}
\begin{proof}
As the set of \pushall permutations is a permutation class, so is $F$. 
By Proposition~\ref{prop:+[x,1]pt}, $F = \{ \sigma \in \left[
\begin{array}{l}
Av(132) \\
Av(213) \\
\end{array}
\right]$ such that there exists an associated decomposition $\sigma = \left[
\begin{array}{l}
\pi \\
\tau \\
\end{array}
\right] $ such that there is no pattern $132$ in $\sigma$ where element $3$ is in $\pi$ and elements $1$ and $2$ are in $\tau$\}.
Hence $E$ and $F$ are in one-to-one correspondence by taking an element of $E$, rotate its diagram by $-\pi/2$ and apply the symmetry with respect to axis $(Oy)$. 
If elements are in one-to-one correspondence by rotation and symmetry so does the basis which proves the result.
\end{proof}

\begin{prop}
\label{prop:+[1,x,1]pt}
A permutation $\oplus[1,\sigma,1]$ is \pushall if and only if $\sigma \in \oplus[Av(213), Av(132)]$.
\end{prop}
\begin{proof}
By Proposition~\ref{prop:+[1,x]pt}, $\oplus[1,\sigma,1]$ is \pushall if and only if $\oplus[\sigma,1] \in \big[Av(213), Av(132)\big]$ and there exists a corresponding decomposition $\sigma = \pi \tau$ such that there is no pattern $213$ in $\sigma$ where element $2$ is in $\pi$ and $13$ are in $\tau$, which is equivalent to $\sigma \in \big[Av(213), Av(132)\big]$ and there exists a corresponding decomposition $\sigma = \pi \tau$ such that there are no pattern $21$ in $\sigma$ where element $2$ is in $\pi$ and element $1$ is in $\tau$, i.e. $\sigma \in \oplus[Av(213), Av(132)]$.
\end{proof}

We are now able to prove Theorem~\ref{thm:+pt} (p.\pageref{thm:+pt}).

\begin{proof}
Permutations of $B_+$ are not \pushall (check each pushall stack word of the right size), hence if $\sigma$ is \pushall it avoids $B_+$. 
Conversely suppose that $\sigma$ avoids $B_+$. 
Let $\sigma = \oplus[\alpha_1 \dots \alpha_r]$ be the $\oplus$-decomposition of $\sigma$ with $r \geq 2$ and $\alpha_i$ $\oplus$-indecomposable for all $i$.
\begin{itemize}
\item If $\alpha_1$ and $\alpha_r$ are non trivial then $\sigma$ is \pushall thanks to Proposition~\ref{prop:B1}. 
Indeed $\sigma$ avoids $B_1 = \{ 132465, 213546, 214365, 214635, 215364, 241365, 314265, 1657243, \linebreak 4652137, 21687435, 54613287 \}$ as $B_1 \subset B_+$.
\item If $\alpha_1$ is trivial then $\sigma = \oplus[1,\pi]$ and $\pi$ avoids $B_2 = \{ 21354, 24135, 31425, 31524, 32514, \\ 42513, 243516, 254613, 325416, 362415, 435162, 462135, 513246, 516243, 521436, 521463,\linebreak 531462, 546132, 4652137 \}$ so that $\sigma$ is \pushall by Proposition~\ref{prop:B2}.
\item If $\alpha_r$ is trivial then $\sigma = \oplus[\pi, 1]$ and $\pi$ avoids $B_3 = \{13524, 14253, 21354, 31524, 31542, \\ 35142, 135462, 143652, 162435, 163254, 246513, 263154, 263514, 354162, 462135, 465213,\linebreak 513642, 516243, 1657243 \}$ hence $\sigma$ is \pushall by Proposition~\ref{prop:B3}.
\end{itemize}
\end{proof}

%\marginpar{Il traine des B+}

We call {\em separable} permutations the class $Av(2413,3142)$.

\begin{thm}
Let $\sigma$ be a separable permutation. 
$\sigma$ is \pushall if and only if $\sigma$ avoids $B = \{ 132465, 213546, 214365, 1354627, 1436527, 1624357, 1632547, 1657243, 4652137,\\ 21687435, 54613287 \}$.
\end{thm}
%\marginpar{definir separable}
\begin{proof}
As permutations of $B$ are not \pushall, every \pushall permutation avoids $B$. 
Conversely, supppose that $\sigma$ avoids $B$. 
As $\sigma$ is separable, $\sigma$ is either $\oplus$-decomposable or $\ominus$-decomposable or trivial (i.e. of size 1), and $\sigma$ avoids $2413$ and $3142$ which added to constraints of $B$ gives that $\sigma$ avoids $B_+$, the set defined in Theorem~\ref{thm:+pt}.
If $\sigma$ is $\oplus$-decomposable, then $\sigma$ is \pushall by Theorem~\ref{thm:+pt}.
If $\sigma$ is $\ominus$-decomposable, then $\sigma = \ominus[\pi^{(1)},\pi^{(2)},\ldots, \pi^{(k)}]$ where each $\pi^{(i)}$ is either trivial or $\oplus$-decomposable. 
So $\sigma$ is \pushall by Proposition~\ref{prop:pushallMoinsDecomposable} and Theorem~\ref{thm:+pt}.
\end{proof}
%\begin{thm}\label{thm:2stacks2stackspushall}
%If $\sigma$ is a permutation and $\pi = \ominus[\pi',1]$ a pattern of $\sigma$S Then the two following properties hold:
%\begin{enumerate}
%\item $\sigma$ is $2$-stack sortable implies that $\pi$ is $2$-stack sortable.
%\item $\pi$ is $2$-stack sortable if and only if $\pi'$ is $2$-stack pushall sortable.
%\end{enumerate}
%\end{thm}
%
%\begin{pf}
%Knuth noticed that $2$-stack sortable permutations are a permutation class that is, every pattern of a $2$-stack sortable permutation is also sortable.
%Thus the first proposition holds.
%For the second item, if $\pi$ is $2$-stack sortable then consider any valid stack word for this permutation. 
%Consider the configuration just after the last occurence of $\rho$ in this word. 
%At this time all elements lie in the stacks as the first element that must be popped in the output is $1$ and this element lie in $H$.
%\end{pf}

%\begin{prop}
%If $\sigma$ is a permutation then $\sigma$ is $2$-stack pushall sortable if and only if $\sigma[1,\ldots,1,Id_k,1,\ldots,1]$ is $2$-stack pushall sortable.
%\end{prop}\label{prop:identityBlocks}
%\marginpar{Faut il mettre ce theoreme?}
%\begin{pf}
%TODO
%\end{pf}

\subsection{Basis of stack sorting class}
In the previous section, we show that \pushall separable permutations form a finitely based permutation class. 
This property does not hold for \pushall permutations and we exhibit an infinite antichain in the following proposition.

\begin{prop}
The basis of \pushall permutation is infinite.
\end{prop}

\begin{proof}
Consider permutations $2n-3\ 2n-1\ 2n-5\ 2n \dots p\ p+5 \dots 1\ 6\ 2\ 4$ for $n \geq 3$.
The first ones are depicted in Figure~\ref{fig:antichaine}. These permutations are simple and incomparable.
To complete the proof, straightforward though technical, just check that those permutations are not \pushall and that every pattern 
of these permutations are \pushall.

\begin{figure}[ht]
\begin{center}
\begin{tikzpicture}[scale=.2]
\useasboundingbox (1.5,0.5) rectangle (8,10);
\permutation{3,5,1,6,2,4}
\end{tikzpicture}
\begin{tikzpicture}[scale=.2]
\useasboundingbox (1,0.5) rectangle (10,10);
\permutation{5,7,3,8,1,6,2,4}
\end{tikzpicture}
\begin{tikzpicture}[scale=.2]
\useasboundingbox (1,0.5) rectangle (12,10);
\permutation{7,9,5,10,3,8,1,6,2,4}
\end{tikzpicture}
\begin{tikzpicture}[scale=.2]
\useasboundingbox (1,0.5) rectangle (14,10);
\permutation{9,11,7,12,5,10,3,8,1,6,2,4}
\end{tikzpicture}
\begin{tikzpicture}[scale=.2]
\useasboundingbox (1,0.5) rectangle (16,14);
\permutation{11,13,9,14,7,12,5,10,3,8,1,6,2,4}
\end{tikzpicture}
\caption{An antichain of the basis of \pushall permutations class.}
\label{fig:antichaine}
\end{center}
\end{figure}

\end{proof}

Note that the basis is infinite and contains a infinite number of simple permutations, and the \pushall class contains also an infinite number of simple permutations.

\begin{prop}\label{prop:basePushallTriable}
If $\sigma$ is in the basis of \pushall permutations, then $\sigma$ is $2$-stack sortable.
\end{prop}
\begin{proof}
Let $\sigma = \sigma_1\sigma_2\ldots\sigma_n$ be in the basis of \pushall permutations. 
By definition, $\sigma_1\sigma_2\ldots\sigma_{n-1}$ is \pushall. 
%$\sigma_n \not= 1$ otherwise $\sigma$ would be \pushall. 
We can sort $\sigma$ (not pushall sort $\sigma$) using the following algorithm. 
Push all elements $\sigma_1$ to $\sigma_{n-1}$ in the stacks following the \pushall operations of $\sigma_1\ldots\sigma_{n-1}$. 
Then pop elements $1, 2, \ldots, \sigma_n-1$, then push $\sigma_n$ and pop it to the output and pop the remaining elements. 
It is easy to check that these operations are allowed.
\end{proof}

Those last two propositions give a partial characterization of the basis of \pushall permutations class and $2$-stack sortable permutations. 
A more accurate result can be given for certain type of permutations in the basis.

\begin{prop}
Let $\pi$ be a $\ominus$-decomposable permutation. 
Then $\pi$ belongs to the basis of $2$-stack sortable permutations class if and only if $\pi = \ominus[\sigma,1]$ where $\sigma$ belongs to the basis of \pushall permutations class.
\end{prop}

\begin{proof}
Let $\pi = \ominus[\sigma,1]$ with $\sigma$ a pemrutation of the basis of \pushall permutations class. 
Proposition~\ref{prop:2stacksMoinsDecomposable} ensures that $\pi$ is not $2$-stack sortable. 
Note also that every pattern of $\pi$ is $2$-stack sortable. 
To prove this result, suppose that you remove a point in the permutation. 
Suppose we delete element $1$ then the obtained permutation is $\sigma$, hence it is $2$-stack sortable by Proposition~\ref{prop:basePushallTriable}. 
Otherwise we delete an element of $\sigma$ leading to $\sigma'$ which is \pushall by the definition of a permutation class basis. 
Then, $\ominus[\sigma',1]$ is $2$-stack pushall sortable using Proposition~\ref{prop:2stacksMoinsDecomposable}. 

Conversely, if $\sigma=\ominus[\pi^{(1)}, \pi^{(2)}, \ldots, \pi^{(k)}]$ belongs to the basis of $2$-stack sortable permutations class, then by Proposition~\ref{prop:2stacksMoinsDecomposable}, either $\pi^{(k)}$ is not $2$-stack sortable which contradicts the minimality of $\sigma$ ($\sigma$ is an element of the basis so that every pattern of $\sigma$ must belong to the class) or there exists $1 \leq i \leq k-1$ such that $\pi^{(i)}$ is not \pushall. 
But in that case, $\ominus[\pi^{(i)},1]$ is not $2$-stack sortable by Proposition~\ref{prop:2stacksMoinsDecomposable} hence $\sigma = \ominus[\pi^{(i)},1]$ by minimality of basis elements. 
If $\pi^{(i)}$ has a proper pattern $\tau$ which is not \pushall then $\ominus[\tau,1]$ is a proper pattern of $\sigma$ which is not $2$-stack sortable. 
This is impossible as $\sigma$ belongs to the basis of $2$-stack sortable permutations class. 
So $\pi^{(i)}$ belongs to the basis of \pushall permutations class, which concludes the proof.
\end{proof}

\section{Sorting and bi-coloring}\label{sec:coloring}

\subsection{A simple characterization}
There is a natural relation between $2$-stack pushall sorting and coloring of permutation diagram into two colors. 
The key idea is to look at the stack configuration once all elements of the permutation are pushed into the stacks. 
Then each element of the permutation belong either to stack $H$ or to stack $V$.
We assign a color to them depending in which stack they lie at this particular step of the sorting. 
In this article we color like \Hzone points that lie in stack $H$ and like \Vzone points in stack $V$.

However by Remark \ref{rem:decompoNonUnique}, this stack configuration is not unique, and neither is the coloring.

\begin{defn}\label{def:validColoring}
A {\em bicoloring} of a permutation $\sigma$ is a coloring of the points of the diagram of $\sigma$ with two colors \G and \R.

A {\em valid coloring} is a bicoloring which avoids each of the four following colored pattern:
\begin{itemize}
\item pattern \RRR: there is a pattern $132$ in \R
\item pattern \GGG: there is a pattern $213$ in \G
\item pattern \GGR: there is a point of \R lying vertically between a pattern $12$ of \G
\item pattern \RRG: there is a point of \G lying horizontally between a pattern $12$ of \R
\end{itemize}
\end{defn}

\begin{defn}
Let $\sigma$ be a permutation. 
To each total stack configuration of $\sigma$ the map $Col$ assigns the bicoloring of $\sigma$ such that elements of $H$ are in \R and elements of $V$ are in \G.
To every bicoloring of a permutation $\sigma$ the map $Conf$ associates the total stack configuration of $\sigma$ such that elements of \G lie in $V$ in decreasing order of value from bottom to top and elements of \R lie in $H$ in increasing order of indices from bottom to top.
\end{defn}

\begin{rem}\label{rem:Col/Conf}
For any bicoloring $b$, $Col(Conf(b))=b$.
For any stack configuration $c$ such that elements of $V$ are in decreasing order of value from bottom to top and elements of $H$ are in increasing order of indices from bottom to top, $Conf(Col(c))=c$.
\end{rem}

\begin{prop}\label{prop:AlgoPush}
Let $b$ be a bicoloring of a permutation $\sigma$.
Then Algorithm~\ref{algo:Push} applied to $b$ returns $true$ \ssi $Conf(b)$ is reachable for $\sigma$.
In this case the stack configuration to which Algorithm~\ref{algo:Push} leads is $Conf(b)$.
\end{prop}

\begin{algorithm}
 \SetAlgoLined
\LinesNumbered
  \KwData{$\sigma$ a permutation and $b$ a bicoloring of $\sigma$.}
  \KwResult{True if the stack configuration corresponding to $b$ is reachable from $\sigma$.}
Begin with the empty stack configuration and $\sigma$ as input and $i=1$\;
%$i \longleftarrow 1$\;
\While{$i \leq |\sigma|$}{
  \eIf{$H$ is empty or $top(H) \in R$}{
      push $\sigma_i$ into $H$\;
      $i \longleftarrow i+1$\;
  }( {\it /* $top(H) \in G$ */})
  {
  \eIf{$\sigma_{i} \in R$ or $\sigma_{i} < top(H)$}{
      \eIf{$V$ is empty or $top(H) < top(V)$}{
           pop $top(H)$ from stack $H$ and push it into $V$\;
      }{Return false\;}
  }( {\it /* $top(H) \in G$,  $\sigma_i \in G$ and $\sigma_{i} > top(H)$*/})
  {
  push $\sigma_i$ into $H$\;
  $i \longleftarrow i+1$\;
  }
  }
}
\While{$H$ is nonempty and $top(H) \in G$}{
  \eIf{$top(H) < top(V)$}{
      pop $top(H)$ from stack $H$ and push it into $V$\;
  }{Return false\;}
}
Return true\;
\caption{Algorithm to obtain a reachable configuration compatible with a bicoloring}
\label{algo:Push}
\end{algorithm}

To state this proposition we need the two following lemmas:

\begin{lem}\label{lem:configAlgoPush}
At each step of Algorithm~\ref{algo:Push}, the stack configuration we have is reachable for $\sigma$, elements of $H$ are in increasing order of indices from bottom to top, elements of $V$ are in decreasing order of value from bottom to top, there is no element of \R in $V$, there is no element of \R above an element of \G in $H$ and elements of \G that lie in $H$ are in increasing order of value from bottom to top.

Moreover index $i$ verifies that if $i\leq |\sigma|$ then $\sigma_i$ is the next element of the input and if $i> |\sigma|$ then there is no more element in the input.
\end{lem}

\begin{pf}
The proof is by induction on the number of stack operations performed by the algorithm.
Algorithm~\ref{algo:Push} begins with the empty stack configuration and $\sigma$ as input and $i=1$ so the properties are true at the beginning. Algorithm~\ref{algo:Push} performs only appropriate stack operations so at each step the configuration obtained is reachable for $\sigma$.
Moreover in a reachable configuration, elements of $H$ are in increasing order of indices.
When an element is put in $V$ (this happens at line $9$ or $21$) then this element is in \G (checked at line $6$ or $19$) and is smaller than the top of $V$ (checked at line $8$ or $20$) so that elements of $V$ remain in decreasing order of value from bottom to top and $V$ contains no element of \R. 
When we put an element in $H$, it can be at line $4$ or $14$. 
In the first case, $H$ is empty or its top is in \R (checked at line $3$) so all its elements are in \R by induction hypothesis. 
In the second case, the top of $H$ is in \G and the element we put in $H$ is in \G and greater than the top of $H$. 
This ensures that there is no element of \R above an element of \G in $H$ and that elements of \G that lie in $H$ are in increasing order from bottom to top (using induction hypothesis). 
Finally $i$ is increased exactly when $\sigma_i$ is put into $H$ so the last property remains true.
\end{pf}

\begin{lem}\label{lem:AlgoPushTerminates}
Algorithm~\ref{algo:Push} terminates in linear time w.r.t $|\sigma|$.
\end{lem}

\begin{pf}
At each step, Algorithm~\ref{algo:Push} performs either a legal move $\rho$, or a legal move $\lambda$, or return false or true (and stops). 
As at most $|\sigma|$ legal moves $\rho$ and $|\sigma|$ legal moves $\lambda$ can be done, Algorithm~\ref{algo:Push} terminates after at most $2|\sigma|+1$ steps. 
As each step is done in constant time, we have the result.
\end{pf}

We are now able to prove Proposition~\ref{prop:AlgoPush}:

\begin{pf}
If Algorithm~\ref{algo:Push} applied to $b$ returns true, then it reaches line $26$. 
In particular the loop of line $19$ stops so the top of $H$ is not in \G. 
Thus by Lemma~\ref{lem:configAlgoPush} there is no element of \G in $H$. 
In addition by the same lemma elements of $H$ are in increasing order of indices from bottom to top, elements of $V$ are in decreasing order of value from bottom to top and there is no element of \R in $V$. 
So the stack configuration we have is $Conf(b)$.
Moreover Lemma~\ref{lem:configAlgoPush} states that the stack configuration we have is reachable for $\sigma$, so $Conf(b)$ is reachable for $\sigma$.

Conversely if $Conf(b)$ is reachable for $\sigma$, then there is a sequence $w$ of appropriate stack operations so that the configuration obtained with $\sigma$ as input is $Conf(b)$. 
Let us prove that the sequence of moves $w'$ performed by Algorithm~\ref{algo:Push} applied to $b$ is $w$. 
We prove by induction on $k \leq |w|$ ($k \geq 0$) that $w$ and $w'$ have the same prefix of length $k$ (obvious for $k=0$). 
First notice that as $Conf(b)$ is a total stack configuration, so $w$ has no letter $\mu$, and that Algorithm~\ref{algo:Push} performs only moves $\lambda$ and $\rho$, so $w'$ has no letter $\mu$. 
Suppose that $w$ and $w'$ have the same prefix $v$ of length $k$ with $k<|w|$, let $c'$ be the stack configuration obtained after permforming moves of $v$ with $\sigma$ as input. 
We want to prove that $w'_{k+1}$ exists and $w'_{k+1} = w_{k+1}$. 
By definition of $w'$, $w'_{k+1}$ is the move performed by Algorithm~\ref{algo:Push} in configuration $c'$ (setting by extension $w'_{k+1}=\alpha$ if Algorithm~\ref{algo:Push} terminates in configuration $c'$, i.e. if $|w'|=k$), and by defintion of $w$, $w_{k+1}$ is a move which allows to go from configuration $c'$ to configuration $Conf(b)$ (maybe with some additional moves).

We check the value of $i$ after Algorithm~\ref{algo:Push} has performed moves $v$. 
We know that at this step stacks are in configuration $c'$.

If $i> |\sigma|$, then from Lemma~\ref{lem:configAlgoPush} in configuration $c'$ all elements of $\sigma$ lie already in the stacks. 
As $w$ is a sequence of appropriate stack operations, then $w_{k+1} \neq \rho$ so $w_{k+1} = \lambda$ ($w$ has no letter $\mu$). 
As $w_{k+1}$ is a move which allows to go from configuration $c'$ to configuration $Conf(b)$ in which there is no elements of \R in $V$ and $V$ is decreasing, then the top of $H$ in $c'$ is in \G and smaller than the top of $V$ (or $V$ is empty). 
As $i> |\sigma|$ and as the top of $H$ in $c'$ is in \G and smaller than the top of $V$ (or $V$ is empty) then Algorithm~\ref{algo:Push} performs line $21$ so $w'_{k+1} = \lambda = w_{k+1}$.

If $i \leq |\sigma|$ then we are in the loop beginning at line~$2$ of the algorithm and from Lemma~\ref{lem:configAlgoPush} $\sigma_i$ is the next element of the input. 
Suppose that $w_{k+1} = \lambda$. 
As $w_{k+1}$ is a legal move which allows to go from configuration $c'$ to configuration $Conf(b)$ in which there is no elements of \R in $V$ and $V$ is decreasing, then $H$ is non empty, the top of $H$ is in \G and smaller than the top of $V$. 
Suppose in addition that $\sigma_i \in \G$. 
As $\sigma_i$ is still on the input after $w_{k+1}$ and $w_{k+1}$ is a move which allows to go to configuration $Conf(b)$ in which $V$ is decreasing, then $\sigma_i$ is smaller than the top of $H$ in $c'$. 
So either $\sigma_i < \text{top(H)}$ or $\sigma_i \in \R$. 
So from $c'$ Algorithm~\ref{algo:Push} performs line $9$ so $w'_{k+1} = \lambda = w_{k+1}$.

Suppose that $w_{k+1} = \rho$. 
If in configuration $c'$ stack $H$ is empty or top(H) $\in R$ then Algorithm~\ref{algo:Push} performs line $4$ so $w'_{k+1} = \rho = w_{k+1}$. 
Otherwise let $\sigma_h$ be the top of $H$ in $c'$, then $\sigma_h \in \G$. 
So $\sigma_h \in V$ in $Conf(b)$. 
But once $w_{k+1} = \rho$ is performed $\sigma_i$ is above $\sigma_h$ in $H$. 
As $w_{k+1}$ is a move which allows to go from configuration $c'$ to configuration $Conf(b)$ then $\sigma_i$ is below $\sigma_h$ in $V$ in $Conf(b)$ (indeed it is impossible that $\sigma_h$ goes to stack $V$ and $\sigma_i$ remains in stack $H$). 
So $\sigma_i \in \G$ and as in $Conf(b)$ elements of $V$ are in decreasing order, $\sigma_i > \sigma_h$. 
So the test of line $7$ of the algorithm is false and Algorithm~\ref{algo:Push} performs line $14$ so $w'_{k+1} = \rho = w_{k+1}$.

This ends the induction. 
We have proved that $w$ is a prefix of $w'$, so Algorithm~\ref{algo:Push} reaches configuration $Conf(b)$. 
We have now to prove that Algorithm~\ref{algo:Push} stops in this configuration and returns true.

When $Conf(b)$ is reached then there is no element in the input anymore, so from Lemma~\ref{lem:configAlgoPush} $i> |\sigma|$, and $top(H) \notin \G$ in $Conf(b)$. 
So both loops while of Algorithm~\ref{algo:Push} are finished and the algorithm reaches line $27$, returns true and terminates in configuration $Conf(b)$.
\end{pf}

\begin{lem}\label{lem:AlgoFalse}
Let $b$ be a bicoloring of a permutation $\sigma$.
If Algorithm~\ref{algo:Push} applied to $b$ returns $false$ then $b$ has a pattern \GGR or a pattern \GGG.
\end{lem}

\begin{pf}
We consider the stack configuration reached when Algorithm~\ref{algo:Push} returns $false$. 
We set $\sigma_h = top(H)$ and $\sigma_v = top(V)$. 
By Lemma~\ref{lem:configAlgoPush}, $\sigma_v \in \G$. 
Algorithm~\ref{algo:Push} returns $false$ by reaching either line $11$ or line $23$. 
In both cases, $\sigma_h \in \G$ and $\sigma_h > \sigma_v$. 
Now we consider the step of the algorithm where $\sigma_v$ was put in $V$, the indice $i$ at this step of the algorithm, and the corresponding configuration $c$ just before the move putting $\sigma_v$ into $V$ is done. 
So at this step $\sigma_v$ is on the top of $H$, and $i > v$. 
If $\sigma_h$ is in $H$ in $c$, then it is below $\sigma_v$, contradicting Lemma~\ref{lem:configAlgoPush} ($\sigma_h > \sigma_v$ and both are in \G). 
As $\sigma_h$ is in $H$ when the algorithm ends, it cannot be in $V$ in $c$. 
So $\sigma_h$ is still in the input and $i \leq h \leq |\sigma|$. 
Recall that we consider the step of the algorithm where $\sigma_v$ is put in $V$. 
This can happen at line $9$ or $21$ but $i\leq |\sigma|$ so it is at line $9$. 
So the test of line $8$ is true, thus either $\sigma_i \in \R$ and then $\sigma_v, \sigma_i, \sigma_h$ is a pattern \GGR of $b$, or $\sigma_i \in \G$ but $\sigma_i < \sigma_v$ and then $\sigma_v, \sigma_i, \sigma_h$ is a pattern \GGG of $b$.
\end{pf}

% Rq: Il n'est pas nécessaire de mettre le test de la ligne 20 à l'intérieur de la 2e boucle, il suffit de le faire une fois juste avant la 2e boucle.

\begin{thm}\label{thm:bijectionColoriageConfiguration}
The map $Col$ is a bijection from the set of reachable total stack configuration of $\sigma$ avoiding the three unsortable patterns to the set of valid coloring of $\sigma$.
Moreover the inverse of $Col$ is the map $Conf$.
\end{thm}

\begin{pf}
Let $c$ be a reachable total stack configuration of $\sigma$ avoiding the three unsortable patterns and set $b=Col(c)$. 
We have to prove that $b$ avoids every forbidden colored pattern of Definition~\ref{def:validColoring}.

If $b$ has a pattern $132$ in \R then there are three element $\sigma_i$, $\sigma_j$ and $\sigma_k$ of \R such that $i<j<k$ and $\sigma_i < \sigma_k < \sigma_j$. 
By definition of $Col$, $\sigma_i$, $\sigma_j$ and $\sigma_k$ lie in $H$. 
As $c$ is reachable and $i<j<k$, $\sigma_i$ is below $\sigma_j$ which is below $\sigma_k$. 
So we have a stack-pattern \patternH in $c$ which contradicts our hypothesis. 
So $b$ has no pattern \RRR.

If $b$ has a pattern $213$ in \G then there are three element $\sigma_i$, $\sigma_j$ and $\sigma_k$ of \G such that $i<j<k$ and $\sigma_j < \sigma_i < \sigma_k$. 
By definition of $Col$, $\sigma_i$, $\sigma_j$ and $\sigma_k$ lie in $V$. 
As $c$ avoids stack-pattern \patternV, $\sigma_k$ is below $\sigma_i$ which is below $\sigma_j$. 
But then $c$ is not reachable: as $\sigma_k$ is below $\sigma_i$ and $\sigma_j$ in $V$, $\sigma_i$ and $\sigma_j$ have to stay in stack $H$ until $\sigma_k$ enters stack $H$. 
But as $i<j$, $\sigma_i$ is below $\sigma_j$ in stack $H$ and cannot be below $\sigma_j$ in stack $V$ as going from stack $H$ to stack $V$ reverse the order. 
So $b$ has no pattern \GGG.
% utiliser plutôt le fait que 213 est le complement de 231 et donc que les points ne peuvent pas être triés en ordre décroissant dans V (?)

If $b$ has a point of \R lying vertically between a pattern $12$ of \G then there are elements $\sigma_i$ and $\sigma_j$ of \G and $\sigma_k$ of \R such that $i<k<j$ and $\sigma_i < \sigma_j$. 
By definition of $Col$, $\sigma_i$ and $\sigma_j$ lie in $V$ and $\sigma_k$ lies in $H$. 
Configuration $c$ is reachable. 
We consider a sequence of stack operations leading to $c$. 
As $i<k$, $\sigma_i$ is already in the stacks when $\sigma_k$ enters $H$. 
As $\sigma_k$ remains in $H$ in $c$ but $\sigma_i$ is in $V$ in $c$, $\sigma_i$ has to be already in $V$ when $\sigma_k$ enters stack $H$. 
As $k<j$, at this moment $\sigma_j$ is not already in stack $V$, so $\sigma_j$ will be above $\sigma_i$ in $V$ and they form a pattern \patternV in $c$, which is excluded. 
So $b$ has no pattern \GGR.

If $b$ has a point of \G lying horizontally between a pattern $12$ of \R then in $c$ these points form a pattern \patternVH which is excluded. 
So $b$ has no pattern \RRG.

Conversely let $b$ be a valid coloring of $\sigma$.
By definition $Conf(b)$ is a total stack configuration of $\sigma$.
We have to prove that $Conf(b)$ is reachable for $\sigma$ and avoids the three unsortable stack patterns.
As $b$ is a valid coloring, it avoids patterns \GGR and \GGG. 
So from Lemma~\ref{lem:AlgoFalse}, Algorithm~\ref{algo:Push} started with input $b$ returns true.
Thus from Proposition~\ref{prop:AlgoPush}, $c$ is reachable for $\sigma$. 
Moreover by definition of $Conf$, $Conf(b)$ avoids pattern \patternV. 
Furthermore we know that in $Conf(b)$, elements of $H$ are in increasing order of indices from bottom to top. 
So if $Conf(b)$ has a pattern \patternH, then $b$ has a pattern \RRR, and if $Conf(b)$ has a pattern \patternVH then $b$ has a pattern \RRG. 
A $b$ is a valid coloring, we conclude that $Conf(b)$ avoids the three unsortable stack patterns.

Now using Remark~\ref{rem:Col/Conf} it's clear that $Conf$ is the inverse of $Col$.
\end{pf}

\begin{thm}\label{thm:equivalenceColoringPushall}
A permutation $\sigma$ is \pushall \ssi its diagram admits a valid coloring.
\end{thm}

\begin{pf}
Consequence of Proposition~\ref{prop:pushallIffConfigurationEvitePatterns} and Theorem~\ref{thm:bijectionColoriageConfiguration}.
\end{pf}

Now thank to Theorem~\ref{thm:equivalenceColoringPushall} we have a naive algorithm to check if a permutation $\sigma$ is \pushall: forall bicoloring $b$ of $\sigma$, we can test if $b$ is valid by checking if $b$ avoids patterns \GGG, \GGR, \RRG and \RRR of Definition~\ref{def:validColoring}.
But first notice that we have a more efficient way to test if a bicoloring is valid:

\begin{prop}\label{prop:Check-Valid-linear}
Let $b$ be a bicoloring of a permutation $\sigma$.
We can check in linear time w.r.t. $|\sigma|$ if $b$ is a valid coloring.
More precisely, $b$ is a valid coloring \ssi Algorithm~\ref{algo:Push} applied to $b$ returns true and Algorithm~\ref{algo:popOut} applied to $Conf(b)$ returns true.
\end{prop}

\begin{pf}
From Theorem~\ref{thm:bijectionColoriageConfiguration}, $b$ is valid \ssi $Conf(b)$ is reachable for $\sigma$ and avoids the three unsortable patterns.
We conclude using Lemma~\ref{lem:AlgoPushTerminates}, Proposition~\ref{prop:AlgoPush}, Theorem~\ref{thm:popable} and Proposition~\ref{prop:AlgoPopOut}.
\end{pf}

Now even using this efficient way to test if a bicoloring is valid, the naive algorithm descrided above is unefficient. 
Indeed there is $2^{|\sigma|}$ bicolorings of $\sigma$, leading to a exponential algorithm. 
Yet we will find a way to restrict the possible number of colorings to a polynomial number. 
The key idea is to look at increasing sequences in the permutation.

\subsection{Increasing sequences in a valid coloring}

First we reformulate the notion of valid coloring thanks to increasing and decreasing sequences.

\begin{prop}\label{prop:rulesR8}
Let $c$ be a bicoloring of a permutation $\sigma$.
Then $c$ is a valid coloring if and only if $c$ respects the following set of rules denoted $\mathcal{R}_8$:
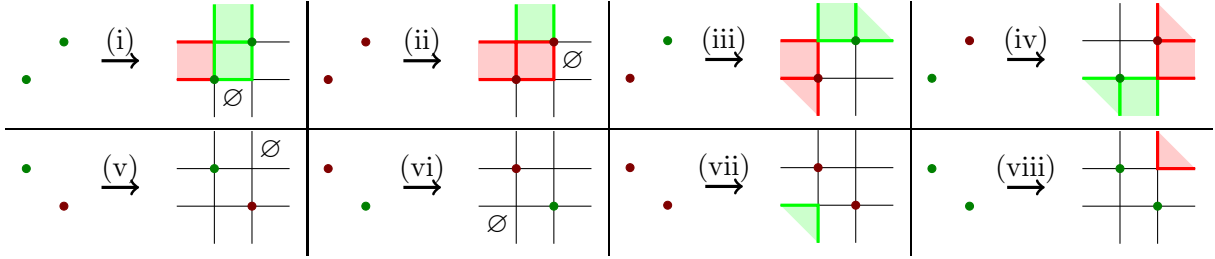
\begin{figure}[H]
\begin{tabular}{p{.23\textwidth}|p{.23\textwidth}|p{.23\textwidth}|p{.23\textwidth}}
\begin{tikzpicture}[scale=.5]
\Vpoint{0}{0};
\Vpoint{1}{1};
\draw [very thick,->] (2,0.5) -- (3,0.5);
\draw (4,0) -- (7,0);
\draw (4,1) -- (7,1);
\draw (5,-1) -- (5,2);
\draw (6,-1) -- (6,2);
\zoneD{H};
\zoneE{V};
\zoneH{V};
\draw (5.5,-0.5) node {$\varnothing$};
\Vpoint{5}{0};
\Vpoint{6}{1};
\etiquette{1};
\end{tikzpicture}
&
\begin{tikzpicture}[scale=.5]
\Hpoint{0}{0};
\Hpoint{1}{1};
\draw [very thick,->] (2,0.5) -- (3,0.5);
\draw (4,0) -- (7,0);
\draw (4,1) -- (7,1);
\draw (5,-1) -- (5,2);
\draw (6,-1) -- (6,2);
\zoneH{V};
\zoneD{H};
\zoneE{H};
\draw (6.5,0.5) node {$\varnothing$};
\Hpoint{5}{0};
\Hpoint{6}{1};
\etiquette{2};
\end{tikzpicture}
&
\begin{tikzpicture}[scale=.5]
\Hpoint{0}{0};
\Vpoint{1}{1};
\draw [very thick,->] (2,0.5) -- (3,0.5);
\draw (4,0) -- (7,0);
\draw (4,1) -- (7,1);
\draw (5,-1) -- (5,2);
\draw (6,-1) -- (6,2);
\zoneA{H};
\zoneD{H};
\zoneI{V};
\zoneH{V};
\Hpoint{5}{0};
\Vpoint{6}{1};
\etiquette{3};
\end{tikzpicture}
&
\begin{tikzpicture}[scale=.5]
\Vpoint{0}{0};
\Hpoint{1}{1};
\draw [very thick,->] (2,0.5) -- (3,0.5);
\draw (4,0) -- (7,0);
\draw (4,1) -- (7,1);
\draw (5,-1) -- (5,2);
\draw (6,-1) -- (6,2);
\zoneA{V};
\zoneB{V};
\zoneF{H};
\zoneI{H};
\Vpoint{5}{0};
\Hpoint{6}{1};
\etiquette{4};
\end{tikzpicture}
\\
\hline
\begin{tikzpicture}[scale=.5]
\Vpoint{0}{1};
\Hpoint{1}{0};
\draw [very thick,->] (2,0.5) -- (3,0.5);
\draw (4,0) -- (7,0);
\draw (4,1) -- (7,1);
\draw (5,-1) -- (5,2);
\draw (6,-1) -- (6,2);
\draw (6.5,1.5) node {$\varnothing$};
\Vpoint{5}{1};
\Hpoint{6}{0};
\etiquette{5};
\end{tikzpicture}
&
\begin{tikzpicture}[scale=.5]
\Hpoint{0}{1};
\Vpoint{1}{0};
\draw [very thick,->] (2,0.5) -- (3,0.5);
\draw (4,0) -- (7,0);
\draw (4,1) -- (7,1);
\draw (5,-1) -- (5,2);
\draw (6,-1) -- (6,2);
\draw (4.5,-0.5) node {$\varnothing$};
\Hpoint{5}{1};
\Vpoint{6}{0};
\etiquette{6};
\end{tikzpicture}
&
\begin{tikzpicture}[scale=.5]
\Hpoint{0}{1};
\Hpoint{1}{0};
\draw [very thick,->] (2,0.5) -- (3,0.5);
\draw (4,0) -- (7,0);
\draw (4,1) -- (7,1);
\draw (5,-1) -- (5,2);
\draw (6,-1) -- (6,2);
\zoneA{V};
\Hpoint{5}{1};
\Hpoint{6}{0};
\etiquette{7};
\end{tikzpicture}
&
\begin{tikzpicture}[scale=.5]
\Vpoint{0}{1};
\Vpoint{1}{0};
\draw [very thick,->] (2,0.5) -- (3,0.5);
\draw (4,0) -- (7,0);
\draw (4,1) -- (7,1);
\draw (5,-1) -- (5,2);
\draw (6,-1) -- (6,2);
\zoneI{H};
\Vpoint{5}{1};
\Vpoint{6}{0};
\etiquette{8};
\end{tikzpicture}
\end{tabular}
\begin{center}
\caption{Rewriting rules $\mathcal{R}_8$}\label{fig:rewrite}
\end{center}
\end{figure}
For example rule $(\rmnum{1})$ means that if two points $(i,\sigma_{i})$ and $(j,\sigma_{j})$ are in increasing order $i<j$ and $\sigma_{i}<\sigma_{j}$ and belong to $\G$ then every point $(k,\sigma_{k})$ of the permutation must respect:
\begin{itemize}
\item If $i<k<j$ then $\sigma_k > \sigma_{i}$ and $(k,\sigma_{k})$ belongs to $\G$.
\item If $k < i$ and $\sigma_{i} < \sigma_{k} < \sigma_{j}$ then $(k,\sigma_{k})$ belongs to $\R$.
\end{itemize}
\end{prop}

\begin{proof}
We prove that $c$ is not valid \ssi $c$ violates a rule of $\mathcal{R}_8$.
Suppose that $c$ is not valid then $c$ has one of the four colored patterns of Definition~\ref{def:validColoring}.
If $c$ has a pattern \RRR then $c$ violate rule $(\rmnum{2})$ applied to elements $1$ and $3$ of the pattern \RRR, as element $2$ of the pattern lies in a zone that should be empty.
If $c$ has a pattern \GGG then $c$ violate rule $(\rmnum{1})$ applied to elements $2$ and $3$ of the pattern \GGG, as element $1$ of the pattern lies in a zone that should be empty.
If $c$ has a pattern \GGR then $c$ violate rule $(\rmnum{1})$ applied to elements of \G of the pattern \GGR.
If $c$ has a pattern \RRG then $c$ violate rule $(\rmnum{2})$ applied to elements of \R of the pattern \RRG.
Conversely if $c$ violates a rule of $\mathcal{R}_8$ then a comprehensive study shows that $c$ has one of the four colored patterns of Definition~\ref{def:validColoring} and is not valid.
\end{proof}

We can use implication rules of $\mathcal{R}_8$ to limit the number of bicoloring to test, using the following idea: knowing the coloring of some points in the permutation (either in \R or in \G), the deduction rules given in Figure~\ref{fig:rewrite} can be applied until we obtain either a contradiction or no more rule can be applied. 
We can try the following algorithm: Set the color of two increasing points of $\sigma$, use implication rules to deduce the color of the other points and test whether the coloring obtained is right. 
Unfortunately, implication rules are not sufficient to ensure that given the color two points, the color of all other points is set. 
We may have to choose arbitrary the color of lots of points. 
To ensure that the number of bicoloring to test is polynomial, we have to study more precisely properties of increasing sequences in a valid bicoloring.

\begin{defn}
Let $c$ be a bicoloring of a permutation $\sigma$. 
We call {\em \ascentRG} a pair of two points $(\sigma_i,\sigma_j)$ such that $i<j$, $\sigma_{i} < \sigma_{j}$, $\sigma_{i}\in \R$ and $\sigma_{j} \in \G$.
We define in the same way \ascentsGR, RR or GG.
\end{defn}

%\marginpar{changer ascent}

Rule $(\rmnum{3})$ of $\mathcal{R}_8$ implies that every \ascentRG fixes the color of all points to the left of $\sigma_i$ below $\sigma_j$ (which are in \R) and to the right of $\sigma_i$ above $\sigma_j$ (which are in \G). 
The following theorem shows that when $\sigma$ is $\ominus$-indecomposable, the color of points to the left of $\sigma_i$ above $\sigma_j$ is also fixed.

% Recall that from Proposition~\ref{prop:pushallMoinsDecomposable} if $\sigma$ is $\ominus$-decomposable then $\sigma$ is $2$-stack pushall sortable if and only if each of the block is also $2$-stack pushall sortable.
% Thus, we can assume that $\sigma$ is $\ominus$-indecomposable.
% We show that given an \ascentRG -- see points $(\sigma_{i_{RG}},\sigma_{j_{RG}})$ in Figure~\ref{fig:cutGraphical} for example -- then all points to the left of $\sigma_{i_{RG}}$ or to the top of $\sigma_{j_{RG}}$ have a determined color.

\begin{thm}\label{thm:RGIncreasing}
Consider a valid coloring of a $\ominus$-decomposable permutation $\sigma$.
If there exist two points $\sigma_i < \sigma_j, i < j$ such that $\sigma_i \in \R$ and $\sigma_j \in \G$, 
then the color of every point $\sigma_k$ with $k < i$ or $\sigma_k > \sigma_j$ is determined by iterations of rules {$\mathcal C_8$} knowing only the color of $\sigma_i$ and $\sigma_j$ 
and can be represented as follows, the second diagram being a short representation of this alternance which will be used in the sequel. 
Furthermore, any increasing sequence $(i,\sigma_{i}),(j,\sigma_{j})$ of points located either to the left of $i$ or to the top of $\sigma_{j}$ is either monochromatic or colored $RG$.
Moreover, knowing $\sigma_{i}$ and $\sigma_{j}$, we can decide the color of the points whose indices are less than $i$ or whose values are greater than $\sigma_{j}$ in linear time.
\end{thm}

\begin{figure}[H]
\begin{center}
\begin{tikzpicture}[scale=.5]
\draw[Hfill] (-1,-1) rectangle (0,2);
\draw[Vfill] (0,2) rectangle (3,3);
\draw (-0.5,2.5) node {$\varnothing$};
\Hpoint{0}{0};
\Vpoint{2}{2};
\draw (0.5,0) node{$\sigma_{i}$};
\draw (2,1.5) node{$\sigma_{j}$};
\begin{scope}
\draw[Hfill] (-2,2) rectangle (-1,3);
\draw [Vfill] (-1,3) rectangle (0,4);
\draw (-1.5,3.5) node {$\varnothing$};
\end{scope}
\begin{scope}[shift={(-1,1)}]
\draw[Hfill] (-2,2) rectangle (-1,3);
\draw [Vfill] (-1,3) rectangle (0,4);
\draw (-1.6,3.8) node {$\ddots$};
\end{scope}
\end{tikzpicture}
\begin{tikzpicture}[scale=.6]
\useasboundingbox (-3,1.5) (6,5);
\fill [Hfill] (-0.5,2) rectangle (1,4);
\fill [Vfill] (1,4) rectangle (3,5.5);
\draw (-0.5,4) -- (3,4);
\draw (1,2) -- (1,5.5);
\Hpoint{1}{3};
\draw (0.6,2.75) node {$\sigma_{i}$};
\Vpoint{2}{4};
\draw (2,4.4) node {$\sigma_{j}$};
\zoneRG{1}{4}{1.5}
\end{tikzpicture}

\caption{Zone $A_{RG}$}\label{fig:zoneRG}
\end{center}
\end{figure}
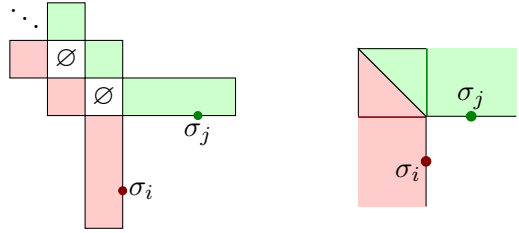

\begin{pf}
The proof is by induction on the {\em assigned} border between the zone not yet assigned to a stack and the {\em assigned} zone containing $\sigma_i$ and $\sigma_j$ where points are 
forced to be in a specific stack. 
At first, the {\em assigned} border is reduced to the segments $(i,1) -- (i,\sigma_j) -- (n,\sigma_j)$ as well as the assigned zone (where $n=|\sigma|$).

More formally we build sequences $(\sigma_{i_k})$ and $(\sigma_{j_k})$ such that $\sigma_{i_k} \sigma_{j_k}$ in an \ascentRG and the color of all points lying in the set $C_k = \{ \sigma_\ell \mid i_k \leq \ell \leq i \text{ and } \sigma_j \leq \sigma_\ell \leq \sigma_{j_k} \}$ is determined and respects Figure~\ref{fig:zoneRG}.
We set $i_0 = i$ and $j_0 = j$.
We prove that if ($\sigma_{i_k} \neq 1$ or $\sigma_{j_k} \neq n$) then we can build $\sigma_{i_{k+1}}$ and $\sigma_{j_{k+1}}$ such that $\sigma_{i_{k+1}} < \sigma_{i_k}$ or $\sigma_{j_{k+1}} > \sigma_{j_k}$.

We set $H_k = \{ \sigma_\ell \mid \ell \leq i_k \text{ and } \sigma_\ell \leq \sigma_{j_k} \}$ and $V_k = \{ \sigma_\ell \mid \ell \geq i_k \text{ and } \sigma_\ell \geq \sigma_{j_k} \}$ (see Figure~\ref{fig:H0V0}). 
By rule (\rmnum{3}) applied to $\sigma_{i_k}$ and $\sigma_{j_k}$, $H_k \subset H$ and $V_k \subset V$. 
Then, different situations may happen depending on whether areas $H_k$ and $V_k$ are empty:

\paragraph{$H_k$ and $V_k$ empty:} Then $\sigma$ is $\ominus$-decomposable which is in contradiction with our hypothesis.

\paragraph{$H_k$ and $V_k$ both non empty:} If both of the colored zones $H_k$ or $V_k$ are non empty, we set $i_{k+1} = \min \{ \ell \mid \sigma_\ell \in H_k \}$ and $\sigma_{j_{k+1}} = \max V_k$ (see Figure~\ref{fig:H0V0}). 
Then $C_{k+1} = C_k \cup H_k \cup V_k \cup Z$ is a partition of $C_{k+1}$, where $Z = \{ \sigma_\ell \mid i_{k+1} \leq \ell \leq i_k \text{ and } \sigma_{j_k} \leq \sigma_\ell \leq \sigma_{j_{k+1}} \}$ (see Figure~\ref{fig:H0V0}). 
The only points of $C_{k+1}$ whose color is not determined yet are those of $Z$. 
If $Z$ is not empty consider a point $\sigma_\ell$ of $Z$. 
If $\sigma_\ell \in V$ then rule (\rmnum{1}) applied to $\sigma_\ell$ and $\sigma_{j_{k+1}}$ is in contradiction with the existence of $\sigma_{i_k}$. 
Hence $\sigma_\ell \in H$ but then rule (\rmnum{2}) applied to $\sigma_{i_{k+1}}$ and $\sigma_\ell$ is in contradiction with the existence of $\sigma_{j_k}$. 
So $Z$ is empty and the color of all points of $C_{k+1}$ is determined and respects Figure~\ref{fig:zoneRG}.

\begin{figure}[H]
\begin{center}
\begin{tikzpicture}[scale=.5]
\Hpoint{0}{0};
\Vpoint{1}{1};
\draw [very thick,->] (2,0.5) -- (3,0.5);
\draw (4,0) -- (7,0);
\draw (4,1) -- (7,1);
\draw (5,-1) -- (5,2);
\draw (6,-1) -- (6,2);
\zoneA{H};
\zoneD{H};
\zoneI{V};
\zoneH{V};
\Hpoint{5}{0};
\Vpoint{6}{1};
\etiquette{3};
\draw (4.5,-0.5) node {$H_k$};
\draw (5.5,1.5) node {$V_k$};
\end{tikzpicture}
\begin{tikzpicture}[scale=.5]
\draw [very thick,->] (2,0.5) -- (3,0.5);
\draw (3,0) -- (7,0);
\draw (3,1) -- (7,1);
\draw (3,2) -- (7,2);
\draw (4,-1) -- (4,3);
\draw (5,-1) -- (5,3);
\draw (6,-1) -- (6,3);
\draw [H,Hfill,very thick] (4,-1) -- (4,1) -- (5,1) -- (5,-1);
\draw [V,Vfill,very thick] (7,1) -- (5,1) -- (5,2) -- (7,2);
\draw (3.5,-0.5) node {$\varnothing$};
\draw (3.5,0.5) node {$\varnothing$};
\draw (5.5,2.5) node {$\varnothing$};
\draw (6.5,2.5) node {$\varnothing$};
\draw (4.5,1.5) node {$Z$};
\Hpoint{5}{0};
\Vpoint{6}{1};
\end{tikzpicture}
\begin{tikzpicture}[scale=.5]
\draw [very thick,->] (2,0.5) -- (3,0.5);
\draw (3,0) -- (7,0);
\draw (3,1) -- (7,1);
\draw (3,2) -- (7,2);
\draw (4,-1) -- (4,3);
\draw (5,-1) -- (5,3);
\draw (6,-1) -- (6,3);
\draw [H,Hfill,very thick] (4,-1) -- (4,1) -- (5,1) -- (5,-1);
\draw [V,Vfill,very thick] (7,1) -- (5,1) -- (5,2) -- (7,2);
\draw (3.5,-0.5) node {$\varnothing$};
\draw (3.5,0.5) node {$\varnothing$};
\draw (5.5,2.5) node {$\varnothing$};
\draw (6.5,2.5) node {$\varnothing$};
\draw (4.5,1.5) node {$\varnothing$};
\Hpoint{5}{0};
\Vpoint{6}{1};
\end{tikzpicture}
\begin{tikzpicture}[scale=.5]
\draw [very thick,->] (2,0.5) -- (3,0.5);
\draw (3,0) -- (7,0);
\draw (3,1) -- (7,1);
\draw (3,2) -- (7,2);
\draw (4,-1) -- (4,3);
\draw (5,-1) -- (5,3);
\draw (6,-1) -- (6,3);
\draw [H,Hfill,very thick] (4,-1) -- (4,1) -- (5,1) -- (5,-1);
\draw [V,Vfill,very thick] (7,1) -- (5,1) -- (5,2) -- (7,2);
\draw [very thick, dashed] (4,-1) -- (4,2) -- (7,2);
\draw (3.5,-0.5) node {$\varnothing$};
\draw (3.5,0.5) node {$\varnothing$};
\draw (5.5,2.5) node {$\varnothing$};
\draw (6.5,2.5) node {$\varnothing$};
\draw (4.5,1.5) node {$\varnothing$};
\Hpoint{5}{0};
\Vpoint{6}{1};
\Hpoint{4}{0.5};
\Vpoint{6.5}{2};
\draw (4.5,0.5) node {{\small $\sigma_{i_{k+1}}$}};
\draw (6.5,1.7) node {{\small $\sigma_{j_{k+1}}$}};
\end{tikzpicture}
\caption{Transitive closure}\label{fig:H0V0}
\end{center}
\end{figure}

\paragraph{Only one area in $H_k$ and $V_k$ is empty:}
The same proof as the preceding case allow us to define a new point $\sigma_{j_{k+1}}$ or $\sigma_{i_{k+1}}$ depending on which area is empty and we can extend the {\em assigned} border as shown in next figure.
\begin{center}
\begin{tikzpicture}[scale=.5]
\Hpoint{0}{0};
\Vpoint{1}{1};
\draw [very thick,->] (2,0.5) -- (3,0.5);
\draw (4,0) -- (7,0);
\draw (4,1) -- (7,1);
\draw (5,-1) -- (5,2);
\draw (6,-1) -- (6,2);
\zoneA{H};
\zoneD{H};
\zoneI{V};
\zoneH{V};
\Hpoint{5}{0};
\Vpoint{6}{1};
\etiquette{3};
\draw (4.5,-0.5) node {$H_k$};
\draw (5.5,1.5) node {$V_k$};
\end{tikzpicture}
\begin{tikzpicture}[scale=.5]
\draw [very thick,->] (2,0.5) -- (3,0.5);
\draw (4,0) -- (7,0);
\draw (4,1) -- (7,1);
\draw (4,2) -- (7,2);
\draw (5,-1) -- (5,3);
\draw (6,-1) -- (6,3);
\draw [H,Hfill,very thick] (5,-1) -- (5,1) -- (5,1) -- (5,-1);
\draw [V,Vfill,very thick] (7,1) -- (5,1) -- (5,2) -- (7,2);
\draw (4.5,-0.5) node {$\varnothing$};
\draw (4.5,0.5) node {$\varnothing$};
\draw (5.5,2.5) node {$\varnothing$};
\draw (6.5,2.5) node {$\varnothing$};
\Hpoint{5}{0};
\Vpoint{6}{1};
\end{tikzpicture}
\begin{tikzpicture}[scale=.5]
\draw [very thick,->] (2,0.5) -- (3,0.5);
\draw (4,0) -- (7,0);
\draw (4,1) -- (7,1);
\draw (4,2) -- (7,2);
\draw (5,-1) -- (5,3);
\draw (6,-1) -- (6,3);
\draw [H,Hfill,very thick] (5,-1) -- (5,1) -- (5,1) -- (5,-1);
\draw [V,Vfill,very thick] (7,1) -- (5,1) -- (5,2) -- (7,2);
\draw (4.5,-0.5) node {$\varnothing$};
\draw (4.5,0.5) node {$\varnothing$};
\draw (5.5,2.5) node {$\varnothing$};
\draw (6.5,2.5) node {$\varnothing$};
\draw [very thick, dashed] (5,-1) -- (5,2) -- (7,2);
\Hpoint{5}{0};
\Vpoint{6}{1};
\Vpoint{6.5}{2};
\draw (6.5,1.7) node {{\small $\sigma_{j_{k+1}}$}};
\end{tikzpicture}
\begin{center}
{Transitive closure}
\end{center}
\end{center}

Hence, the {\em assigned} zone keeps growing until all permutation points are assigned, proving Theorem~\ref{thm:RGIncreasing}.
\end{pf}

We also have a similar result extending rule $(\rmnum{4})$:

\begin{thm}\label{thm:GRIncreasing}
Consider a valid coloring of a $\ominus$-decomposable permutation $\sigma$.
If there exist two points $\sigma_i < \sigma_j, i < j$ such that $\sigma_i \in \G$ and $\sigma_j \in \R$,
then the color of each point $\sigma_k$ with $k > j$ or $\sigma_k < \sigma_i$ is determined. 
Such a zone will be represented as 
\begin{tikzpicture}[scale=.5]
\useasboundingbox (1,0) (4,3);
\fill [Vfill] (1,0) rectangle (3,1);
\fill [Hfill] (3,1) rectangle (4,3);
\draw (1,1) -- (4,1);
\draw (3,0) -- (3,3);
\Hpoint{3}{2};
\draw (3.5,1.7) node {{\tiny $\sigma_j$}};
\Vpoint{2}{1};
\draw (2,0.7) node {{\tiny $\sigma_i$}};
\zoneGR{4}{0}{1};
\end{tikzpicture} in the sequel.
\end{thm}
\begin{pf}
Notice that rules are symetric so that the same proof as for Theorem~\ref{thm:RGIncreasing} holds.
\end{pf}

Knowing Theorem~\ref{thm:RGIncreasing} and Theorem~\ref{thm:GRIncreasing}, to set the color of as much points as possible, we better have to choose the lower right \ascentRG or the upper left \ascentGR. 
Let us now define properly these particular ascents.

We consider a valid bicoloring $c$ of a permutation $\sigma$.
We define $A_{RG}$ as the set of \ascentsRG of $c$.

\begin{lem}
Suppose $A_{RG} \neq \varnothing$. 
Among \ascentsRG of $c$, the pair $(\sigma_i,\sigma_j)$ which maximizes $i$ first then minimizes $\sigma_{j}$ (for $i$ fixed) is the same than the pair that minimizes $\sigma_{j}$ first then maximizes $i$ (for $\sigma_{j}$ fixed).
\end{lem}
\begin{proof}
Let $(\sigma_{i_{0}},\sigma_{j_{0}})$ be the pair that maximizes $i_0$ first then minimizes $\sigma_{j_{0}}$ and $(\sigma_{i_1},\sigma_{j_1})$ be the pair that minimizes $\sigma_{j_1}$ first then maximizes $i_1$. 
Then by definition $i_{0} \geq i_{1}$ and $\sigma_{j_{1}} \leq \sigma_{j_{0}}$.

If $j_{1} < i_{0}$ then $(\sigma_{j_1}, \sigma_{j_0})$ is an \ascentGG and rule $(\rmnum{1})$ is in contradiction with $\sigma_{i_0} \in H$ as $j_{1} < i_{0} < j_{0}$.
If $\sigma_{j_1} < \sigma_{i_0}$ then $(\sigma_{i_1}, \sigma_{i_0})$ is an \ascentRR and rule $(\rmnum{2})$ is in contradiction with $\sigma_{j_1} \in V$ as $\sigma_{i_1} < \sigma_{j_1} < \sigma_{i_0}$. 
Hence $(\sigma_{i_0}, \sigma_{j_1})$ is an \ascentRG.
Then by definition of $j_0$, $\sigma_{j_{0}} \leq \sigma_{j_{1}}$ and by definition of $i_1$, $i_{1} \geq i_{0}$. 
So $(\sigma_{i_{0}},\sigma_{j_{0}}) = (\sigma_{i_1},\sigma_{j_1})$.
\end{proof}

By the preceding lemma, when $A_{RG} \neq \varnothing$ we can define $i_{RG},j_{RG}$ as the lower right \ascentRG. 
By symmetry, we can also define $i_{GR},j_{GR}$ the upper left \ascentGR when $A_{RG} \neq \varnothing$, where $A_{GR}$ is the set similar to $A_{RG}$ but for \ascentsGR.

Now we have all the tools to prove that there are only a polynomial number of bicolorings to test. 
We juste have to do a case study depending on $A_{RG}$ or $A_{GR}$ are empty.

\subsection{Case study}

Recall that from Proposition~\ref{prop:pushallMoinsDecomposable} if $\sigma$ is $\ominus$-decomposable then $\sigma$ is $2$-stack pushall sortable 
if and only if each $\ominus$-indecomposable block of $\sigma$ is $2$-stack pushall sortable.
Thus, we can assume that $\sigma$ is $\ominus$-indecomposable.

In this section, we consider a valid coloring $c$ of a $\ominus$-indecomposable permutation $\sigma$.
We prove that knowing if there are ascents RG or GR in $c$ and knowing $i_{RG}$, $j_{RG}$, $i_{GR}$ and $j_{GR}$ (if they exist), we can deduce the color of every point of $\sigma$.

We prove this considering $4$ cases depending on whether there are ascents RG or GR in $c$.

\subsubsection{There is no bicolored ascents}

If $A_{RG}$ and $A_{GR}$ are both empty, then the coloring is monochromatic:
\begin{prop}\label{prop:monochromatic}
Let $\sigma$ be a $\ominus$-indecomposable permutation and $c$ a valid coloring of $\sigma$ such that every pattern $12$ of $\sigma$ is monochromatic. 
Then all points of $\sigma$ have the same color.
\end{prop}
\begin{pf}
Let $\sigma_i$ and $\sigma_j$ be two consecutive left-to-right minima of $\sigma$. 
By definition there are no point below $\sigma_i$ and to the left of $\sigma_j$ as shown by the empty sign in the following figure
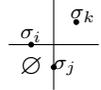
\begin{tikzpicture}[scale=.3]
\draw (0,2) -- (4,2);
\draw (2,0) -- (2,4);
\draw (1,1) node {$\varnothing$};
\draw (1,2) [fill] circle (3pt);
\draw (1,2.4) node {{\scriptsize $\sigma_i$}};
\draw (2,1) [fill] circle (3pt);
\draw (2.5,1) node {{\scriptsize $\sigma_j$}};
\draw (3,3) [fill] circle (3pt);
\draw (3.3,3.3) node {{\scriptsize $\sigma_k$}};
\end{tikzpicture}.
As $\sigma$ is $\ominus$-indecomposable, there exist a point $\sigma_k$ above $\sigma_j$ and to the right of $\sigma_i$. 
As increasing subsequences are monochromatic, $\sigma_i$ and $\sigma_k$ have the same color. 
The same goes for $\sigma_j$ and $\sigma_k$. Thus $\sigma_i$ and $\sigma_j$ have the same color. 
So all left-to-right minima of $\sigma$ have the same color.
By definition of left-to-right minima, for every non-minimal point $\sigma_l$ there exists a left-to-right minima $\sigma_m$ such that $(\sigma_m,\sigma_l)$ is a pattern $12$ of $\sigma$. 
Thus $\sigma_l$ has the same color as $\sigma_m$, and all points of $\sigma$ have the same color.
\end{pf}

\subsubsection{There is no \ascentRG but some \ascentsGR}

We suppose in this section that there exists at least one \ascentGR but no \ascentRG.
As $A_{GR}$ is non empty, $i_{GR}$ and $j_{GR}$ are defined.
We prove that once $i_{GR}$ and $j_{GR}$ are determined, then it fixes the color of every other point of the permutation.

\begin{prop}\label{prop:diagrammesGR}
Let $\sigma$ be a $\ominus$-indecomposable permutation and $c$ a valid coloring of $\sigma$ such that there is no increasing subsequence RG in $c$ and there is at least an increasing sequence GR in $c$. 
Then $c$ has one of the following shapes (where maybe $a=i_{GR}$ or $b=j_{GR}$):

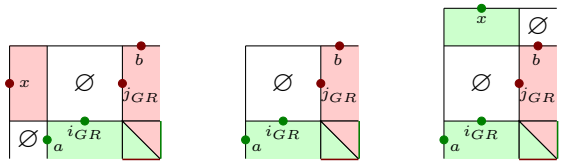
\begin{center}
\begin{tikzpicture}[scale=.5]
\useasboundingbox (0,-1) (5,4);
\fill [Vfill] (1,0) rectangle (3,1);
\fill [Hfill] (0,1) rectangle (1,3);
\fill [Hfill] (3,1) rectangle (4,3);
\draw (0,1) -- (4,1);
\draw (0,3) -- (4,3);
\draw (0,0) -- (0,3);
\draw (1,0) -- (1,3);
\draw (3,0) -- (3,3);
\draw (2,2) node {$\varnothing$};
\draw (0.5,0.5) node {$\varnothing$};
\Hpoint{3}{2};
\draw (3.5,1.7) node {{\tiny $j_{GR}$}};
\Vpoint{2}{1};
\draw (2,0.7) node {{\tiny $i_{GR}$}};
\Vpoint{1}{0.5};
\draw (1.3,0.3) node {{\tiny $a$}};
\Hpoint{3.5}{3};
\draw (3.45,2.65) node {{\tiny $b$}};
\Hpoint{0}{2};
\draw (0.4,2) node {{\tiny $x$}};
\zoneGR{4}{0}{1};
\end{tikzpicture}
\begin{tikzpicture}[scale=.5]
\useasboundingbox (0,-1) (5,4);
\fill [Vfill] (1,0) rectangle (3,1);
\fill [Hfill] (3,1) rectangle (4,3);
\draw (1,1) -- (4,1);
\draw (1,3) -- (4,3);
\draw (1,0) -- (1,3);
\draw (3,0) -- (3,3);
\draw (2,2) node {$\varnothing$};
\Hpoint{3}{2};
\draw (3.5,1.7) node {{\tiny $j_{GR}$}};
\Vpoint{2}{1};
\draw (2,0.7) node {{\tiny $i_{GR}$}};
\Vpoint{1}{0.5};
\draw (1.3,0.3) node {{\tiny $a$}};
\Hpoint{3.5}{3};
\draw (3.45,2.65) node {{\tiny $b$}};
\zoneGR{4}{0}{1};
\end{tikzpicture}
\begin{tikzpicture}[scale=.5]
\useasboundingbox (0,-1) (5,4);
\fill [Vfill] (1,0) rectangle (3,1);
\fill [Hfill] (3,1) rectangle (4,3);
\fill [Vfill] (1,3) rectangle (3,4);
\draw (1,1) -- (4,1);
\draw (1,3) -- (4,3);
\draw (1,4) -- (4,4);
\draw (1,0) -- (1,4);
\draw (3,0) -- (3,4);
\draw (3.5,3.5) node {$\varnothing$};
\draw (2,2) node {$\varnothing$};
\Hpoint{3}{2};
\draw (3.5,1.7) node {{\tiny $j_{GR}$}};
\Vpoint{2}{1};
\draw (2,0.7) node {{\tiny $i_{GR}$}};
\Vpoint{1}{0.5};
\draw (1.3,0.3) node {{\tiny $a$}};
\Hpoint{3.5}{3};
\draw (3.45,2.65) node {{\tiny $b$}};
\Vpoint{2}{4};
\draw (2,3.7) node {{\tiny $x$}};
\zoneGR{4}{0}{1};
\end{tikzpicture}
\end{center}
\end{prop}

\begin{rem}
Here and in all the following, when a zone of a diagram is colored with \R (resp. \G), it means than if there are some points lying in this zone, they are in \R (resp. \G). And when a zone of a diagram has an empty sign, it means than this zone is empty.
% Attention a ne pas utiliser les mots red et green, pas cohérent pour la version noir et blanc
\end{rem}

\begin{pf}
The color of every point $\sigma_k$ such that $k > j_{GR}$ or $\sigma_k < \sigma_{i_{GR}}$ is determined by Theorem~\ref{thm:GRIncreasing} (see the first diagram of Figure~\ref{fig:tikzGR1}).
Note that we denote by $*$ the zone where the color of the points is unknown.
By maximality of $\sigma_{i_{GR}}$, any point above $\sigma_{i_{GR}}$ and lower left with respect to $\sigma_{j_{GR}}$ is in \R. By minimality of $j_{GR}$, any point to the left of $\sigma_{j_{GR}}$ and top right with respect to $\sigma_{i_{GR}}$ is in \G. As no point can be both in \R and in \G, we know that the zone between $\sigma_{i_{GR}}$ and $\sigma_{j_{GR}}$ is empty, as shown in the second diagram of Figure~\ref{fig:tikzGR1}.

\begin{figure}[H]
\begin{center}
\begin{tikzpicture}[scale=.5]
\useasboundingbox (0,-0.5) (5,5);
\fill [Vfill] (0,0) rectangle (3,1);
\fill [Hfill] (3,1) rectangle (4,4);
\draw (0,1) -- (4,1);
\draw (3,0) -- (3,4);
\Hpoint{3}{2};
\draw (3.5,1.7) node {{\tiny $j_{GR}$}};
\Vpoint{2}{1};
\draw (2,0.7) node {{\tiny $i_{GR}$}};
\draw (1.5,2.5) node {$*$};
\zoneGR{4}{0}{1};
\end{tikzpicture}
\begin{tikzpicture}[scale=.5]
\useasboundingbox (0,-0.5) (5,5);
\fill [Vfill] (0,0) rectangle (3,1);
\fill [Hfill] (0,1) rectangle (2,2);
\fill [Hfill] (3,1) rectangle (4,4);
\fill [Vfill] (2,2) rectangle (3,4);
\draw (0,1) -- (4,1);
\draw (0,2) -- (3,2);
\draw (2,1) -- (2,4);
\draw (3,0) -- (3,4);
\draw (2.5,1.5) node {$\varnothing$};
\Hpoint{3}{2};
\draw (3.5,1.7) node {{\tiny $j_{GR}$}};
\Vpoint{2}{1};
\draw (2,0.7) node {{\tiny $i_{GR}$}};
\draw (1,3) node {$*$};
\zoneGR{4}{0}{1};
\end{tikzpicture}
\begin{tikzpicture}[scale=.5]
\useasboundingbox (0,-0.5) (5,5);
\fill [Vfill] (1,0) rectangle (3,1);
\fill [Hfill] (0,1) rectangle (1,2);
\fill [Hfill] (3,1) rectangle (4,4);
\fill [Vfill] (1,2) rectangle (3,4);
\draw (0,1) -- (4,1);
\draw (0,2) -- (3,2);
\draw (2,1) -- (2,4);
\draw (1,0) -- (1,4);
\draw (3,0) -- (3,4);
\draw (2.5,1.5) node {$\varnothing$};
\draw (1.5,1.5) node {$\varnothing$};
\draw (0.5,0.5) node {$\varnothing$};
\Hpoint{3}{2};
\draw (3.5,1.7) node {{\tiny $j_{GR}$}};
\Vpoint{2}{1};
\draw (2,0.7) node {{\tiny $i_{GR}$}};
\Vpoint{1}{0.5};
\draw (1.3,0.3) node {{\tiny $a$}};
\draw (0.5,3) node {$*$};
\zoneGR{4}{0}{1};
\end{tikzpicture}
\begin{tikzpicture}[scale=.5]
\useasboundingbox (0,-0.5) (5,5);
\fill [Vfill] (1,0) rectangle (3,1);
\fill [Hfill] (0,1) rectangle (1,3);
\fill [Hfill] (3,1) rectangle (4,3);
\fill [Vfill] (1,3) rectangle (3,4);
\draw (0,1) -- (4,1);
\draw (0,2) -- (3,2);
\draw (0,3) -- (4,3);
\draw (2,1) -- (2,4);
\draw (1,0) -- (1,4);
\draw (3,0) -- (3,4);
\draw (3.5,3.5) node {$\varnothing$};
\draw (2.5,2.5) node {$\varnothing$};
\draw (1.5,2.5) node {$\varnothing$};
\draw (2.5,1.5) node {$\varnothing$};
\draw (1.5,1.5) node {$\varnothing$};
\draw (0.5,0.5) node {$\varnothing$};
\Hpoint{3}{2};
\draw (3.5,1.7) node {{\tiny $j_{GR}$}};
\Vpoint{2}{1};
\draw (2,0.7) node {{\tiny $i_{GR}$}};
\Vpoint{1}{0.5};
\draw (1.3,0.3) node {{\tiny $a$}};
\Hpoint{3.5}{3};
\draw (3.45,2.65) node {{\tiny $b$}};
\draw (0.5,3.5) node {$*$};
\zoneGR{4}{0}{1};
\end{tikzpicture}
\begin{tikzpicture}[scale=.5]
\useasboundingbox (0,-0.5) (5,5);
\fill [Vfill] (1,0) rectangle (3,1);
\fill [Hfill] (0,1) rectangle (1,3);
\fill [Hfill] (3,1) rectangle (4,3);
\fill [Vfill] (1,3) rectangle (3,4);
\draw (0,1) -- (4,1);
\draw (0,3) -- (4,3);
\draw (1,0) -- (1,4);
\draw (3,0) -- (3,4);
\draw (3.5,3.5) node {$\varnothing$};
\draw (2,2) node {$\varnothing$};
\draw (0.5,0.5) node {$\varnothing$};
\Hpoint{3}{2};
\draw (3.5,1.7) node {{\tiny $j_{GR}$}};
\Vpoint{2}{1};
\draw (2,0.7) node {{\tiny $i_{GR}$}};
\Vpoint{1}{0.5};
\draw (1.3,0.3) node {{\tiny $a$}};
\Hpoint{3.5}{3};
\draw (3.45,2.65) node {{\tiny $b$}};
\draw (0.5,3.5) node {$*$};
\draw (0.5,2) node {$1$};
\draw (2,3.5) node {$2$};
\zoneGR{4}{0}{1};
\end{tikzpicture}

\caption{Only bicolored increasing subsequences GR exist}\label{fig:tikzGR1}
\end{center}
\end{figure}
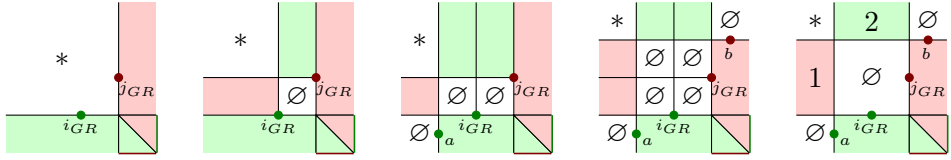

Let $a$ be the leftmost point among points below $i_{GR}$ (notice that $a$ may be equal to $i_{GR}$). Applying rule (\rmnum{1}) to $a$ and $i_{GR}$, we obtain the third diagram (note that if $a=i_{GR}$ the column between $i_{GR}$ and $a$ does not exist). Let $b$ be the topmost point to the right of $j_{GR}$ ($b$ may be equal to $j_{GR}$). Applying rule (\rmnum{2}) to $j_{GR}$ and $b$, we obtain the fourth diagram of Figure~\ref{fig:tikzGR1} (if $b=j_{GR}$ the column between $j_{GR}$ and $b$ does not exist).

At last, we number two different areas and discuss about the different cases whether these zones are empty or not. These zones are pictured in the fifth diagram of Figure~\ref{fig:tikzGR1}.

\paragraph{Zone $1$ is not empty}

Let $x$ be the leftmost point inside zone $1$. Note that $x$ may be above or below $j_{GR}$. First diagram of Figure~\ref{fig:tikz1nonVide} illustrates the position of point $x$. Applying rule (\rmnum{2}) to $x$ and $b$ we obtain the second diagram of Figure~\ref{fig:tikz1nonVide}.
By hypothesis, there are no increasing sequence RG, thus there are no points in \G in the up-right quadrant of $x$. This leads to the third diagram.
At last, if the zone $*$ is not empty, then $\sigma$ is $\ominus$-decomposable by cutting along the row of $b$ and the column of $x$. Thus $*$ is empty and all points have a determined color, as in the first diagram of Proposition~\ref{prop:diagrammesGR}.

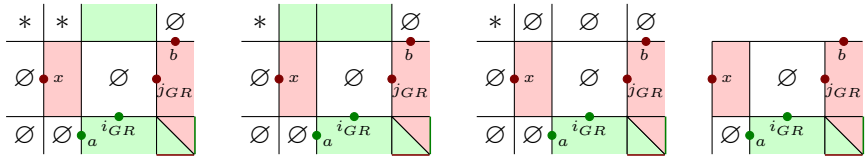
\begin{figure}[H]
\begin{center}
\begin{tikzpicture}[scale=.5]
\useasboundingbox (-1,-0.5) (5,5);
\fill [Vfill] (1,0) rectangle (3,1);
\fill [Hfill] (0,1) rectangle (1,3);
\fill [Hfill] (3,1) rectangle (4,3);
\fill [Vfill] (1,3) rectangle (3,4);
\draw (-1,1) -- (4,1);
\draw (-1,3) -- (4,3);
\draw (0,0) -- (0,4);
\draw (1,0) -- (1,4);
\draw (3,0) -- (3,4);
\draw (3.5,3.5) node {$\varnothing$};
\draw (2,2) node {$\varnothing$};
\draw (0.5,0.5) node {$\varnothing$};
\draw (-0.5,0.5) node {$\varnothing$};
\draw (-0.5,2) node {$\varnothing$};
\Hpoint{3}{2};
\draw (3.5,1.7) node {{\tiny $j_{GR}$}};
\Vpoint{2}{1};
\draw (2,0.7) node {{\tiny $i_{GR}$}};
\Vpoint{1}{0.5};
\draw (1.3,0.3) node {{\tiny $a$}};
\Hpoint{3.5}{3};
\draw (3.45,2.65) node {{\tiny $b$}};
\Hpoint{0}{2};
\draw (0.4,2) node {{\tiny $x$}};
\draw (0.5,3.5) node {$*$};
\draw (-0.5,3.5) node {$*$};
\zoneGR{4}{0}{1};
\end{tikzpicture}
\begin{tikzpicture}[scale=.5]
\useasboundingbox (-1,-0.5) (5,5);
\fill [Vfill] (1,0) rectangle (3,1);
\fill [Hfill] (0,1) rectangle (1,3);
\fill [Hfill] (3,1) rectangle (4,3);
\fill [Vfill] (0,3) rectangle (3,4);
\draw (-1,1) -- (4,1);
\draw (-1,3) -- (4,3);
\draw (0,0) -- (0,4);
\draw (1,0) -- (1,4);
\draw (3,0) -- (3,4);
\draw (3.5,3.5) node {$\varnothing$};
\draw (2,2) node {$\varnothing$};
\draw (0.5,0.5) node {$\varnothing$};
\draw (-0.5,0.5) node {$\varnothing$};
\draw (-0.5,2) node {$\varnothing$};
\Hpoint{3}{2};
\draw (3.5,1.7) node {{\tiny $j_{GR}$}};
\Vpoint{2}{1};
\draw (2,0.7) node {{\tiny $i_{GR}$}};
\Vpoint{1}{0.5};
\draw (1.3,0.3) node {{\tiny $a$}};
\Hpoint{3.5}{3};
\draw (3.45,2.65) node {{\tiny $b$}};
\Hpoint{0}{2};
\draw (0.4,2) node {{\tiny $x$}};
\draw (-0.5,3.5) node {$*$};
\zoneGR{4}{0}{1};
\end{tikzpicture}
\begin{tikzpicture}[scale=.5]
\useasboundingbox (-1,-0.5) (5,5);
\fill [Vfill] (1,0) rectangle (3,1);
\fill [Hfill] (0,1) rectangle (1,3);
\fill [Hfill] (3,1) rectangle (4,3);
\draw (-1,1) -- (4,1);
\draw (-1,3) -- (4,3);
\draw (0,0) -- (0,4);
\draw (1,0) -- (1,4);
\draw (3,0) -- (3,4);
\draw (3.5,3.5) node {$\varnothing$};
\draw (2,2) node {$\varnothing$};
\draw (0.5,0.5) node {$\varnothing$};
\draw (-0.5,0.5) node {$\varnothing$};
\draw (-0.5,2) node {$\varnothing$};
\draw (0.5,3.5) node {$\varnothing$};
\draw (2,3.5) node {$\varnothing$};
\Hpoint{3}{2};
\draw (3.5,1.7) node {{\tiny $j_{GR}$}};
\Vpoint{2}{1};
\draw (2,0.7) node {{\tiny $i_{GR}$}};
\Vpoint{1}{0.5};
\draw (1.3,0.3) node {{\tiny $a$}};
\Hpoint{3.5}{3};
\draw (3.45,2.65) node {{\tiny $b$}};
\Hpoint{0}{2};
\draw (0.4,2) node {{\tiny $x$}};
\draw (-0.5,3.5) node {$*$};
\zoneGR{4}{0}{1};
\end{tikzpicture}
\begin{tikzpicture}[scale=.5]
\useasboundingbox (0,-0.5) (5,4);
\fill [Vfill] (1,0) rectangle (3,1);
\fill [Hfill] (0,1) rectangle (1,3);
\fill [Hfill] (3,1) rectangle (4,3);
\draw (0,1) -- (4,1);
\draw (0,3) -- (4,3);
\draw (0,0) -- (0,3);
\draw (1,0) -- (1,3);
\draw (3,0) -- (3,3);
\draw (2,2) node {$\varnothing$};
\draw (0.5,0.5) node {$\varnothing$};
\Hpoint{3}{2};
\draw (3.5,1.7) node {{\tiny $j_{GR}$}};
\Vpoint{2}{1};
\draw (2,0.7) node {{\tiny $i_{GR}$}};
\Vpoint{1}{0.5};
\draw (1.3,0.3) node {{\tiny $a$}};
\Hpoint{3.5}{3};
\draw (3.45,2.65) node {{\tiny $b$}};
\Hpoint{0}{2};
\draw (0.4,2) node {{\tiny $x$}};
\zoneGR{4}{0}{1};
\end{tikzpicture}

\caption{Zone $1$ is not empty\label{fig:tikz1nonVide}}
\end{center}
\end{figure}

\paragraph{Zone $1$ is empty}

Suppose that zone $1$ is empty. If zone $2$ is also empty then as $\sigma$ is $\ominus$-indecomposable, zone $*$ is also empty and all points have a determined color, as in the second diagram of Proposition~\ref{prop:diagrammesGR}.

\begin{figure}[H]
\begin{center}
\begin{tikzpicture}[scale=.5]
\useasboundingbox (0,-0.5) (5,5);
\fill [Vfill] (1,0) rectangle (3,1);
\fill [Hfill] (3,1) rectangle (4,3);
\fill [Vfill] (1,3) rectangle (3,4);
\draw (0,1) -- (4,1);
\draw (0,3) -- (4,3);
\draw (1,0) -- (1,4);
\draw (3,0) -- (3,4);
\draw (3.5,3.5) node {$\varnothing$};
\draw (2,2) node {$\varnothing$};
\draw (0.5,2) node {$\varnothing$};
\draw (0.5,0.5) node {$\varnothing$};
\Hpoint{3}{2};
\draw (3.5,1.7) node {{\tiny $j_{GR}$}};
\Vpoint{2}{1};
\draw (2,0.7) node {{\tiny $i_{GR}$}};
\Vpoint{1}{0.5};
\draw (1.3,0.3) node {{\tiny $a$}};
\Hpoint{3.5}{3};
\draw (3.45,2.65) node {{\tiny $b$}};
\draw (0.5,3.5) node {$*$};
\draw (2,3.5) node {$2$};
\zoneGR{4}{0}{1};
\end{tikzpicture}
\begin{tikzpicture}[scale=.5]
\useasboundingbox (0,-0.5) (5,6);
\fill [Vfill] (1,0) rectangle (3,1);
\fill [Hfill] (3,1) rectangle (4,3);
\fill [Vfill] (1,3) rectangle (3,4);
\draw (0,1) -- (4,1);
\draw (0,3) -- (4,3);
\draw (0,4) -- (4,4);
\draw (1,0) -- (1,5);
\draw (3,0) -- (3,5);
\draw (3.5,3.5) node {$\varnothing$};
\draw (2,2) node {$\varnothing$};
\draw (0.5,2) node {$\varnothing$};
\draw (0.5,0.5) node {$\varnothing$};
\draw (2,4.5) node {$\varnothing$};
\draw (3.5,4.5) node {$\varnothing$};
\Hpoint{3}{2};
\draw (3.5,1.7) node {{\tiny $j_{GR}$}};
\Vpoint{2}{1};
\draw (2,0.7) node {{\tiny $i_{GR}$}};
\Vpoint{1}{0.5};
\draw (1.3,0.3) node {{\tiny $a$}};
\Hpoint{3.5}{3};
\draw (3.45,2.65) node {{\tiny $b$}};
\Vpoint{2}{4};
\draw (2,3.7) node {{\tiny $x$}};
\draw (0.5,3.5) node {$*$};
\draw (0.5,4.5) node {$*$};
\zoneGR{4}{0}{1};
\end{tikzpicture}
\begin{tikzpicture}[scale=.5]
\useasboundingbox (0,-0.5) (5,6);
\fill [Vfill] (1,0) rectangle (3,1);
\fill [Hfill] (3,1) rectangle (4,3);
\fill [Vfill] (1,3) rectangle (3,4);
\fill [Hfill] (0,3) rectangle (1,4);
\draw (0,1) -- (4,1);
\draw (0,3) -- (4,3);
\draw (0,4) -- (4,4);
\draw (1,0) -- (1,5);
\draw (3,0) -- (3,5);
\draw (3.5,3.5) node {$\varnothing$};
\draw (2,2) node {$\varnothing$};
\draw (0.5,2) node {$\varnothing$};
\draw (0.5,0.5) node {$\varnothing$};
\draw (2,4.5) node {$\varnothing$};
\draw (3.5,4.5) node {$\varnothing$};
\Hpoint{3}{2};
\draw (3.5,1.7) node {{\tiny $j_{GR}$}};
\Vpoint{2}{1};
\draw (2,0.7) node {{\tiny $i_{GR}$}};
\Vpoint{1}{0.5};
\draw (1.3,0.3) node {{\tiny $a$}};
\Hpoint{3.5}{3};
\draw (3.45,2.65) node {{\tiny $b$}};
\Vpoint{2}{4};
\draw (2,3.7) node {{\tiny $x$}};
\draw (0.5,4.5) node {$*$};
\zoneGR{4}{0}{1};
\end{tikzpicture}
\begin{tikzpicture}[scale=.5]
\useasboundingbox (0,-0.5) (5,6);
\fill [Vfill] (1,0) rectangle (3,1);
\fill [Hfill] (3,1) rectangle (4,3);
\fill [Vfill] (1,3) rectangle (3,4);
\draw (0,1) -- (4,1);
\draw (0,3) -- (4,3);
\draw (0,4) -- (4,4);
\draw (1,0) -- (1,5);
\draw (3,0) -- (3,5);
\draw (3.5,3.5) node {$\varnothing$};
\draw (2,2) node {$\varnothing$};
\draw (0.5,2) node {$\varnothing$};
\draw (0.5,0.5) node {$\varnothing$};
\draw (2,4.5) node {$\varnothing$};
\draw (3.5,4.5) node {$\varnothing$};
\Hpoint{3}{2};
\draw (3.5,1.7) node {{\tiny $j_{GR}$}};
\Vpoint{2}{1};
\draw (2,0.7) node {{\tiny $i_{GR}$}};
\Vpoint{1}{0.5};
\draw (1.3,0.3) node {{\tiny $a$}};
\Hpoint{3.5}{3};
\draw (3.45,2.65) node {{\tiny $b$}};
\Vpoint{2}{4};
\draw (2,3.7) node {{\tiny $x$}};
\draw (0.5,4.5) node {$*$};
\draw (0.5,3.5) node {$\varnothing$};
\zoneGR{4}{0}{1};
\end{tikzpicture}
\begin{tikzpicture}[scale=.5]
\useasboundingbox (0,-0.5) (5,4);
\fill [Vfill] (1,0) rectangle (3,1);
\fill [Hfill] (3,1) rectangle (4,3);
\fill [Vfill] (1,3) rectangle (3,4);
\draw (1,1) -- (4,1);
\draw (1,3) -- (4,3);
\draw (1,4) -- (4,4);
\draw (1,0) -- (1,4);
\draw (3,0) -- (3,4);
\draw (3.5,3.5) node {$\varnothing$};
\draw (2,2) node {$\varnothing$};
\Hpoint{3}{2};
\draw (3.5,1.7) node {{\tiny $j_{GR}$}};
\Vpoint{2}{1};
\draw (2,0.7) node {{\tiny $i_{GR}$}};
\Vpoint{1}{0.5};
\draw (1.3,0.3) node {{\tiny $a$}};
\Hpoint{3.5}{3};
\draw (3.45,2.65) node {{\tiny $b$}};
\Vpoint{2}{4};
\draw (2,3.7) node {{\tiny $x$}};
\zoneGR{4}{0}{1};
\end{tikzpicture}
\caption{Zone $1$ is empty\label{fig:tikz1Vide}}
\end{center}
\end{figure}
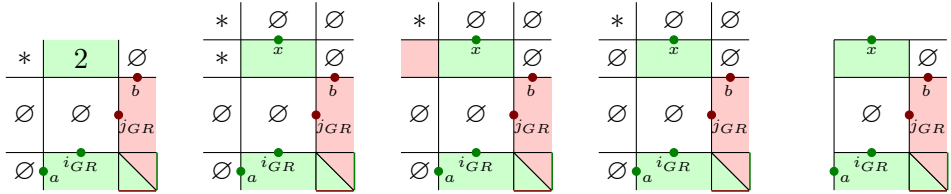

Otherwise, zone $2$ is not empty and let $x$ be the topmost point inside zone $2$ ($x$ may be to the left or to the right of $i_{GR}$). This is depicted in the second diagram of Figure~\ref{fig:tikz1Vide}. We apply rule (\rmnum{1}) to $a$ and $x$ to obtain the third diagram. As there is no increasing subsequence RG, there is no point of \R in the lower left quadrant of $x$ as depicted in the fourth diagram. Moreover, $\sigma$ is $\ominus$-indecomposable, thus zone $*$ is empty and each point has a determined color, as in the last diagram of Proposition~\ref{prop:diagrammesGR}.
\end{pf}

\begin{defn}\label{def:C_GR}
Let $\sigma$ be a permutation and $i$ and $j$ two indices of $\sigma$ such that $\sigma_i \sigma_j$ is an ascent.
Set $a = \min \{k \mid \sigma_k \leq \sigma_i\}$ and $b$ such that $\sigma_b = \max \{\sigma_k \mid k \geq j\}$.
We define $C_{GR}(\sigma, i, j)$ as the partial bicoloring of $\sigma$ having the following shape:

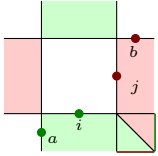
\begin{tikzpicture}[scale=.5]
\useasboundingbox (0,-1) (5,4.5);
\fill [Vfill] (1,0) rectangle (3,1);
\fill [Hfill] (0,1) rectangle (1,3);
\fill [Hfill] (3,1) rectangle (4,3);
\fill [Vfill] (1,3) rectangle (3,4);
\draw (1,0) -- (1,4);
\draw (3,0) -- (3,4);
\draw (0,1) -- (4,1);
\draw (0,3) -- (4,3);
\Hpoint{3}{2};
\draw (3.5,1.7) node {{\tiny $j$}};
\Vpoint{2}{1};
\draw (2,0.7) node {{\tiny $i$}};
\Vpoint{1}{0.5};
\draw (1.3,0.3) node {{\tiny $a$}};
\Hpoint{3.5}{3};
\draw (3.45,2.65) node {{\tiny $b$}};
\zoneGR{4}{0}{1};
\end{tikzpicture}
\end{defn}

\begin{prop}\label{prop:C_GR}
Let $\sigma$ be a $\ominus$-indecomposable permutation and $c$ a valid coloring of $\sigma$ such that there is no increasing subsequence RG in $c$ and there is at least an increasing sequence GR in $c$. 
Then $c = C_{GR}(\sigma, i_{GR}, j_{GR})$.
\end{prop}

\begin{pf}
This is a direct consequence of Proposition~\ref{prop:diagrammesGR} and Definition~\ref{def:C_GR}.
\end{pf}

\subsubsection{All bicolored increasing sequences are labeled $RG$}

We suppose in this section that there exists at least one \ascentRG but no \ascentGR.
As $A_{RG}$ is non empty, $i_{RG}$ and $j_{RG}$ are defined.
We prove that once $i_{RG}$ and $j_{RG}$ are determined, then it fixes the color of every other point of the permutation.

\begin{prop}\label{prop:diagrammesRG}
Let $\sigma$ be a $\ominus$-indecomposable permutation and $c$ a valid coloring of $\sigma$ such that there is no increasing subsequence GR in $c$ and there is at least an increasing sequence RG in $c$. 
Then $c$ has one of the following shapes (where maybe $a=j_{RG}$ or $b=i_{RG}$):

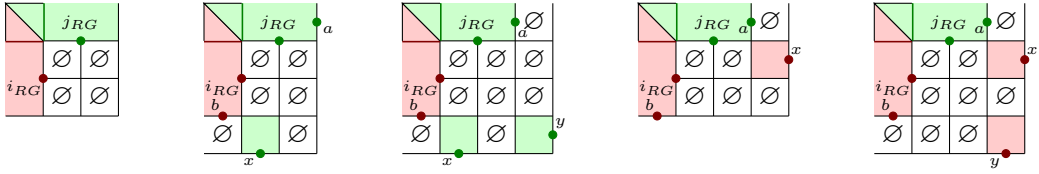
\begin{center}
\begin{tikzpicture}[scale=.5]
\useasboundingbox (0,-1) (5,6);
\fill [Hfill] (0,2) rectangle (1,4);
\fill [Vfill] (1,4) rectangle (3,5);
\draw (0,4) -- (3,4);
\draw (1,3) -- (3,3);
\draw (0,2) -- (3,2);
\draw (1,2) -- (1,5);
\draw (2,2) -- (2,4);
\draw (3,2) -- (3,5);
\draw (2.5,3.5) node {$\varnothing$};
\draw (2.5,2.5) node {$\varnothing$};
\draw (1.5,2.5) node {$\varnothing$};
\draw (1.5,3.5) node {$\varnothing$};
\Hpoint{1}{3};
\draw (0.5,2.75) node {{\tiny $i_{RG}$}};
\Vpoint{2}{4};
\draw (2,4.5) node {{\tiny $j_{RG}$}};
\zoneRG{1}{4}{1}
\end{tikzpicture}
\begin{tikzpicture}[scale=.5]
\useasboundingbox (0,-1) (5,6);
\fill [Hfill] (0,2) rectangle (1,4);
\fill [Vfill] (1,4) rectangle (3,5);
\fill [Vfill] (1,1) rectangle (2,2);
\draw (0,4) -- (3,4);
\draw (1,3) -- (3,3);
\draw (0,2) -- (3,2);
\draw (0,1) -- (3,1);
\draw (1,1) -- (1,5);
\draw (2,1) -- (2,4);
\draw (3,1) -- (3,5);
\draw (2.5,3.5) node {$\varnothing$};
\draw (2.5,2.5) node {$\varnothing$};
\draw (0.5,1.5) node {$\varnothing$};
\draw (2.5,1.5) node {$\varnothing$};
\draw (1.5,2.5) node {$\varnothing$};
\draw (1.5,3.5) node {$\varnothing$};
\Hpoint{1}{3};
\draw (0.5,2.75) node {{\tiny $i_{RG}$}};
\Vpoint{2}{4};
\draw (2,4.5) node {{\tiny $j_{RG}$}};
\Vpoint{3}{4.5};
\draw (3.3,4.3) node {{\tiny $a$}};
\Hpoint{0.5}{2};
\draw (0.3,2.3) node {{\tiny $b$}};
\Vpoint{1.5}{1};
\draw (1.2,0.8) node {{\tiny $x$}};
\zoneRG{1}{4}{1}
\end{tikzpicture}
\begin{tikzpicture}[scale=.5]
\useasboundingbox (0,-1) (6,6);
\fill [Hfill] (0,2) rectangle (1,4);
\fill [Vfill] (1,4) rectangle (3,5);
\fill [Vfill] (1,1) rectangle (2,2);
\fill [Vfill] (3,1) rectangle (4,2);
\draw (0,4) -- (4,4);
\draw (1,3) -- (4,3);
\draw (0,2) -- (4,2);
\draw (0,1) -- (4,1);
\draw (1,1) -- (1,5);
\draw (2,1) -- (2,4);
\draw (3,1) -- (3,5);
\draw (4,1) -- (4,5);
\draw (2.5,3.5) node {$\varnothing$};
\draw (2.5,2.5) node {$\varnothing$};
\draw (0.5,1.5) node {$\varnothing$};
\draw (2.5,1.5) node {$\varnothing$};
\draw (1.5,2.5) node {$\varnothing$};
\draw (1.5,3.5) node {$\varnothing$};
\draw (3.5,2.5) node {$\varnothing$};
\draw (3.5,3.5) node {$\varnothing$};
\draw (3.5,4.5) node {$\varnothing$};
\Hpoint{1}{3};
\draw (0.5,2.75) node {{\tiny $i_{RG}$}};
\Vpoint{2}{4};
\draw (2,4.5) node {{\tiny $j_{RG}$}};
\Vpoint{3}{4.5};
\draw (3.2,4.2) node {{\tiny $a$}};
\Hpoint{0.5}{2};
\draw (0.3,2.3) node {{\tiny $b$}};
\Vpoint{1.5}{1};
\draw (1.2,0.8) node {{\tiny $x$}};
\Vpoint{4}{1.5};
\draw (4.2,1.8) node {{\tiny $y$}};
\zoneRG{1}{4}{1}
\end{tikzpicture}
\begin{tikzpicture}[scale=.5]
\useasboundingbox (0,-1) (6,6);
\fill [Hfill] (0,2) rectangle (1,4);
\fill [Vfill] (1,4) rectangle (3,5);
\fill [Hfill] (3,3) rectangle (4,4);
\draw (0,4) -- (4,4);
\draw (1,3) -- (4,3);
\draw (0,2) -- (4,2);
\draw (1,2) -- (1,5);
\draw (2,2) -- (2,4);
\draw (3,2) -- (3,5);
\draw (4,2) -- (4,5);
\draw (2.5,3.5) node {$\varnothing$};
\draw (2.5,2.5) node {$\varnothing$};
\draw (1.5,2.5) node {$\varnothing$};
\draw (1.5,3.5) node {$\varnothing$};
\draw (3.5,2.5) node {$\varnothing$};
\draw (3.5,4.5) node {$\varnothing$};
\Hpoint{1}{3};
\draw (0.5,2.75) node {{\tiny $i_{RG}$}};
\Vpoint{2}{4};
\draw (2,4.5) node {{\tiny $j_{RG}$}};
\Vpoint{3}{4.5};
\draw (2.8,4.3) node {{\tiny $a$}};
\Hpoint{0.5}{2};
\draw (0.3,2.3) node {{\tiny $b$}};
\Hpoint{4}{3.5};
\draw (4.2,3.75) node {{\tiny $x$}};
\zoneRG{1}{4}{1}
\end{tikzpicture}
\begin{tikzpicture}[scale=.5]
\useasboundingbox (0,-1) (5,6);
\fill [Hfill] (0,2) rectangle (1,4);
\fill [Vfill] (1,4) rectangle (3,5);
\fill [Hfill] (3,3) rectangle (4,4);
\fill [Hfill] (3,1) rectangle (4,2);
\draw (0,4) -- (4,4);
\draw (1,3) -- (4,3);
\draw (0,2) -- (4,2);
\draw (0,1) -- (4,1);
\draw (1,1) -- (1,5);
\draw (2,1) -- (2,4);
\draw (3,1) -- (3,5);
\draw (4,1) -- (4,5);
\draw (2.5,3.5) node {$\varnothing$};
\draw (2.5,2.5) node {$\varnothing$};
\draw (0.5,1.5) node {$\varnothing$};
\draw (1.5,1.5) node {$\varnothing$};
\draw (2.5,1.5) node {$\varnothing$};
\draw (1.5,2.5) node {$\varnothing$};
\draw (1.5,3.5) node {$\varnothing$};
\draw (3.5,2.5) node {$\varnothing$};
\draw (3.5,4.5) node {$\varnothing$};
\Hpoint{1}{3};
\draw (0.5,2.75) node {{\tiny $i_{RG}$}};
\Vpoint{2}{4};
\draw (2,4.5) node {{\tiny $j_{RG}$}};
\Vpoint{3}{4.5};
\draw (2.8,4.3) node {{\tiny $a$}};
\Hpoint{0.5}{2};
\draw (0.3,2.3) node {{\tiny $b$}};
\Hpoint{4}{3.5};
\draw (4.2,3.75) node {{\tiny $x$}};
\Hpoint{3.5}{1};
\draw (3.2,0.75) node {{\tiny $y$}};
\zoneRG{1}{4}{1}
\end{tikzpicture}
\end{center}
\end{prop}

\begin{pf}
The color of every point $\sigma_k$ such that $k < i_{RG}$ or $\sigma_k > \sigma_{j_{RG}}$ is determined by Theorem~\ref{thm:RGIncreasing} (see the first diagram of Figure~\ref{fig:tikz2RG}). We denote by $*$ the zone where the color of the points is unknown.
By maximality of $i_{RG}$ and minimality of $\sigma_{j_{GR}}$ we know the color of some other points, and as no point can be both in \R and in \G, we know that the zone between $\sigma_{i_{RG}}$ and $\sigma_{j_{RG}}$ must be empty, as shown in the second diagram of Figure~\ref{fig:tikz2RG}.

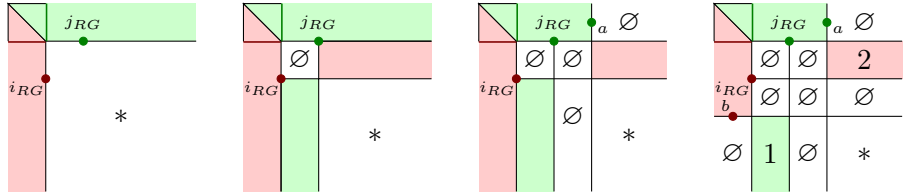
\begin{figure}[H]
\begin{center}
\begin{tikzpicture}[scale=.5]
\useasboundingbox (0,0) (6,5);
\fill [Hfill] (0,0) rectangle (1,4);
\fill [Vfill] (1,4) rectangle (5,5);
\draw (0,4) -- (5,4);
\draw (1,0) -- (1,5);
\Hpoint{1}{3};
\draw (0.5,2.75) node {{\tiny $i_{RG}$}};
\Vpoint{2}{4};
\draw (2,4.5) node {{\tiny $j_{RG}$}};
\draw (3,2) node {$*$};
\zoneRG{1}{4}{1}
\end{tikzpicture}
\begin{tikzpicture}[scale=.5]
\useasboundingbox (0,0) (6,5);
\fill [Hfill] (0,0) rectangle (1,4);
\fill [Vfill] (1,4) rectangle (5,5);
\draw [Hfill] (5,3) -- (2,3) -- (2,4) -- (5,4);
\draw [Vfill] (1,0) -- (1,3) -- (2,3) -- (2,0);
\draw (0,4) -- (5,4);
\draw (1,0) -- (1,5);
\draw (1.5,3.5) node {$\varnothing$};
\Hpoint{1}{3};
\draw (0.5,2.75) node {{\tiny $i_{RG}$}};
\Vpoint{2}{4};
\draw (2,4.5) node {{\tiny $j_{RG}$}};
\draw (3.5,1.5) node {$*$};
\zoneRG{1}{4}{1}
\end{tikzpicture}
\begin{tikzpicture}[scale=.5]
\useasboundingbox (0,0) (6,5);
\fill [Hfill] (0,0) rectangle (1,4);
\fill [Vfill] (1,4) rectangle (3,5);
\fill [Hfill] (3,3) rectangle (5,4);
\fill [Vfill] (1,0) rectangle (2,3);
\draw (0,4) -- (5,4);
\draw (1,3) -- (5,3);
\draw (1,0) -- (1,5);
\draw (2,0) -- (2,4);
\draw (3,0) -- (3,5);
\draw (1.5,3.5) node {$\varnothing$};
\draw (2.5,3.5) node {$\varnothing$};
\draw (2.5,2) node {$\varnothing$};
\draw (4,4.5) node {$\varnothing$};
\Hpoint{1}{3};
\draw (0.5,2.75) node {{\tiny $i_{RG}$}};
\Vpoint{2}{4};
\draw (2,4.5) node {{\tiny $j_{RG}$}};
\Vpoint{3}{4.5};
\draw (3.3,4.3) node {{\tiny $a$}};
\draw (4,1.5) node {$*$};
\zoneRG{1}{4}{1}
\end{tikzpicture}
\begin{tikzpicture}[scale=.5]
\useasboundingbox (0,0) (5,5);
\fill [Hfill] (0,2) rectangle (1,4);
\fill [Vfill] (1,4) rectangle (3,5);
\fill [Hfill] (3,3) rectangle (5,4);
\fill [Vfill] (1,0) rectangle (2,2);
\draw (0,4) -- (5,4);
\draw (1,3) -- (5,3);
\draw (0,2) -- (5,2);
\draw (1,0) -- (1,5);
\draw (2,0) -- (2,4);
\draw (3,0) -- (3,5);
\draw (2.5,3.5) node {$\varnothing$};
\draw (2.5,2.5) node {$\varnothing$};
\draw (0.5,1) node {$\varnothing$};
\draw (2.5,1) node {$\varnothing$};
\draw (1.5,2.5) node {$\varnothing$};
\draw (1.5,3.5) node {$\varnothing$};
\draw (4,2.5) node {$\varnothing$};
\draw (4,4.5) node {$\varnothing$};
\Hpoint{1}{3};
\draw (0.5,2.75) node {{\tiny $i_{RG}$}};
\Vpoint{2}{4};
\draw (2,4.5) node {{\tiny $j_{RG}$}};
\Vpoint{3}{4.5};
\draw (3.3,4.3) node {{\tiny $a$}};
\Hpoint{0.5}{2};
\draw (0.3,2.3) node {{\tiny $b$}};
\draw (1.5,1) node {$1$};
\draw (4,3.5) node {$2$};
\draw (4,1) node {$*$};
\zoneRG{1}{4}{1}
\end{tikzpicture}

\caption{All bicolored increasing sequences are labeled $RG$\label{fig:tikz2RG}}
\end{center}
\end{figure}

Let $a$ be the rightmost point among points above $j_{RG}$ (maybe $a = j_{RG}$). Rule (\rmnum{1}) applied to points $j_{RG}$ and $a$ gives the third diagram of Figure~\ref{fig:tikz2RG} (note that if $a=j_{RG}$ the column between $j_{RG}$ and $a$ does not exist).
Similarly let $b$ be the lowest point among points to the left of $i_{RG}$ ($b$ may be equal to $i_{RG}$). Rule (\rmnum{2}) applied to $b$ and $i_{RG}$ leads to the fourth diagram of Figure~\ref{fig:tikz2RG}. Note also that we numbered two specific zones in this diagram and we study now the different cases where they are empty or not.

\paragraph{Zone $1$ is non-empty}

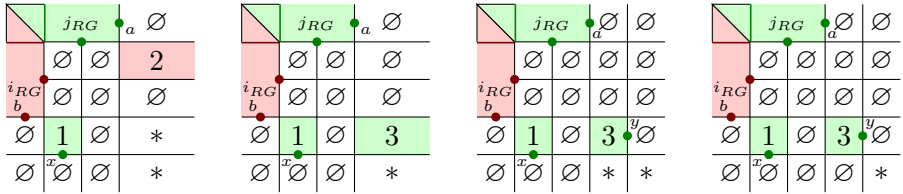
\begin{figure}[H]
\begin{center}
\begin{tikzpicture}[scale=.5]
\useasboundingbox (0,-0.5) (6,6);
\fill [Hfill] (0,2) rectangle (1,4);
\fill [Vfill] (1,4) rectangle (3,5);
\fill [Hfill] (3,3) rectangle (5,4);
\fill [Vfill] (1,1) rectangle (2,2);
\draw (0,4) -- (5,4);
\draw (1,3) -- (5,3);
\draw (0,2) -- (5,2);
\draw (0,1) -- (5,1);
\draw (1,0) -- (1,5);
\draw (2,0) -- (2,4);
\draw (3,0) -- (3,5);
\draw (0.5,0.5) node {$\varnothing$};
\draw (1.5,0.5) node {$\varnothing$};
\draw (2.5,0.5) node {$\varnothing$};
\draw (2.5,3.5) node {$\varnothing$};
\draw (2.5,2.5) node {$\varnothing$};
\draw (0.5,1.5) node {$\varnothing$};
\draw (2.5,1.5) node {$\varnothing$};
\draw (1.5,2.5) node {$\varnothing$};
\draw (1.5,3.5) node {$\varnothing$};
\draw (4,2.5) node {$\varnothing$};
\draw (4,4.5) node {$\varnothing$};
\Hpoint{1}{3};
\draw (0.5,2.75) node {{\tiny $i_{RG}$}};
\Vpoint{2}{4};
\draw (2,4.5) node {{\tiny $j_{RG}$}};
\Vpoint{3}{4.5};
\draw (3.3,4.3) node {{\tiny $a$}};
\Hpoint{0.5}{2};
\draw (0.3,2.3) node {{\tiny $b$}};
\Vpoint{1.5}{1};
\draw (1.2,0.8) node {{\tiny $x$}};
\draw (1.5,1.5) node {$1$};
\draw (4,3.5) node {$2$};
\draw (4,0.5) node {$*$};
\draw (4,1.5) node {$*$};
\zoneRG{1}{4}{1}
\end{tikzpicture}
\begin{tikzpicture}[scale=.5]
\useasboundingbox (0,-0.5) (6,6);
\fill [Hfill] (0,2) rectangle (1,4);
\fill [Vfill] (1,4) rectangle (3,5);
\fill [Vfill] (1,1) rectangle (2,2);
\fill [Vfill] (3,1) rectangle (5,2);
\draw (0,4) -- (5,4);
\draw (1,3) -- (5,3);
\draw (0,2) -- (5,2);
\draw (0,1) -- (5,1);
\draw (1,0) -- (1,5);
\draw (2,0) -- (2,4);
\draw (3,0) -- (3,5);
\draw (0.5,0.5) node {$\varnothing$};
\draw (1.5,0.5) node {$\varnothing$};
\draw (2.5,0.5) node {$\varnothing$};
\draw (2.5,3.5) node {$\varnothing$};
\draw (2.5,2.5) node {$\varnothing$};
\draw (0.5,1.5) node {$\varnothing$};
\draw (2.5,1.5) node {$\varnothing$};
\draw (1.5,2.5) node {$\varnothing$};
\draw (1.5,3.5) node {$\varnothing$};
\draw (4,2.5) node {$\varnothing$};
\draw (4,4.5) node {$\varnothing$};
\Hpoint{1}{3};
\draw (0.5,2.75) node {{\tiny $i_{RG}$}};
\Vpoint{2}{4};
\draw (2,4.5) node {{\tiny $j_{RG}$}};
\Vpoint{3}{4.5};
\draw (3.3,4.3) node {{\tiny $a$}};
\Hpoint{0.5}{2};
\draw (0.3,2.3) node {{\tiny $b$}};
\Vpoint{1.5}{1};
\draw (1.2,0.8) node {{\tiny $x$}};
\draw (1.5,1.5) node {$1$};
\draw (4,3.5) node {$\varnothing$};
\draw (4,1.5) node {$3$};
\draw (4,0.5) node {$*$};
\zoneRG{1}{4}{1}
\end{tikzpicture}
\begin{tikzpicture}[scale=.5]
\useasboundingbox (0,-0.5) (6,6);
\fill [Hfill] (0,2) rectangle (1,4);
\fill [Vfill] (1,4) rectangle (3,5);
\fill [Vfill] (1,1) rectangle (2,2);
\fill [Vfill] (3,1) rectangle (4,2);
\draw (0,4) -- (5,4);
\draw (1,3) -- (5,3);
\draw (0,2) -- (5,2);
\draw (0,1) -- (5,1);
\draw (1,0) -- (1,5);
\draw (2,0) -- (2,4);
\draw (3,0) -- (3,5);
\draw (4,0) -- (4,5);
\draw (0.5,0.5) node {$\varnothing$};
\draw (1.5,0.5) node {$\varnothing$};
\draw (2.5,0.5) node {$\varnothing$};
\draw (2.5,3.5) node {$\varnothing$};
\draw (2.5,2.5) node {$\varnothing$};
\draw (0.5,1.5) node {$\varnothing$};
\draw (2.5,1.5) node {$\varnothing$};
\draw (1.5,2.5) node {$\varnothing$};
\draw (1.5,3.5) node {$\varnothing$};
\draw (3.5,2.5) node {$\varnothing$};
\draw (3.5,3.5) node {$\varnothing$};
\draw (3.5,4.5) node {$\varnothing$};
\draw (4.5,1.5) node {$\varnothing$};
\draw (4.5,2.5) node {$\varnothing$};
\draw (4.5,3.5) node {$\varnothing$};
\draw (4.5,4.5) node {$\varnothing$};
\Hpoint{1}{3};
\draw (0.5,2.75) node {{\tiny $i_{RG}$}};
\Vpoint{2}{4};
\draw (2,4.5) node {{\tiny $j_{RG}$}};
\Vpoint{3}{4.5};
\draw (3.2,4.2) node {{\tiny $a$}};
\Hpoint{0.5}{2};
\draw (0.3,2.3) node {{\tiny $b$}};
\Vpoint{1.5}{1};
\draw (1.2,0.8) node {{\tiny $x$}};
\Vpoint{4}{1.5};
\draw (4.2,1.8) node {{\tiny $y$}};
\draw (1.5,1.5) node {$1$};
\draw (3.5,1.5) node {$3$};
\draw (3.5,0.5) node {$*$};
\draw (4.5,0.5) node {$*$};
\zoneRG{1}{4}{1}
\end{tikzpicture}
\begin{tikzpicture}[scale=.5]
\useasboundingbox (0,-0.5) (6,6);
\fill [Hfill] (0,2) rectangle (1,4);
\fill [Vfill] (1,4) rectangle (3,5);
\fill [Vfill] (1,1) rectangle (2,2);
\fill [Vfill] (3,1) rectangle (4,2);
\draw (0,4) -- (5,4);
\draw (1,3) -- (5,3);
\draw (0,2) -- (5,2);
\draw (0,1) -- (5,1);
\draw (1,0) -- (1,5);
\draw (2,0) -- (2,4);
\draw (3,0) -- (3,5);
\draw (4,0) -- (4,5);
\draw (1.5,0.5) node {$\varnothing$};
\draw (2.5,0.5) node {$\varnothing$};
\draw (2.5,3.5) node {$\varnothing$};
\draw (2.5,2.5) node {$\varnothing$};
\draw (0.5,1.5) node {$\varnothing$};
\draw (2.5,1.5) node {$\varnothing$};
\draw (1.5,2.5) node {$\varnothing$};
\draw (1.5,3.5) node {$\varnothing$};
\draw (3.5,2.5) node {$\varnothing$};
\draw (3.5,3.5) node {$\varnothing$};
\draw (3.5,4.5) node {$\varnothing$};
\draw (4.5,1.5) node {$\varnothing$};
\draw (4.5,2.5) node {$\varnothing$};
\draw (4.5,3.5) node {$\varnothing$};
\draw (4.5,4.5) node {$\varnothing$};
\Hpoint{1}{3};
\draw (0.5,2.75) node {{\tiny $i_{RG}$}};
\Vpoint{2}{4};
\draw (2,4.5) node {{\tiny $j_{RG}$}};
\Vpoint{3}{4.5};
\draw (3.2,4.2) node {{\tiny $a$}};
\Hpoint{0.5}{2};
\draw (0.3,2.3) node {{\tiny $b$}};
\Vpoint{1.5}{1};
\draw (1.2,0.8) node {{\tiny $x$}};
\Vpoint{4}{1.5};
\draw (4.2,1.8) node {{\tiny $y$}};
\draw (1.5,1.5) node {$1$};
\draw (3.5,1.5) node {$3$};
\draw (3.5,0.5) node {$\varnothing$};
\draw (0.5,0.5) node {$\varnothing$};
\draw (4.5,0.5) node {$*$};
\zoneRG{1}{4}{1}
\end{tikzpicture}

\caption{Zone $1$ is non-empty\label{fig:tikz2RG1nonVide}}
\end{center}
\end{figure}

If zone $1$ is non-empty, let $x$ be the lowest point inside this zone (see Figure~\ref{fig:tikz2RG1nonVide}).
As there do not exist an increasing sequence GR, every point to the top-right of $x$ is in \G as shown in the second diagram of Figure~\ref{fig:tikz2RG1nonVide}, where we define a zone $3$.
If zone $3$ is empty then zone $*$ is empty as $\sigma$ is $\ominus$-indecomposable, hence every point has a assigned color as in the first diagram of Proposition~\ref{prop:diagrammesRG}.
If zone $3$ is non empty, let $y$ be the rightmost point inside this zone as shown in the third diagram.
Applying rule (\rmnum{1}) to $x$ and $y$ add another empty zone, leading to the last diagram.
As $\sigma$ is $\ominus$-indecomposable, zone $*$ is empty and all points have an assigned color as in the second diagram of Proposition~\ref{prop:diagrammesRG}.

\paragraph{Zone $1$ is empty}

Suppose that zone $1$ is empty. If zone $2$ is also empty then as $\sigma$ is $\ominus$-indecomposable, zone $*$ is also empty and all points have a determined color, as in the third diagram of Proposition~\ref{prop:diagrammesRG}.

If zone $2$ is non-empty, let $x$ be the rightmost point of zone $2$.

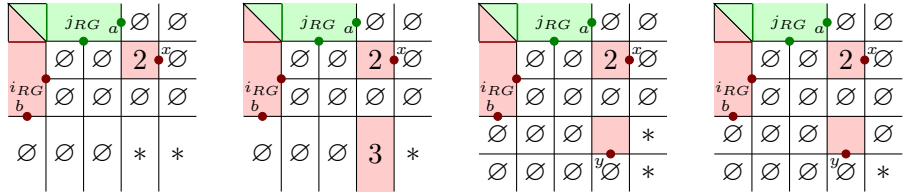
\begin{figure}[H]
\begin{center}
\begin{tikzpicture}[scale=.5]
\useasboundingbox (0,-0.5) (6,6);
\fill [Hfill] (0,2) rectangle (1,4);
\fill [Vfill] (1,4) rectangle (3,5);
\fill [Hfill] (3,3) rectangle (4,4);
\draw (0,4) -- (5,4);
\draw (1,3) -- (5,3);
\draw (0,2) -- (5,2);
\draw (1,0) -- (1,5);
\draw (2,0) -- (2,4);
\draw (3,0) -- (3,5);
\draw (4,0) -- (4,5);
\draw (2.5,3.5) node {$\varnothing$};
\draw (2.5,2.5) node {$\varnothing$};
\draw (0.5,1) node {$\varnothing$};
\draw (1.5,1) node {$\varnothing$};
\draw (2.5,1) node {$\varnothing$};
\draw (1.5,2.5) node {$\varnothing$};
\draw (1.5,3.5) node {$\varnothing$};
\draw (3.5,2.5) node {$\varnothing$};
\draw (3.5,4.5) node {$\varnothing$};
\draw (4.5,2.5) node {$\varnothing$};
\draw (4.5,3.5) node {$\varnothing$};
\draw (4.5,4.5) node {$\varnothing$};
\Hpoint{1}{3};
\draw (0.5,2.75) node {{\tiny $i_{RG}$}};
\Vpoint{2}{4};
\draw (2,4.5) node {{\tiny $j_{RG}$}};
\Vpoint{3}{4.5};
\draw (2.8,4.3) node {{\tiny $a$}};
\Hpoint{0.5}{2};
\draw (0.3,2.3) node {{\tiny $b$}};
\Hpoint{4}{3.5};
\draw (4.2,3.75) node {{\tiny $x$}};
\draw (3.5,3.5) node {$2$};
\draw (3.5,1) node {$*$};
\draw (4.5,1) node {$*$};
\zoneRG{1}{4}{1}
\end{tikzpicture}
\begin{tikzpicture}[scale=.5]
\useasboundingbox (0,-0.5) (6,6);
\fill [Hfill] (0,2) rectangle (1,4);
\fill [Vfill] (1,4) rectangle (3,5);
\fill [Hfill] (3,3) rectangle (4,4);
\fill [Hfill] (3,0) rectangle (4,2);
\draw (0,4) -- (5,4);
\draw (1,3) -- (5,3);
\draw (0,2) -- (5,2);
\draw (1,0) -- (1,5);
\draw (2,0) -- (2,4);
\draw (3,0) -- (3,5);
\draw (4,0) -- (4,5);
\draw (2.5,3.5) node {$\varnothing$};
\draw (2.5,2.5) node {$\varnothing$};
\draw (0.5,1) node {$\varnothing$};
\draw (1.5,1) node {$\varnothing$};
\draw (2.5,1) node {$\varnothing$};
\draw (1.5,2.5) node {$\varnothing$};
\draw (1.5,3.5) node {$\varnothing$};
\draw (3.5,2.5) node {$\varnothing$};
\draw (3.5,4.5) node {$\varnothing$};
\draw (4.5,2.5) node {$\varnothing$};
\draw (4.5,3.5) node {$\varnothing$};
\draw (4.5,4.5) node {$\varnothing$};
\Hpoint{1}{3};
\draw (0.5,2.75) node {{\tiny $i_{RG}$}};
\Vpoint{2}{4};
\draw (2,4.5) node {{\tiny $j_{RG}$}};
\Vpoint{3}{4.5};
\draw (2.8,4.3) node {{\tiny $a$}};
\Hpoint{0.5}{2};
\draw (0.3,2.3) node {{\tiny $b$}};
\Hpoint{4}{3.5};
\draw (4.2,3.75) node {{\tiny $x$}};
\draw (3.5,3.5) node {$2$};
\draw (3.5,1) node {$3$};
\draw (4.5,1) node {$*$};
\zoneRG{1}{4}{1}
\end{tikzpicture}
\begin{tikzpicture}[scale=.5]
\useasboundingbox (0,-0.5) (6,6);
\fill [Hfill] (0,2) rectangle (1,4);
\fill [Vfill] (1,4) rectangle (3,5);
\fill [Hfill] (3,3) rectangle (4,4);
\fill [Hfill] (3,1) rectangle (4,2);
\draw (0,4) -- (5,4);
\draw (1,3) -- (5,3);
\draw (0,2) -- (5,2);
\draw (0,1) -- (5,1);
\draw (1,0) -- (1,5);
\draw (2,0) -- (2,4);
\draw (3,0) -- (3,5);
\draw (4,0) -- (4,5);
\draw (2.5,3.5) node {$\varnothing$};
\draw (2.5,2.5) node {$\varnothing$};
\draw (0.5,0.5) node {$\varnothing$};
\draw (1.5,0.5) node {$\varnothing$};
\draw (2.5,0.5) node {$\varnothing$};
\draw (3.5,0.5) node {$\varnothing$};
\draw (0.5,1.5) node {$\varnothing$};
\draw (1.5,1.5) node {$\varnothing$};
\draw (2.5,1.5) node {$\varnothing$};
\draw (1.5,2.5) node {$\varnothing$};
\draw (1.5,3.5) node {$\varnothing$};
\draw (3.5,2.5) node {$\varnothing$};
\draw (3.5,4.5) node {$\varnothing$};
\draw (4.5,2.5) node {$\varnothing$};
\draw (4.5,3.5) node {$\varnothing$};
\draw (4.5,4.5) node {$\varnothing$};
\Hpoint{1}{3};
\draw (0.5,2.75) node {{\tiny $i_{RG}$}};
\Vpoint{2}{4};
\draw (2,4.5) node {{\tiny $j_{RG}$}};
\Vpoint{3}{4.5};
\draw (2.8,4.3) node {{\tiny $a$}};
\Hpoint{0.5}{2};
\draw (0.3,2.3) node {{\tiny $b$}};
\Hpoint{4}{3.5};
\draw (4.2,3.75) node {{\tiny $x$}};
\Hpoint{3.5}{1};
\draw (3.2,0.75) node {{\tiny $y$}};
\draw (3.5,3.5) node {$2$};
\draw (4.5,0.5) node {$*$};
\draw (4.5,1.5) node {$*$};
\zoneRG{1}{4}{1}
\end{tikzpicture}
\begin{tikzpicture}[scale=.5]
\useasboundingbox (0,-0.5) (6,6);
\fill [Hfill] (0,2) rectangle (1,4);
\fill [Vfill] (1,4) rectangle (3,5);
\fill [Hfill] (3,3) rectangle (4,4);
\fill [Hfill] (3,1) rectangle (4,2);
\draw (0,4) -- (5,4);
\draw (1,3) -- (5,3);
\draw (0,2) -- (5,2);
\draw (0,1) -- (5,1);
\draw (1,0) -- (1,5);
\draw (2,0) -- (2,4);
\draw (3,0) -- (3,5);
\draw (4,0) -- (4,5);
\draw (2.5,3.5) node {$\varnothing$};
\draw (2.5,2.5) node {$\varnothing$};
\draw (0.5,0.5) node {$\varnothing$};
\draw (1.5,0.5) node {$\varnothing$};
\draw (2.5,0.5) node {$\varnothing$};
\draw (3.5,0.5) node {$\varnothing$};
\draw (0.5,1.5) node {$\varnothing$};
\draw (1.5,1.5) node {$\varnothing$};
\draw (2.5,1.5) node {$\varnothing$};
\draw (1.5,2.5) node {$\varnothing$};
\draw (1.5,3.5) node {$\varnothing$};
\draw (3.5,2.5) node {$\varnothing$};
\draw (3.5,4.5) node {$\varnothing$};
\draw (4.5,2.5) node {$\varnothing$};
\draw (4.5,3.5) node {$\varnothing$};
\draw (4.5,4.5) node {$\varnothing$};
\Hpoint{1}{3};
\draw (0.5,2.75) node {{\tiny $i_{RG}$}};
\Vpoint{2}{4};
\draw (2,4.5) node {{\tiny $j_{RG}$}};
\Vpoint{3}{4.5};
\draw (2.8,4.3) node {{\tiny $a$}};
\Hpoint{0.5}{2};
\draw (0.3,2.3) node {{\tiny $b$}};
\Hpoint{4}{3.5};
\draw (4.2,3.75) node {{\tiny $x$}};
\Hpoint{3.5}{1};
\draw (3.2,0.75) node {{\tiny $y$}};
\draw (3.5,3.5) node {$2$};
\draw (4.5,0.5) node {$*$};
\draw (4.5,1.5) node {$\varnothing$};
\zoneRG{1}{4}{1}
\end{tikzpicture}
\caption{Zone $1$ is empty\label{fig:tikz2RG1Vide}}
\end{center}
\end{figure}

As there are no increasing subsequence GR, all points in the lower left quadrant of $x$ lie in \R as shown in the second diagram of Figure~\ref{fig:tikz2RG1Vide} where we define a zone $3$.
If zone $3$ is empty then as $\sigma$ is $\ominus$-indecomposable zone $*$ is also empty and all points have a determined color, as in the fourth diagram of Proposition~\ref{prop:diagrammesRG}.
Otherwise zone $3$ is non-empty and let $y$ the lowest point in zone $3$ as depicted in the third diagram.
We apply rule (\rmnum{2}) to $x$ and $y$ leading to the fourth diagram.
As $\sigma$ is $\ominus$-indecomposable, zone $*$ is empty and all points have a determined color, as in the last diagram of Proposition~\ref{prop:diagrammesRG}.
\end{pf}

\begin{defn}\label{def:C_RG}
Let $\sigma$ be a permutation and $i$ and $j$ two indices of $\sigma$ such that $\sigma_i \sigma_j$ is an ascent.
Set $a = \max \{k \mid \sigma_k \geq \sigma_j\}$ and $b$ such that $\sigma_b = \min \{\sigma_k \mid k \leq i\}$.
We define $C_{RG}(\sigma, i, j)$ as the partial bicoloring of $\sigma$ having the following shape:

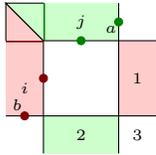
\begin{tikzpicture}[scale=.5]
\useasboundingbox (0,0.5) (5,6);
\fill [Hfill] (0,2) rectangle (1,4);
\fill [Vfill] (1,4) rectangle (3,5);
\fill [Hfill] (3,2) rectangle (4,4);
\fill [Vfill] (1,1) rectangle (3,2);
\draw (0,4) -- (4,4);
\draw (0,2) -- (4,2);
\draw (1,1) -- (1,5);
\draw (3,1) -- (3,5);
\Hpoint{1}{3};
\draw (0.5,2.75) node {{\tiny $i$}};
\Vpoint{2}{4};
\draw (2,4.5) node {{\tiny $j$}};
\Vpoint{3}{4.5};
\draw (2.8,4.3) node {{\tiny $a$}};
\Hpoint{0.5}{2};
\draw (0.3,2.3) node {{\tiny $b$}};
\draw (3.5,3) node {{\tiny $1$}};
\draw (2,1.5) node {{\tiny $2$}};
\draw (3.5,1.5) node {{\tiny $3$}};
\zoneRG{1}{4}{1}
% mettre ou pas les bordures du diagramme ?
\end{tikzpicture}
where points of zone $3$ are in \G if zone $1$ is empty and zone $2$ is nonempty, 
in \R if zone $1$ is nonempty and zone $2$ is empty, and have no color otherwise.
\end{defn}

\begin{prop}\label{prop:C_RG}
Let $\sigma$ be a $\ominus$-indecomposable permutation and $c$ a valid coloring of $\sigma$ such that there is no increasing subsequence RG in $c$ and there is at least an increasing sequence GR in $c$. 
Then $c = C_{RG}(\sigma, i_{RG}, j_{RG})$.
\end{prop}

\begin{pf}
This is a direct consequence of Proposition~\ref{prop:diagrammesRG} and Definition~\ref{def:C_RG}.
\end{pf}

\subsubsection{There exist both increasing sequences labeled $GR$ and $RG$}

In this section we study the last case that remains to deal, 
{\em i.e.} there is at least one increasing sequence colored $RG$ and at least one colored $GR$.
As $A_{GR}$ and $A_{RG}$ are non empty, $i_{GR}$, $j_{GR}$, $i_{RG}$ and $j_{RG}$ are defined.
We prove that once $i_{GR}$, $j_{GR}$, $i_{RG}$ and $j_{RG}$ are determined, then it fixes the color of every other point of the permutation.

% In this last part, we study colorings $C$ of permutations $\sigma$ which have at least one increasing sequence colored $RG$ and one colored $GR$.
% 
% This indeed correspond to the first case of Proposition~\ref{prop:decomposition}.
% 
% In that case, let define the four points $i_{RG}$, $j_{RG}$, $i_{GR}$ and $j_{GR}$ as seen before. 
% To complete our study we unfortunately have to consider $4$ different relative positions of these $4$ points. 

\begin{prop}\label{prop:diagrammes*}
Let $\sigma$ be a permutation and $c$ a valid coloring of $\sigma$ such that
there exists at least an increasing sequence colored $GR$ and at least an increasing sequence colored $RG$.
Then $c$ has one of the following shapes:

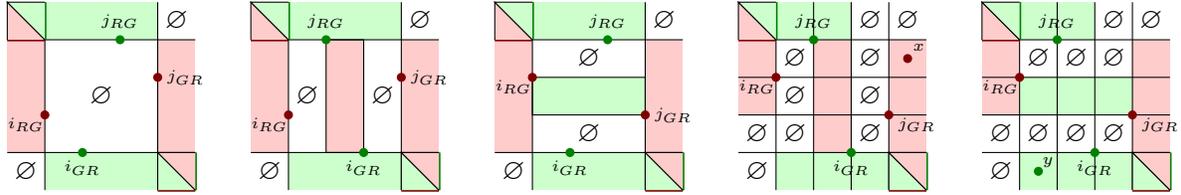
\begin{tikzpicture}[scale=.5]%4
\useasboundingbox (1,-1) (6.25,6);
\fill [Hfill] (0,1) rectangle (1,4);
\fill [Vfill] (1,4) rectangle (4,5);
\fill [Hfill] (4,1) rectangle (5,4);
\fill [Vfill] (1,0) rectangle (4,1);
\draw (0,1) -- (5,1);
\draw (0,4) -- (5,4);
\draw (1,0) -- (1,5);
\draw (4,0) -- (4,5);
\draw (2.5,2.5) node {$\varnothing$};
\draw (0.5,0.5) node {$\varnothing$};
\draw (4.5,4.5) node {$\varnothing$};
\Hpoint{1}{2};
\draw (0.5,1.75) node {{\tiny $i_{RG}$}};
\Vpoint{3}{4};
\draw (3,4.5) node {{\tiny $j_{RG}$}};
\Vpoint{2}{1};
\draw (2,0.6) node {{\tiny $i_{GR}$}};
\Hpoint{4}{3};
\draw (4.75,3) node {{\tiny $j_{GR}$}};
\zoneRG{1}{4}{1}
\zoneGR{5}{0}{1}

\end{tikzpicture}
\begin{tikzpicture}[scale=.5]%5
\useasboundingbox (0,-1) (6.25,6);
\fill [Hfill] (0,1) rectangle (1,4);
\fill [Vfill] (1,4) rectangle (4,5);
\fill [Hfill] (4,1) rectangle (5,4);
\fill [Vfill] (1,0) rectangle (4,1);
\draw [Hfill] (2,1) rectangle (3,4);
\draw (0,1) -- (5,1);
\draw (0,4) -- (5,4);
\draw (1,0) -- (1,5);
\draw (4,0) -- (4,5);
\draw (1.5,2.5) node {$\varnothing$};
\draw (3.5,2.5) node {$\varnothing$};
\draw (0.5,0.5) node {$\varnothing$};
\draw (4.5,4.5) node {$\varnothing$};
\Hpoint{1}{2};
\draw (0.5,1.75) node {{\tiny $i_{RG}$}};
\Vpoint{2}{4};
\draw (2,4.5) node {{\tiny $j_{RG}$}};
\Vpoint{3}{1};
\draw (3,0.6) node {{\tiny $i_{GR}$}};
\Hpoint{4}{3};
\draw (4.75,3) node {{\tiny $j_{GR}$}};
\zoneRG{1}{4}{1}
\zoneGR{5}{0}{1}

\end{tikzpicture}
\begin{tikzpicture}[scale=.5]%3
\useasboundingbox (0,-1) (6.25,6);
\fill [Hfill] (0,1) rectangle (1,4);
\fill [Vfill] (1,4) rectangle (4,5);
\fill [Hfill] (4,1) rectangle (5,4);
\fill [Vfill] (1,0) rectangle (4,1);
\draw [Vfill] (1,2) rectangle (4,3);
\draw (0,1) -- (5,1);
\draw (0,4) -- (5,4);
\draw (1,0) -- (1,5);
\draw (4,0) -- (4,5);
\draw (2.5,3.5) node {$\varnothing$};
\draw (2.5,1.5) node {$\varnothing$};
\draw (0.5,0.5) node {$\varnothing$};
\draw (4.5,4.5) node {$\varnothing$};
\Hpoint{1}{3};
\draw (0.5,2.75) node {{\tiny $i_{RG}$}};
\Vpoint{3}{4};
\draw (3,4.5) node {{\tiny $j_{RG}$}};
\Vpoint{2}{1};
\draw (2,0.6) node {{\tiny $i_{GR}$}};
\Hpoint{4}{2};
\draw (4.75,2) node {{\tiny $j_{GR}$}};
\zoneRG{1}{4}{1}
\zoneGR{5}{0}{1}

\end{tikzpicture}
\begin{tikzpicture}[scale=.5]%1
\useasboundingbox (0,-1) (6.25,6);
\fill [Hfill] (0,2) rectangle (1,4);
\fill [Vfill] (1,4) rectangle (3,5);
\fill [Hfill] (4,1) rectangle (5,4);
\fill [Vfill] (1,0) rectangle (4,1);
\fill [Hfill] (2,1) rectangle (3,4);
\draw (0,1) -- (5,1);
\draw (0,2) -- (5,2);
\draw (0,3) -- (5,3);
\draw (0,4) -- (5,4);
\draw (1,0) -- (1,5);
\draw (2,0) -- (2,5);
\draw (3,0) -- (3,5);
\draw (4,0) -- (4,5);
\draw (1.5,1.5) node {$\varnothing$};
\draw (1.5,3.5) node {$\varnothing$};
\draw (3.5,3.5) node {$\varnothing$};
\draw (3.5,1.5) node {$\varnothing$};
\draw (0.5,0.5) node {$\varnothing$};
\draw (0.5,1.5) node {$\varnothing$};
\draw (3.5,4.5) node {$\varnothing$};
\draw (4.5,4.5) node {$\varnothing$};
\draw (1.5,2.5) node {$\varnothing$};
\draw (3.5,2.5) node {$\varnothing$};
\Hpoint{1}{3};
\draw (0.5,2.75) node {{\tiny $i_{RG}$}};
\Vpoint{2}{4};
\draw (2,4.5) node {{\tiny $j_{RG}$}};
\Vpoint{3}{1};
\draw (3,0.6) node {{\tiny $i_{GR}$}};
\Hpoint{4}{2};
\draw (4.75,1.75) node {{\tiny $j_{GR}$}};
\Hpoint{4.5}{3.5};
\draw (4.8,3.8) node {{\tiny $x$}};
\zoneRG{1}{4}{1}
\zoneGR{5}{0}{1}
\end{tikzpicture}
\begin{tikzpicture}[scale=.5]%2
\useasboundingbox (0,-1) (5,6);
\fill [Hfill] (0,2) rectangle (1,4);
\fill [Vfill] (1,4) rectangle (3,5);
\fill [Hfill] (4,1) rectangle (5,4);
\fill [Vfill] (1,0) rectangle (4,1);
\fill [Vfill] (1,2) rectangle (4,3);
\draw (0,1) -- (5,1);
\draw (0,2) -- (5,2);
\draw (0,3) -- (5,3);
\draw (0,4) -- (5,4);
\draw (1,0) -- (1,5);
\draw (2,0) -- (2,5);
\draw (3,0) -- (3,5);
\draw (4,0) -- (4,5);
\draw (1.5,1.5) node {$\varnothing$};
\draw (1.5,3.5) node {$\varnothing$};
\draw (3.5,3.5) node {$\varnothing$};
\draw (3.5,1.5) node {$\varnothing$};
\draw (0.5,0.5) node {$\varnothing$};
\draw (0.5,1.5) node {$\varnothing$};
\draw (3.5,4.5) node {$\varnothing$};
\draw (4.5,4.5) node {$\varnothing$};
\draw (2.5,1.5) node {$\varnothing$};
\draw (2.5,3.5) node {$\varnothing$};
\Hpoint{1}{3};
\draw (0.5,2.75) node {{\tiny $i_{RG}$}};
\Vpoint{2}{4};
\draw (2,4.5) node {{\tiny $j_{RG}$}};
\Vpoint{3}{1};
\draw (3,0.6) node {{\tiny $i_{GR}$}};
\Hpoint{4}{2};
\draw (4.75,1.75) node {{\tiny $j_{GR}$}};
\Vpoint{1.5}{0.5};
\draw (1.75,0.75) node {{\tiny $y$}};
\zoneRG{1}{4}{1}
\zoneGR{5}{0}{1}

\end{tikzpicture}

\end{prop}

\begin{pf}
By maximality of $i_{GR}$ and minimality of $\sigma_{j_{GR}}$ we have:

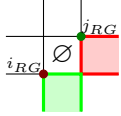
\begin{tikzpicture}[scale=.5]
\draw (4,0) -- (7,0);
\draw (4,1) -- (7,1);
\draw (5,-1) -- (5,2);
\draw (6,-1) -- (6,2);
\zoneF{H};
\zoneB{V};
\draw (5.5,0.5) node {$\varnothing$};
\Hpoint{5}{0};
\Vpoint{6}{1};
\draw (6.5,1.25) node {{\tiny $j_{RG}$}};
\draw (4.5,0.25) node {{\tiny $i_{RG}$}};
\end{tikzpicture}
By Theorem~\ref{thm:RGIncreasing} we obtain:
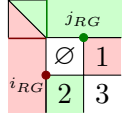
\begin{tikzpicture}[scale=.5]
%\useasboundingbox (0,-1) (7.75,6);
\fill [Hfill] (0,2) rectangle (1,4);
\fill [Vfill] (1,4) rectangle (3,5);
\draw [Hfill] (3,3) -- (2,3) -- (2,4) -- (3,4);
\draw [Vfill] (1,2) -- (1,3) -- (2,3) -- (2,2);
\draw (0,4) -- (3,4);
\draw (1,2) -- (1,5);
\draw (1.5,3.5) node {$\varnothing$};
\Hpoint{1}{3};
\draw (0.5,2.75) node {{\tiny $i_{RG}$}};
\Vpoint{2}{4};
\draw (2,4.5) node {{\tiny $j_{RG}$}};
\draw (2.5,3.5) node {$1$};
\draw (1.5,2.5) node {$2$};
\draw (2.5,2.5) node {$3$};
\zoneRG{1}{4}{1}
\end{tikzpicture}

Recall that there exist an increasing sequence $GR$. 
 $i_{GR}$ lies in quadrant $2$ or $3$ and $j_{GR}$ in quadrant $1$ or $3$. 
Hence the coloring $c$ as either one of the $4$ following shapes:

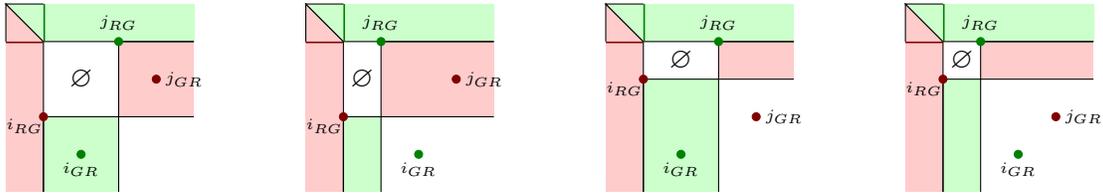
\begin{tikzpicture}[scale=.5]%3
\useasboundingbox (0,-1) (7.75,6);
\fill [Hfill] (0,0) rectangle (1,4);
\fill [Vfill] (1,4) rectangle (5,5);
\draw [Hfill] (5,2) -- (3,2) -- (3,4) -- (5,4);
\draw [Vfill] (1,0) -- (1,2) -- (3,2) -- (3,0);
\draw (0,4) -- (5,4);
\draw (1,0) -- (1,5);
\draw (2,3) node {$\varnothing$};
\Hpoint{1}{2};
\draw (0.5,1.75) node {{\tiny $i_{RG}$}};
\Vpoint{3}{4};
\draw (3,4.5) node {{\tiny $j_{RG}$}};
\Vpoint{2}{1};
\draw (2,0.6) node {{\tiny $i_{GR}$}};
\Hpoint{4}{3};
\draw (4.75,3) node {{\tiny $j_{GR}$}};
\zoneRG{1}{4}{1}
\end{tikzpicture}
\begin{tikzpicture}[scale=.5]%4
\useasboundingbox (0,-1) (7.75,6);
\fill [Hfill] (0,0) rectangle (1,4);
\fill [Vfill] (1,4) rectangle (5,5);
\draw [Hfill] (5,2) -- (2,2) -- (2,4) -- (5,4);
\draw [Vfill] (1,0) -- (1,2) -- (2,2) -- (2,0);
\draw (0,4) -- (5,4);
\draw (1,0) -- (1,5);
\draw (1.5,3) node {$\varnothing$};
\Hpoint{1}{2};
\draw (0.5,1.75) node {{\tiny $i_{RG}$}};
\Vpoint{2}{4};
\draw (2,4.5) node {{\tiny $j_{RG}$}};
\Vpoint{3}{1};
\draw (3,0.6) node {{\tiny $i_{GR}$}};
\Hpoint{4}{3};
\draw (4.75,3) node {{\tiny $j_{GR}$}};
\zoneRG{1}{4}{1}
\end{tikzpicture}
\begin{tikzpicture}[scale=.5]%2
\useasboundingbox (0,-1) (7.75,6);
\fill [Hfill] (0,0) rectangle (1,4);
\fill [Vfill] (1,4) rectangle (5,5);
\draw [Hfill] (5,3) -- (3,3) -- (3,4) -- (5,4);
\draw [Vfill] (1,0) -- (1,3) -- (3,3) -- (3,0);
\draw (0,4) -- (5,4);
\draw (1,0) -- (1,5);
\draw (2,3.5) node {$\varnothing$};
\Hpoint{1}{3};
\draw (0.5,2.75) node {{\tiny $i_{RG}$}};
\Vpoint{3}{4};
\draw (3,4.5) node {{\tiny $j_{RG}$}};
\Vpoint{2}{1};
\draw (2,0.6) node {{\tiny $i_{GR}$}};
\Hpoint{4}{2};
\draw (4.75,2) node {{\tiny $j_{GR}$}};
\zoneRG{1}{4}{1}
\end{tikzpicture}
\begin{tikzpicture}[scale=.5]%1
\useasboundingbox (0,-1) (6,6);
\fill [Hfill] (0,0) rectangle (1,4);
\fill [Vfill] (1,4) rectangle (5,5);
\draw [Hfill] (5,3) -- (2,3) -- (2,4) -- (5,4);
\draw [Vfill] (1,0) -- (1,3) -- (2,3) -- (2,0);
\draw (0,4) -- (5,4);
\draw (1,0) -- (1,5);
\draw (1.5,3.5) node {$\varnothing$};
\Hpoint{1}{3};
\draw (0.5,2.75) node {{\tiny $i_{RG}$}};
\Vpoint{2}{4};
\draw (2,4.5) node {{\tiny $j_{RG}$}};
\Vpoint{3}{1};
\draw (3,0.6) node {{\tiny $i_{GR}$}};
\Hpoint{4}{2};
\draw (4.75,2) node {{\tiny $j_{GR}$}};
\zoneRG{1}{4}{1}
\end{tikzpicture}

Applying Theorem~\ref{thm:GRIncreasing} to $i_{GR}$ and $j_{GR}$ we obtain these new diagrams:

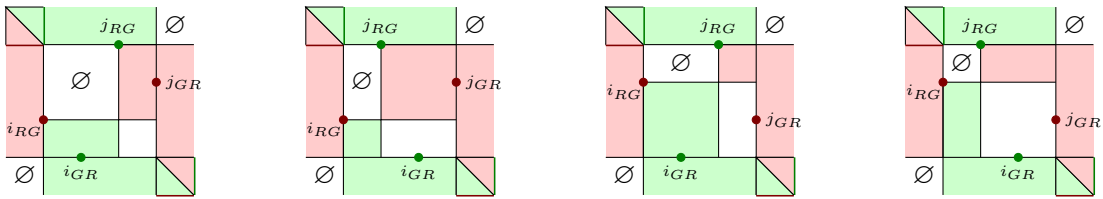
\begin{tikzpicture}[scale=.5]%3
\useasboundingbox (0,-1) (7.75,6);
\fill [Hfill] (0,1) rectangle (1,4);
\fill [Vfill] (1,4) rectangle (4,5);
\fill [Hfill] (4,1) rectangle (5,4);
\fill [Vfill] (1,0) rectangle (4,1);
\draw [Hfill] (4,2) -- (3,2) -- (3,4) -- (4,4);
\draw [Vfill] (1,1) -- (1,2) -- (3,2) -- (3,1);
\draw (0,1) -- (5,1);
\draw (0,4) -- (5,4);
\draw (1,0) -- (1,5);
\draw (4,0) -- (4,5);
\draw (2,3) node {$\varnothing$};
\draw (0.5,0.5) node {$\varnothing$};
\draw (4.5,4.5) node {$\varnothing$};
\Hpoint{1}{2};
\draw (0.5,1.75) node {{\tiny $i_{RG}$}};
\Vpoint{3}{4};
\draw (3,4.5) node {{\tiny $j_{RG}$}};
\Vpoint{2}{1};
\draw (2,0.6) node {{\tiny $i_{GR}$}};
\Hpoint{4}{3};
\draw (4.75,3) node {{\tiny $j_{GR}$}};
\zoneRG{1}{4}{1}
\zoneGR{5}{0}{1}
\end{tikzpicture}
\begin{tikzpicture}[scale=.5]%4
\useasboundingbox (0,-1) (7.75,6);
\fill [Hfill] (0,1) rectangle (1,4);
\fill [Vfill] (1,4) rectangle (4,5);
\fill [Hfill] (4,1) rectangle (5,4);
\fill [Vfill] (1,0) rectangle (4,1);
\draw [Hfill] (4,2) -- (2,2) -- (2,4) -- (4,4);
\draw [Vfill] (1,1) -- (1,2) -- (2,2) -- (2,1);
\draw (0,1) -- (5,1);
\draw (0,4) -- (5,4);
\draw (1,0) -- (1,5);
\draw (4,0) -- (4,5);
\draw (1.5,3) node {$\varnothing$};
\draw (0.5,0.5) node {$\varnothing$};
\draw (4.5,4.5) node {$\varnothing$};
\Hpoint{1}{2};
\draw (0.5,1.75) node {{\tiny $i_{RG}$}};
\Vpoint{2}{4};
\draw (2,4.5) node {{\tiny $j_{RG}$}};
\Vpoint{3}{1};
\draw (3,0.6) node {{\tiny $i_{GR}$}};
\Hpoint{4}{3};
\draw (4.75,3) node {{\tiny $j_{GR}$}};
\zoneRG{1}{4}{1}
\zoneGR{5}{0}{1}
\end{tikzpicture}
\begin{tikzpicture}[scale=.5]%2
\useasboundingbox (0,-1) (7.75,6);
\fill [Hfill] (0,1) rectangle (1,4);
\fill [Vfill] (1,4) rectangle (4,5);
\fill [Hfill] (4,1) rectangle (5,4);
\fill [Vfill] (1,0) rectangle (4,1);
\draw [Hfill] (4,3) -- (3,3) -- (3,4) -- (4,4);
\draw [Vfill] (1,1) -- (1,3) -- (3,3) -- (3,1);
\draw (0,1) -- (5,1);
\draw (0,4) -- (5,4);
\draw (1,0) -- (1,5);
\draw (4,0) -- (4,5);
\draw (2,3.5) node {$\varnothing$};
\draw (0.5,0.5) node {$\varnothing$};
\draw (4.5,4.5) node {$\varnothing$};
\Hpoint{1}{3};
\draw (0.5,2.75) node {{\tiny $i_{RG}$}};
\Vpoint{3}{4};
\draw (3,4.5) node {{\tiny $j_{RG}$}};
\Vpoint{2}{1};
\draw (2,0.6) node {{\tiny $i_{GR}$}};
\Hpoint{4}{2};
\draw (4.75,2) node {{\tiny $j_{GR}$}};
\zoneRG{1}{4}{1}
\zoneGR{5}{0}{1}
\end{tikzpicture}
\begin{tikzpicture}[scale=.5]%1
\useasboundingbox (0,-1) (6,6);
\fill [Hfill] (0,1) rectangle (1,4);
\fill [Vfill] (1,4) rectangle (4,5);
\fill [Hfill] (4,1) rectangle (5,4);
\fill [Vfill] (1,0) rectangle (4,1);
\draw [Hfill] (4,3) -- (2,3) -- (2,4) -- (4,4);
\draw [Vfill] (1,1) -- (1,3) -- (2,3) -- (2,1);
\draw (0,1) -- (5,1);
\draw (0,4) -- (5,4);
\draw (1,0) -- (1,5);
\draw (4,0) -- (4,5);
\draw (1.5,3.5) node {$\varnothing$};
\draw (0.5,0.5) node {$\varnothing$};
\draw (4.5,4.5) node {$\varnothing$};
\Hpoint{1}{3};
\draw (0.5,2.75) node {{\tiny $i_{RG}$}};
\Vpoint{2}{4};
\draw (2,4.5) node {{\tiny $j_{RG}$}};
\Vpoint{3}{1};
\draw (3,0.6) node {{\tiny $i_{GR}$}};
\Hpoint{4}{2};
\draw (4.75,2) node {{\tiny $j_{GR}$}};
\zoneRG{1}{4}{1}
\zoneGR{5}{0}{1}
\end{tikzpicture}

Finally, using maximality of $\sigma_{i_{GR}}$ and minimality of $j_{GR}$, we obtain:

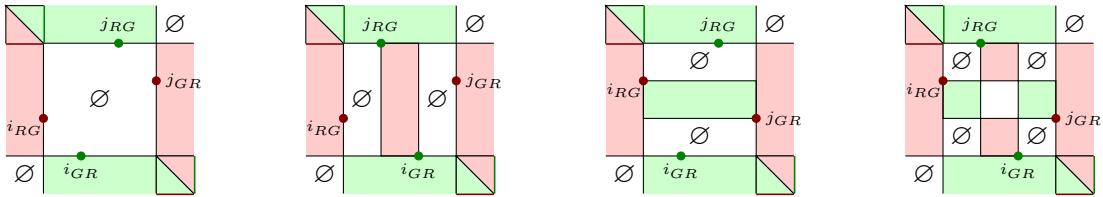
\begin{tikzpicture}[scale=.5]%3
\useasboundingbox (0,-1) (7.75,6);
\fill [Hfill] (0,1) rectangle (1,4);
\fill [Vfill] (1,4) rectangle (4,5);
\fill [Hfill] (4,1) rectangle (5,4);
\fill [Vfill] (1,0) rectangle (4,1);
\draw (0,1) -- (5,1);
\draw (0,4) -- (5,4);
\draw (1,0) -- (1,5);
\draw (4,0) -- (4,5);
\draw (2.5,2.5) node {$\varnothing$};
\draw (0.5,0.5) node {$\varnothing$};
\draw (4.5,4.5) node {$\varnothing$};
\Hpoint{1}{2};
\draw (0.5,1.75) node {{\tiny $i_{RG}$}};
\Vpoint{3}{4};
\draw (3,4.5) node {{\tiny $j_{RG}$}};
\Vpoint{2}{1};
\draw (2,0.6) node {{\tiny $i_{GR}$}};
\Hpoint{4}{3};
\draw (4.75,3) node {{\tiny $j_{GR}$}};
\zoneGR{5}{0}{1}
\zoneRG{1}{4}{1}

\end{tikzpicture}
\begin{tikzpicture}[scale=.5]%4
\useasboundingbox (0,-1) (7.75,6);
\fill [Hfill] (0,1) rectangle (1,4);
\fill [Vfill] (1,4) rectangle (4,5);
\fill [Hfill] (4,1) rectangle (5,4);
\fill [Vfill] (1,0) rectangle (4,1);
\draw [Hfill] (2,1) rectangle (3,4);
\draw (0,1) -- (5,1);
\draw (0,4) -- (5,4);
\draw (1,0) -- (1,5);
\draw (4,0) -- (4,5);
\draw (1.5,2.5) node {$\varnothing$};
\draw (3.5,2.5) node {$\varnothing$};
\draw (0.5,0.5) node {$\varnothing$};
\draw (4.5,4.5) node {$\varnothing$};
\Hpoint{1}{2};
\draw (0.5,1.75) node {{\tiny $i_{RG}$}};
\Vpoint{2}{4};
\draw (2,4.5) node {{\tiny $j_{RG}$}};
\Vpoint{3}{1};
\draw (3,0.6) node {{\tiny $i_{GR}$}};
\Hpoint{4}{3};
\draw (4.75,3) node {{\tiny $j_{GR}$}};
\zoneGR{5}{0}{1}
\zoneRG{1}{4}{1}

\end{tikzpicture}
\begin{tikzpicture}[scale=.5]%2
\useasboundingbox (0,-1) (7.75,6);
\fill [Hfill] (0,1) rectangle (1,4);
\fill [Vfill] (1,4) rectangle (4,5);
\fill [Hfill] (4,1) rectangle (5,4);
\fill [Vfill] (1,0) rectangle (4,1);
\draw [Vfill] (1,2) rectangle (4,3);
\draw (0,1) -- (5,1);
\draw (0,4) -- (5,4);
\draw (1,0) -- (1,5);
\draw (4,0) -- (4,5);
\draw (2.5,3.5) node {$\varnothing$};
\draw (2.5,1.5) node {$\varnothing$};
\draw (0.5,0.5) node {$\varnothing$};
\draw (4.5,4.5) node {$\varnothing$};
\Hpoint{1}{3};
\draw (0.5,2.75) node {{\tiny $i_{RG}$}};
\Vpoint{3}{4};
\draw (3,4.5) node {{\tiny $j_{RG}$}};
\Vpoint{2}{1};
\draw (2,0.6) node {{\tiny $i_{GR}$}};
\Hpoint{4}{2};
\draw (4.75,2) node {{\tiny $j_{GR}$}};
\zoneGR{5}{0}{1}
\zoneRG{1}{4}{1}

\end{tikzpicture}
\begin{tikzpicture}[scale=.5]%1
\useasboundingbox (0,-1) (6,6);
\fill [Hfill] (0,1) rectangle (1,4);
\fill [Vfill] (1,4) rectangle (4,5);
\fill [Hfill] (4,1) rectangle (5,4);
\fill [Vfill] (1,0) rectangle (4,1);
\draw [Hfill] (2,1) rectangle (3,2);
\draw [Vfill] (1,2) rectangle (2,3);
\draw [Hfill] (2,3) rectangle (3,4);
\draw [Vfill] (3,2) rectangle (4,3);
\draw (0,1) -- (5,1);
\draw (0,4) -- (5,4);
\draw (1,0) -- (1,5);
\draw (4,0) -- (4,5);
\draw (1.5,1.5) node {$\varnothing$};
\draw (1.5,3.5) node {$\varnothing$};
\draw (3.5,3.5) node {$\varnothing$};
\draw (3.5,1.5) node {$\varnothing$};
\draw (0.5,0.5) node {$\varnothing$};
\draw (4.5,4.5) node {$\varnothing$};
\Hpoint{1}{3};
\draw (0.5,2.75) node {{\tiny $i_{RG}$}};
\Vpoint{2}{4};
\draw (2,4.5) node {{\tiny $j_{RG}$}};
\Vpoint{3}{1};
\draw (3,0.6) node {{\tiny $i_{GR}$}};
\Hpoint{4}{2};
\draw (4.75,2) node {{\tiny $j_{GR}$}};
\zoneGR{5}{0}{1}
\zoneRG{1}{4}{1}
\end{tikzpicture}
In the first $3$ diagrams, the color of each point is determined 
-- recall that upper-left and lower-right points are determined by Theorems~\ref{thm:RGIncreasing} and \ref{thm:GRIncreasing} -- 
and only depend on $i_{RG}$, $j_{RG}$, $i_{GR}$ and $j_{GR}$.

This leaves us with the last diagram of Figure~\ref{fig:RGGR1} for which we have again to consider several cases. 
Note that in this diagram we named several zones whose emptiness is relevant and we denote once more the unknown zone by $*$.

\begin{figure}[H]
\begin{center}
\begin{tikzpicture}[scale=.5]
\useasboundingbox (0,-0.5) (7.75,6);
\fill [Hfill] (0,1) rectangle (1,4);
\fill [Vfill] (1,4) rectangle (4,5);
\fill [Hfill] (4,1) rectangle (5,4);
\fill [Vfill] (1,0) rectangle (4,1);
\fill [Hfill] (2,1) rectangle (3,2);
\fill [Vfill] (1,2) rectangle (2,3);
\fill [Hfill] (2,3) rectangle (3,4);
\fill [Vfill] (3,2) rectangle (4,3);
\draw (0,1) -- (5,1);
\draw (0,2) -- (5,2);
\draw (0,3) -- (5,3);
\draw (0,4) -- (5,4);
\draw (1,0) -- (1,5);
\draw (2,0) -- (2,5);
\draw (3,0) -- (3,5);
\draw (4,0) -- (4,5);
\draw (1.5,1.5) node {$\varnothing$};
\draw (1.5,3.5) node {$\varnothing$};
\draw (3.5,3.5) node {$\varnothing$};
\draw (3.5,1.5) node {$\varnothing$};
\draw (0.5,0.5) node {$\varnothing$};
\draw (4.5,4.5) node {$\varnothing$};
\Hpoint{1}{3};
\draw (0.5,2.75) node {{\tiny $i_{RG}$}};
\Vpoint{2}{4};
\draw (2,4.5) node {{\tiny $j_{RG}$}};
\Vpoint{3}{1};
\draw (3,0.6) node {{\tiny $i_{GR}$}};
\Hpoint{4}{2};
\draw (4.75,1.75) node {{\tiny $j_{GR}$}};
\draw (4.5,3.5) node {$A$};
\draw (1.5,0.5) node {$B$};
\draw (0.5,1.5) node {$C$};
\draw (3.5,4.5) node {$D$};
\draw (2.5,2.5) node {$*$};
\zoneGR{5}{0}{1}
\zoneRG{1}{4}{1}

\end{tikzpicture}
\begin{tikzpicture}[scale=.5]
\useasboundingbox (0,-0.5) (7.75,6);
\fill [Hfill] (0,2) rectangle (1,4);
\fill [Vfill] (1,4) rectangle (3,5);
\fill [Hfill] (4,1) rectangle (5,4);
\fill [Vfill] (1,0) rectangle (4,1);
\fill [Hfill] (2,1) rectangle (3,2);
\fill [Vfill] (1,2) rectangle (2,3);
\fill [Hfill] (2,3) rectangle (3,4);
\fill [Vfill] (3,2) rectangle (4,3);
\draw (0,1) -- (5,1);
\draw (0,2) -- (5,2);
\draw (0,3) -- (5,3);
\draw (0,4) -- (5,4);
\draw (1,0) -- (1,5);
\draw (2,0) -- (2,5);
\draw (3,0) -- (3,5);
\draw (4,0) -- (4,5);
\draw (1.5,1.5) node {$\varnothing$};
\draw (1.5,3.5) node {$\varnothing$};
\draw (3.5,3.5) node {$\varnothing$};
\draw (3.5,1.5) node {$\varnothing$};
\draw (0.5,0.5) node {$\varnothing$};
\draw (4.5,4.5) node {$\varnothing$};
\Hpoint{1}{3};
\draw (0.5,2.75) node {{\tiny $i_{RG}$}};
\Vpoint{2}{4};
\draw (2,4.5) node {{\tiny $j_{RG}$}};
\Vpoint{3}{1};
\draw (3,0.6) node {{\tiny $i_{GR}$}};
\Hpoint{4}{2};
\draw (4.75,1.75) node {{\tiny $j_{GR}$}};
\draw (4.5,3.5) node {$A$};
\draw (1.5,0.5) node {$B$};
\draw (0.5,1.5) node {$\varnothing$};
\draw (3.5,4.5) node {$\varnothing$};
\draw (2.5,2.5) node {$*$};
\zoneGR{5}{0}{1}
\zoneRG{1}{4}{1}

\end{tikzpicture}
\begin{tikzpicture}[scale=.5]
\useasboundingbox (0,-0.5) (7.75,6);
\fill [Hfill] (0,2) rectangle (1,4);
\fill [Vfill] (1,4) rectangle (3,5);
\fill [Hfill] (4,1) rectangle (5,4);
\fill [Vfill] (1,0) rectangle (4,1);
\fill [Hfill] (2,1) rectangle (3,4);
\draw (0,1) -- (5,1);
\draw (0,2) -- (5,2);
\draw (0,3) -- (5,3);
\draw (0,4) -- (5,4);
\draw (1,0) -- (1,5);
\draw (2,0) -- (2,5);
\draw (3,0) -- (3,5);
\draw (4,0) -- (4,5);
\draw (1.5,1.5) node {$\varnothing$};
\draw (1.5,3.5) node {$\varnothing$};
\draw (3.5,3.5) node {$\varnothing$};
\draw (3.5,1.5) node {$\varnothing$};
\draw (0.5,0.5) node {$\varnothing$};
\draw (0.5,1.5) node {$\varnothing$};
\draw (3.5,4.5) node {$\varnothing$};
\draw (4.5,4.5) node {$\varnothing$};
\draw (1.5,2.5) node {$\varnothing$};
\draw (3.5,2.5) node {$\varnothing$};
\Hpoint{1}{3};
\draw (0.5,2.75) node {{\tiny $i_{RG}$}};
\Vpoint{2}{4};
\draw (2,4.5) node {{\tiny $j_{RG}$}};
\Vpoint{3}{1};
\draw (3,0.6) node {{\tiny $i_{GR}$}};
\Hpoint{4}{2};
\draw (4.75,1.75) node {{\tiny $j_{GR}$}};
\Hpoint{4.5}{3.5};
\draw (4.8,3.8) node {{\tiny $x$}};
\zoneGR{5}{0}{1}
\zoneRG{1}{4}{1}

\end{tikzpicture}
\begin{tikzpicture}[scale=.5]
\useasboundingbox (0,-0.5) (6,6);
\fill [Hfill] (0,2) rectangle (1,4);
\fill [Vfill] (1,4) rectangle (3,5);
\fill [Hfill] (4,1) rectangle (5,4);
\fill [Vfill] (1,0) rectangle (4,1);
\fill [Vfill] (1,2) rectangle (4,3);
\draw (0,1) -- (5,1);
\draw (0,2) -- (5,2);
\draw (0,3) -- (5,3);
\draw (0,4) -- (5,4);
\draw (1,0) -- (1,5);
\draw (2,0) -- (2,5);
\draw (3,0) -- (3,5);
\draw (4,0) -- (4,5);
\draw (1.5,1.5) node {$\varnothing$};
\draw (1.5,3.5) node {$\varnothing$};
\draw (3.5,3.5) node {$\varnothing$};
\draw (3.5,1.5) node {$\varnothing$};
\draw (0.5,0.5) node {$\varnothing$};
\draw (0.5,1.5) node {$\varnothing$};
\draw (3.5,4.5) node {$\varnothing$};
\draw (4.5,4.5) node {$\varnothing$};
\draw (2.5,1.5) node {$\varnothing$};
\draw (2.5,3.5) node {$\varnothing$};
\Hpoint{1}{3};
\draw (0.5,2.75) node {{\tiny $i_{RG}$}};
\Vpoint{2}{4};
\draw (2,4.5) node {{\tiny $j_{RG}$}};
\Vpoint{3}{1};
\draw (3,0.6) node {{\tiny $i_{GR}$}};
\Hpoint{4}{2};
\draw (4.75,1.75) node {{\tiny $j_{GR}$}};
\Vpoint{1.5}{0.5};
\draw (1.75,0.75) node {{\tiny $y$}};
\zoneGR{5}{0}{1}
\zoneRG{1}{4}{1}

\end{tikzpicture}

\caption{There exist increasing sequences labeled RG and GR\label{fig:RGGR1}}
\end{center}
\end{figure}
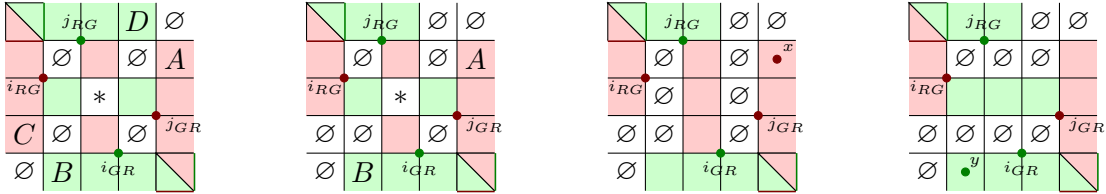

Applying rule (\rmnum{7}) to $i_{RG}$ and $j_{GR}$ implies that zone $C$ is empty. 
Similary rule (\rmnum{8}) applied to $j_{GR}$ and $i_{GR}$ proves that zone $D$ is empty.
If there exists a point $x$ in zone $A$, then applying rule (\rmnum{2}) to $j_{GR}$ and $x$, all points in $*$ are determined 
-- they lie in \R -- as shown in the third diagram. 
Symetrically, if there exists a point $y$ in $B$ then applying rule (\rmnum{1}) to $y$ and $i_{GR}$, all points in $*$ should be in \G -- see diagram $4$ --.

Thus this leaves us with the case where both $A$ and $B$ are empty. 
We show that this case is not possible.

\paragraph{$A$ and $B$ are empty}

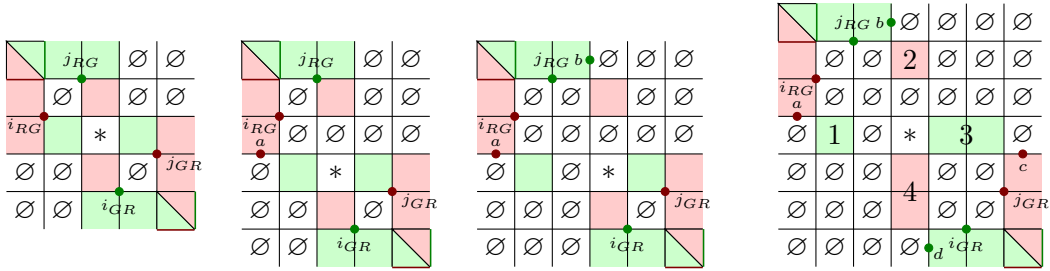
\begin{figure}[H]
\begin{center}
\begin{tikzpicture}[scale=.5]
\useasboundingbox (0,-1.5) (6,6);
\fill [Hfill] (0,2) rectangle (1,4);
\fill [Vfill] (1,4) rectangle (3,5);
\fill [Hfill] (4,1) rectangle (5,3);
\fill [Vfill] (2,0) rectangle (4,1);
\fill [Hfill] (2,1) rectangle (3,2);
\fill [Vfill] (1,2) rectangle (2,3);
\fill [Hfill] (2,3) rectangle (3,4);
\fill [Vfill] (3,2) rectangle (4,3);
\draw (0,1) -- (5,1);
\draw (0,2) -- (5,2);
\draw (0,3) -- (5,3);
\draw (0,4) -- (5,4);
\draw (1,0) -- (1,5);
\draw (2,0) -- (2,5);
\draw (3,0) -- (3,5);
\draw (4,0) -- (4,5);
\draw (1.5,1.5) node {$\varnothing$};
\draw (1.5,3.5) node {$\varnothing$};
\draw (3.5,3.5) node {$\varnothing$};
\draw (3.5,1.5) node {$\varnothing$};
\draw (0.5,0.5) node {$\varnothing$};
\draw (4.5,4.5) node {$\varnothing$};
\Hpoint{1}{3};
\draw (0.5,2.75) node {{\tiny $i_{RG}$}};
\Vpoint{2}{4};
\draw (2,4.5) node {{\tiny $j_{RG}$}};
\Vpoint{3}{1};
\draw (3,0.6) node {{\tiny $i_{GR}$}};
\Hpoint{4}{2};
\draw (4.75,1.75) node {{\tiny $j_{GR}$}};
\draw (4.5,3.5) node {$\varnothing$};
\draw (1.5,0.5) node {$\varnothing$};
\draw (0.5,1.5) node {$\varnothing$};
\draw (3.5,4.5) node {$\varnothing$};
\draw (2.5,2.5) node {$*$};
\zoneRG{1}{4}{1}
\zoneGR{5}{0}{1}
\end{tikzpicture}
\begin{tikzpicture}[scale=.5]
\useasboundingbox (0,-0.5) (6,6);
\fill [Hfill] (0,3) rectangle (1,5);
\fill [Vfill] (1,5) rectangle (3,6);
\fill [Hfill] (4,1) rectangle (5,3);
\fill [Vfill] (2,0) rectangle (4,1);
\fill [Hfill] (2,1) rectangle (3,2);
\fill [Vfill] (1,2) rectangle (2,3);
\fill [Hfill] (2,4) rectangle (3,5);
\fill [Vfill] (3,2) rectangle (4,3);
\draw (0,1) -- (5,1);
\draw (0,2) -- (5,2);
\draw (0,3) -- (5,3);
\draw (0,4) -- (5,4);
\draw (0,5) -- (5,5);
\draw (1,0) -- (1,6);
\draw (2,0) -- (2,6);
\draw (3,0) -- (3,6);
\draw (4,0) -- (4,6);
\draw (1.5,1.5) node {$\varnothing$};
\draw (1.5,4.5) node {$\varnothing$};
\draw (3.5,4.5) node {$\varnothing$};
\draw (3.5,1.5) node {$\varnothing$};
\draw (0.5,0.5) node {$\varnothing$};
\draw (4.5,5.5) node {$\varnothing$};
\Hpoint{1}{4};
\draw (0.5,3.75) node {{\tiny $i_{RG}$}};
\Hpoint{0.5}{3};
\draw (0.5,3.3) node {{\tiny $a$}};
\Vpoint{2}{5};
\draw (2,5.5) node {{\tiny $j_{RG}$}};
\Vpoint{3}{1};
\draw (3,0.6) node {{\tiny $i_{GR}$}};
\Hpoint{4}{2};
\draw (4.75,1.75) node {{\tiny $j_{GR}$}};
\draw (4.5,4.5) node {$\varnothing$};
\draw (1.5,0.5) node {$\varnothing$};
\draw (0.5,1.5) node {$\varnothing$};
\draw (3.5,5.5) node {$\varnothing$};
\draw (2.5,2.5) node {$*$};
\draw (1.5,3.5) node {$\varnothing$};
\draw (2.5,3.5) node {$\varnothing$};
\draw (3.5,3.5) node {$\varnothing$};
\draw (4.5,3.5) node {$\varnothing$};
\draw (0.5,2.5) node {$\varnothing$};
\zoneRG{1}{5}{1}
\zoneGR{5}{0}{1}
\end{tikzpicture}
\begin{tikzpicture}[scale=.5]
\useasboundingbox (0,-0.5) (7.75,6.5);
\fill [Hfill] (0,3) rectangle (1,5);
\fill [Vfill] (1,5) rectangle (3,6);
\fill [Hfill] (5,1) rectangle (6,3);
\fill [Vfill] (3,0) rectangle (5,1);
\fill [Hfill] (3,1) rectangle (4,2);
\fill [Vfill] (1,2) rectangle (2,3);
\fill [Hfill] (3,4) rectangle (4,5);
\fill [Vfill] (4,2) rectangle (5,3);
\draw (0,1) -- (6,1);
\draw (0,2) -- (6,2);
\draw (0,3) -- (6,3);
\draw (0,4) -- (6,4);
\draw (0,5) -- (6,5);
\draw (1,0) -- (1,6);
\draw (2,0) -- (2,6);
\draw (3,0) -- (3,6);
\draw (4,0) -- (4,6);
\draw (5,0) -- (5,6);
\draw (1.5,1.5) node {$\varnothing$};
\draw (1.5,4.5) node {$\varnothing$};
\draw (4.5,4.5) node {$\varnothing$};
\draw (4.5,1.5) node {$\varnothing$};
\draw (0.5,0.5) node {$\varnothing$};
\draw (5.5,5.5) node {$\varnothing$};
\Hpoint{1}{4};
\draw (0.5,3.75) node {{\tiny $i_{RG}$}};
\Hpoint{0.5}{3};
\draw (0.5,3.3) node {{\tiny $a$}};
\Vpoint{2}{5};
\draw (2,5.5) node {{\tiny $j_{RG}$}};
\Vpoint{3}{5.5};
\draw (2.7,5.5) node {{\tiny $b$}};
\Vpoint{4}{1};
\draw (4,0.6) node {{\tiny $i_{GR}$}};
\Hpoint{5}{2};
\draw (5.75,1.75) node {{\tiny $j_{GR}$}};
\draw (5.5,4.5) node {$\varnothing$};
\draw (1.5,0.5) node {$\varnothing$};
\draw (0.5,1.5) node {$\varnothing$};
\draw (4.5,5.5) node {$\varnothing$};
\draw (3.5,2.5) node {$*$};
\draw (1.5,3.5) node {$\varnothing$};
\draw (3.5,3.5) node {$\varnothing$};
\draw (4.5,3.5) node {$\varnothing$};
\draw (5.5,3.5) node {$\varnothing$};
\draw (0.5,2.5) node {$\varnothing$};
\draw (2.5,0.5) node {$\varnothing$};
\draw (2.5,1.5) node {$\varnothing$};
\draw (2.5,2.5) node {$\varnothing$};
\draw (2.5,3.5) node {$\varnothing$};
\draw (2.5,4.5) node {$\varnothing$};
\draw (3.5,5.5) node {$\varnothing$};
\zoneRG{1}{5}{1}
\zoneGR{6}{0}{1}
\end{tikzpicture}
\begin{tikzpicture}[scale=.5]
\useasboundingbox (0,-0.5) (7,7.5);
\fill [Hfill] (0,4) rectangle (1,6);
\fill [Vfill] (1,6) rectangle (3,7);
\fill [Hfill] (6,1) rectangle (7,3);
\fill [Vfill] (4,0) rectangle (6,1);
\fill [Hfill] (3,1) rectangle (4,3);
\fill [Vfill] (1,3) rectangle (2,4);
\fill [Hfill] (3,5) rectangle (4,6);
\fill [Vfill] (4,3) rectangle (6,4);
\draw (0,1) -- (7,1);
\draw (0,2) -- (7,2);
\draw (0,3) -- (7,3);
\draw (0,4) -- (7,4);
\draw (0,5) -- (7,5);
\draw (0,6) -- (7,6);
\draw (1,0) -- (1,7);
\draw (2,0) -- (2,7);
\draw (3,0) -- (3,7);
\draw (4,0) -- (4,7);
\draw (5,0) -- (5,7);
\draw (6,0) -- (6,7);
\draw (1.5,1.5) node {$\varnothing$};
\draw (1.5,5.5) node {$\varnothing$};
\draw (5.5,5.5) node {$\varnothing$};
\draw (5.5,1.5) node {$\varnothing$};
\draw (0.5,0.5) node {$\varnothing$};
\draw (6.5,6.5) node {$\varnothing$};
\Hpoint{1}{5};
\draw (0.5,4.75) node {{\tiny $i_{RG}$}};
\Hpoint{0.5}{4};
\draw (0.5,4.3) node {{\tiny $a$}};
\Vpoint{2}{6};
\draw (2,6.5) node {{\tiny $j_{RG}$}};
\Vpoint{3}{6.5};
\draw (2.7,6.5) node {{\tiny $b$}};
\Vpoint{5}{1};
\draw (5,0.6) node {{\tiny $i_{GR}$}};
\Vpoint{4}{0.5};
\draw (4.25,0.35) node {{\tiny $d$}};
\Hpoint{6}{2};
\draw (6.75,1.75) node {{\tiny $j_{GR}$}};
\Hpoint{6.5}{3};
\draw (6.5,2.65) node {{\tiny $c$}};
\draw (6.5,5.5) node {$\varnothing$};
\draw (1.5,0.5) node {$\varnothing$};
\draw (3.5,0.5) node {$\varnothing$};
\draw (0.5,1.5) node {$\varnothing$};
\draw (5.5,6.5) node {$\varnothing$};
\draw (3.5,3.5) node {$*$};
\draw (1.5,4.5) node {$\varnothing$};
\draw (3.5,4.5) node {$\varnothing$};
\draw (5.5,4.5) node {$\varnothing$};
\draw (6.5,4.5) node {$\varnothing$};
\draw (0.5,3.5) node {$\varnothing$};
\draw (2.5,0.5) node {$\varnothing$};
\draw (2.5,1.5) node {$\varnothing$};
\draw (2.5,3.5) node {$\varnothing$};
\draw (2.5,4.5) node {$\varnothing$};
\draw (2.5,5.5) node {$\varnothing$};
\draw (6.5,3.5) node {$\varnothing$};
\draw (3.5,6.5) node {$\varnothing$};
\draw (0.5,2.5) node {$\varnothing$};
\draw (1.5,2.5) node {$\varnothing$};
\draw (2.5,2.5) node {$\varnothing$};
\draw (5.5,2.5) node {$\varnothing$};
\draw (4.5,1.5) node {$\varnothing$};
\draw (4.5,2.5) node {$\varnothing$};
\draw (4.5,4.5) node {$\varnothing$};
\draw (4.5,5.5) node {$\varnothing$};
\draw (4.5,6.5) node {$\varnothing$};
\draw (1.5,3.5) node {$1$};
\draw (3.5,5.5) node {$2$};
\draw (5,3.5) node {$3$};
\draw (3.5,2) node {$4$};
\zoneRG{1}{6}{1}
\zoneGR{7}{0}{1}
\end{tikzpicture}
\caption{$A$ and $B$ are empty\label{fig:ABempty}}
\end{center}
\end{figure}
Then the permutation is colored as shown in the first diagram of Figure~\ref{fig:ABempty}.
Let $a$ be the lowest point among points to the left of $i_{RG}$ ($a$ may be equal to $i_{RG}$). 
Rule (\rmnum{2}) applied to $a$ and $i_{RG}$ implies the coloring shown in the second diagram 
-- notice that if $a=i_{RG}$, the line between $a$ and $i_{GR}$ does not exist --. 
Similarly, define $b$ as the rightmost point among points above $j_{GR}$ ($b$ may be equal to $j_{GR}$). 
Rule (\rmnum{1}) applied to $b$ and $j_{RG}$ leads to the third diagram.
At last we consider the topmost point $c$ among points to the right of $j_{GR}$ (maybe $c=j_{GR}$) and we apply rule (\rmnum{2}) to $c$ and $j_{GR}$. 
We also introduce $d$ as the leftmost point among points below $i_{GR}$ (maybe $d=i_{GR}$). 
Rule (\rmnum{1}) applied to $d$ and $i_{GR}$ leads to the last diagram where different zones are numbered.
We now study different cases according whether zone $1$ is empty or not, and we prove that both are excluded.

\paragraph{Zone $1$ is empty}
Suppose that zone $1$ is empty. As $\sigma$ is $\ominus$-indecomposable then zone $2$ must contain at least one point. Denote by $x$ the rightmost point of this zone. Figure~\ref{fig:RGGR1empty} illustrates the proof.
\begin{figure}[H]
\begin{center}
\begin{tikzpicture}[scale=.5]
\useasboundingbox (0,-0.5) (8.75,7.5);
\fill [Hfill] (0,4) rectangle (1,6);
\fill [Vfill] (1,6) rectangle (3,7);
\fill [Hfill] (7,1) rectangle (8,3);
\fill [Vfill] (5,0) rectangle (7,1);
\fill [Hfill] (3,1) rectangle (5,3);
\fill [Hfill] (3,5) rectangle (4,6);
\fill [Vfill] (5,3) rectangle (7,4);
\draw (0,1) -- (8,1);
\draw (0,2) -- (8,2);
\draw (0,3) -- (8,3);
\draw (0,4) -- (8,4);
\draw (0,5) -- (8,5);
\draw (0,6) -- (8,6);
\draw (1,0) -- (1,7);
\draw (2,0) -- (2,7);
\draw (3,0) -- (3,7);
\draw (4,0) -- (4,7);
\draw (5,0) -- (5,7);
\draw (6,0) -- (6,7);
\draw (7,0) -- (7,7);
\draw (1.5,1.5) node {$\varnothing$};
\draw (1.5,3.5) node {$\varnothing$};
\draw (1.5,5.5) node {$\varnothing$};
\draw (6.5,5.5) node {$\varnothing$};
\draw (6.5,1.5) node {$\varnothing$};
\draw (0.5,0.5) node {$\varnothing$};
\draw (7.5,6.5) node {$\varnothing$};
\Hpoint{1}{5};
\draw (0.5,4.75) node {{\tiny $i_{RG}$}};
\Hpoint{0.5}{4};
\draw (0.5,4.3) node {{\tiny $a$}};
\Vpoint{2}{6};
\draw (2,6.5) node {{\tiny $j_{RG}$}};
\Vpoint{3}{6.5};
\draw (2.7,6.5) node {{\tiny $b$}};
\Vpoint{6}{1};
\draw (6,0.6) node {{\tiny $i_{GR}$}};
\Vpoint{5}{0.5};
\draw (5.25,0.35) node {{\tiny $d$}};
\Hpoint{7}{2};
\draw (7.75,1.75) node {{\tiny $j_{GR}$}};
\Hpoint{7.5}{3};
\draw (7.5,2.65) node {{\tiny $c$}};
\Hpoint{4}{5.5};
\draw (4.2,5.8) node {{\tiny $x$}};
\draw (7.5,5.5) node {$\varnothing$};
\draw (1.5,0.5) node {$\varnothing$};
\draw (3.5,0.5) node {$\varnothing$};
\draw (0.5,1.5) node {$\varnothing$};
\draw (6.5,6.5) node {$\varnothing$};
\draw (3.5,3.5) node {$*$};
\draw (1.5,4.5) node {$\varnothing$};
\draw (3.5,4.5) node {$\varnothing$};
\draw (6.5,4.5) node {$\varnothing$};
\draw (7.5,4.5) node {$\varnothing$};
\draw (0.5,3.5) node {$\varnothing$};
\draw (2.5,0.5) node {$\varnothing$};
\draw (2.5,1.5) node {$\varnothing$};
\draw (2.5,3.5) node {$\varnothing$};
\draw (2.5,4.5) node {$\varnothing$};
\draw (2.5,5.5) node {$\varnothing$};
\draw (7.5,3.5) node {$\varnothing$};
\draw (3.5,6.5) node {$\varnothing$};
\draw (0.5,2.5) node {$\varnothing$};
\draw (1.5,2.5) node {$\varnothing$};
\draw (2.5,2.5) node {$\varnothing$};
\draw (6.5,2.5) node {$\varnothing$};
\draw (5.5,1.5) node {$\varnothing$};
\draw (5.5,2.5) node {$\varnothing$};
\draw (5.5,4.5) node {$\varnothing$};
\draw (5.5,5.5) node {$\varnothing$};
\draw (5.5,6.5) node {$\varnothing$};
\draw (4.5,0.5) node {$\varnothing$};
\draw (4.5,3.5) node {$*$};
\draw (4.5,4.5) node {$\varnothing$};
\draw (4.5,6.5) node {$\varnothing$};
\draw (4.5,5.5) node {$\varnothing$};
\draw (3.5,5.5) node {$2$};
\draw (6,3.5) node {$3$};
\draw (4,2) node {$4$};
\zoneRG{1}{6}{1}
\zoneGR{8}{0}{1}
\end{tikzpicture}
\begin{tikzpicture}[scale=.5]
\useasboundingbox (0,-0.5) (8.75,7.5);
\fill [Hfill] (0,4) rectangle (1,6);
\fill [Vfill] (1,6) rectangle (3,7);
\fill [Hfill] (7,1) rectangle (8,3);
\fill [Vfill] (5,0) rectangle (7,1);
\fill [Hfill] (4,1) rectangle (5,3);
\fill [Hfill] (3,5) rectangle (4,6);
\fill [Vfill] (5,3) rectangle (7,4);
\draw (0,1) -- (8,1);
\draw (0,2) -- (8,2);
\draw (0,3) -- (8,3);
\draw (0,4) -- (8,4);
\draw (0,5) -- (8,5);
\draw (0,6) -- (8,6);
\draw (1,0) -- (1,7);
\draw (2,0) -- (2,7);
\draw (3,0) -- (3,7);
\draw (4,0) -- (4,7);
\draw (5,0) -- (5,7);
\draw (6,0) -- (6,7);
\draw (7,0) -- (7,7);
\draw (1.5,1.5) node {$\varnothing$};
\draw (1.5,3.5) node {$\varnothing$};
\draw (1.5,5.5) node {$\varnothing$};
\draw (6.5,5.5) node {$\varnothing$};
\draw (6.5,1.5) node {$\varnothing$};
\draw (0.5,0.5) node {$\varnothing$};
\draw (7.5,6.5) node {$\varnothing$};
\Hpoint{1}{5};
\draw (0.5,4.75) node {{\tiny $i_{RG}$}};
\Hpoint{0.5}{4};
\draw (0.5,4.3) node {{\tiny $a$}};
\Vpoint{2}{6};
\draw (2,6.5) node {{\tiny $j_{RG}$}};
\Vpoint{3}{6.5};
\draw (2.7,6.5) node {{\tiny $b$}};
\Vpoint{6}{1};
\draw (6,0.6) node {{\tiny $i_{GR}$}};
\Vpoint{5}{0.5};
\draw (5.25,0.35) node {{\tiny $d$}};
\Hpoint{7}{2};
\draw (7.75,1.75) node {{\tiny $j_{GR}$}};
\Hpoint{7.5}{3};
\draw (7.5,2.65) node {{\tiny $c$}};
\Hpoint{4}{5.5};
\draw (4.2,5.8) node {{\tiny $x$}};
\draw (7.5,5.5) node {$\varnothing$};
\draw (1.5,0.5) node {$\varnothing$};
\draw (3.5,0.5) node {$\varnothing$};
\draw (0.5,1.5) node {$\varnothing$};
\draw (6.5,6.5) node {$\varnothing$};
\draw (3.5,3.5) node {$*$};
\draw (1.5,4.5) node {$\varnothing$};
\draw (3.5,4.5) node {$\varnothing$};
\draw (6.5,4.5) node {$\varnothing$};
\draw (7.5,4.5) node {$\varnothing$};
\draw (0.5,3.5) node {$\varnothing$};
\draw (2.5,0.5) node {$\varnothing$};
\draw (2.5,1.5) node {$\varnothing$};
\draw (2.5,3.5) node {$\varnothing$};
\draw (2.5,4.5) node {$\varnothing$};
\draw (2.5,5.5) node {$\varnothing$};
\draw (7.5,3.5) node {$\varnothing$};
\draw (3.5,6.5) node {$\varnothing$};
\draw (0.5,2.5) node {$\varnothing$};
\draw (1.5,2.5) node {$\varnothing$};
\draw (2.5,2.5) node {$\varnothing$};
\draw (6.5,2.5) node {$\varnothing$};
\draw (5.5,1.5) node {$\varnothing$};
\draw (5.5,2.5) node {$\varnothing$};
\draw (5.5,4.5) node {$\varnothing$};
\draw (5.5,5.5) node {$\varnothing$};
\draw (5.5,6.5) node {$\varnothing$};
\draw (4.5,0.5) node {$\varnothing$};
\draw (4.5,3.5) node {$*$};
\draw (4.5,4.5) node {$\varnothing$};
\draw (4.5,6.5) node {$\varnothing$};
\draw (4.5,5.5) node {$\varnothing$};
\draw (3.5,1.5) node {$\varnothing$};
\draw (3.5,2.5) node {$\varnothing$};
\draw (3.5,5.5) node {$2$};
\draw (6,3.5) node {$3$};
\draw (4.5,2) node {$4$};
\zoneRG{1}{6}{1}
\zoneGR{8}{0}{1}

\end{tikzpicture}
\begin{tikzpicture}[scale=.5]
\useasboundingbox (0,-0.5) (8.75,7.5);
\fill [Hfill] (0,4) rectangle (1,6);
\fill [Vfill] (1,6) rectangle (3,7);
\fill [Hfill] (7,1) rectangle (8,3);
\fill [Vfill] (5,0) rectangle (7,1);
\fill [Hfill] (4,1) rectangle (5,3);
\fill [Hfill] (3,5) rectangle (4,6);
\fill [Vfill] (5,3) rectangle (7,4);
\fill [Hfill] (3,3) rectangle (4,4);
\draw (0,1) -- (8,1);
\draw (0,2) -- (8,2);
\draw (0,3) -- (8,3);
\draw (0,4) -- (8,4);
\draw (0,5) -- (8,5);
\draw (0,6) -- (8,6);
\draw (1,0) -- (1,7);
\draw (2,0) -- (2,7);
\draw (3,0) -- (3,7);
\draw (4,0) -- (4,7);
\draw (5,0) -- (5,7);
\draw (6,0) -- (6,7);
\draw (7,0) -- (7,7);
\draw (1.5,1.5) node {$\varnothing$};
\draw (1.5,3.5) node {$\varnothing$};
\draw (1.5,5.5) node {$\varnothing$};
\draw (6.5,5.5) node {$\varnothing$};
\draw (6.5,1.5) node {$\varnothing$};
\draw (0.5,0.5) node {$\varnothing$};
\draw (7.5,6.5) node {$\varnothing$};
\Hpoint{1}{5};
\draw (0.5,4.75) node {{\tiny $i_{RG}$}};
\Hpoint{0.5}{4};
\draw (0.5,4.3) node {{\tiny $a$}};
\Vpoint{2}{6};
\draw (2,6.5) node {{\tiny $j_{RG}$}};
\Vpoint{3}{6.5};
\draw (2.7,6.5) node {{\tiny $b$}};
\Vpoint{6}{1};
\draw (6,0.6) node {{\tiny $i_{GR}$}};
\Vpoint{5}{0.5};
\draw (5.25,0.35) node {{\tiny $d$}};
\Hpoint{7}{2};
\draw (7.75,1.75) node {{\tiny $j_{GR}$}};
\Hpoint{7.5}{3};
\draw (7.5,2.65) node {{\tiny $c$}};
\Hpoint{4}{5.5};
\draw (4.2,5.8) node {{\tiny $x$}};
\draw (7.5,5.5) node {$\varnothing$};
\draw (1.5,0.5) node {$\varnothing$};
\draw (3.5,0.5) node {$\varnothing$};
\draw (0.5,1.5) node {$\varnothing$};
\draw (6.5,6.5) node {$\varnothing$};
\draw (3.5,3.5) node {$A$};
\draw (1.5,4.5) node {$\varnothing$};
\draw (3.5,4.5) node {$\varnothing$};
\draw (6.5,4.5) node {$\varnothing$};
\draw (7.5,4.5) node {$\varnothing$};
\draw (0.5,3.5) node {$\varnothing$};
\draw (2.5,0.5) node {$\varnothing$};
\draw (2.5,1.5) node {$\varnothing$};
\draw (2.5,3.5) node {$\varnothing$};
\draw (2.5,4.5) node {$\varnothing$};
\draw (2.5,5.5) node {$\varnothing$};
\draw (7.5,3.5) node {$\varnothing$};
\draw (3.5,6.5) node {$\varnothing$};
\draw (0.5,2.5) node {$\varnothing$};
\draw (1.5,2.5) node {$\varnothing$};
\draw (2.5,2.5) node {$\varnothing$};
\draw (6.5,2.5) node {$\varnothing$};
\draw (5.5,1.5) node {$\varnothing$};
\draw (5.5,2.5) node {$\varnothing$};
\draw (5.5,4.5) node {$\varnothing$};
\draw (5.5,5.5) node {$\varnothing$};
\draw (5.5,6.5) node {$\varnothing$};
\draw (4.5,0.5) node {$\varnothing$};
\draw (4.5,3.5) node {$*$};
\draw (4.5,4.5) node {$\varnothing$};
\draw (4.5,6.5) node {$\varnothing$};
\draw (4.5,5.5) node {$\varnothing$};
\draw (3.5,1.5) node {$\varnothing$};
\draw (3.5,2.5) node {$\varnothing$};
\draw (3.5,5.5) node {$2$};
\draw (6,3.5) node {$3$};
\draw (4.5,2) node {$4$};
\zoneRG{1}{6}{1}
\zoneGR{8}{0}{1}
\end{tikzpicture}

\caption{Increasing sequences RG and GR exist and $1$ is empty.\label{fig:RGGR1empty}}
\end{center}
\end{figure}
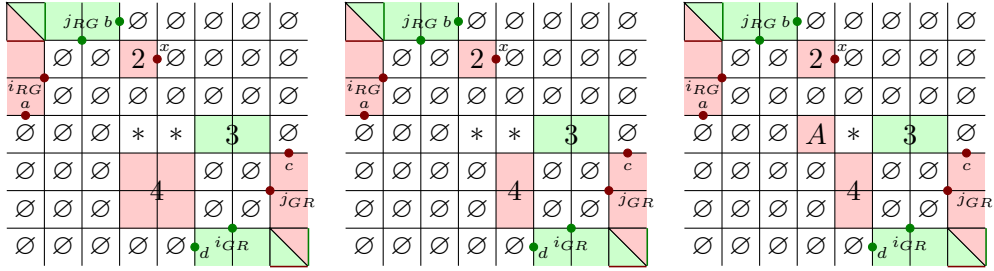

Rule (\rmnum{7}) applied to $x$ and $c$ leads to the second diagram.
Moreover as ($i_{GR}$, $j_{GR}$) is the topmost and leftmost increasing sequence GR, all points to the lower left quadrant of $x$ lie in \R,
leading to the third diagram where we define a zone $A$.

\paragraph{Zone $4$ is not empty}

We prove that this case is not possible.
If zone $4$ is not empty, let $y$ be its leftmost point (above or below $j_{GR}$) as illustrated in the first diagram of Figure~\ref{fig:zone4nonempty}. 

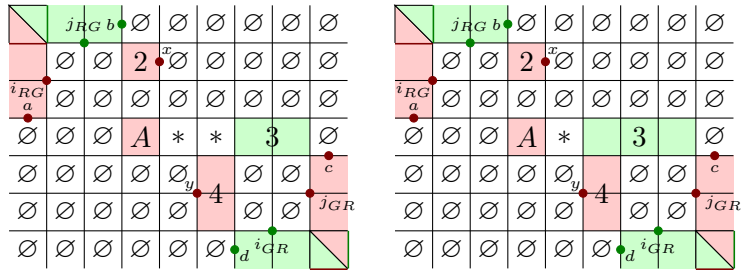
\begin{figure}[H]
\begin{center}
\begin{tikzpicture}[scale=.5]
\useasboundingbox (0,-0.5) (10,7.5);
\fill [Hfill] (0,4) rectangle (1,6);
\fill [Vfill] (1,6) rectangle (3,7);
\fill [Hfill] (8,1) rectangle (9,3);
\fill [Vfill] (6,0) rectangle (8,1);
\fill [Hfill] (5,1) rectangle (6,3);
\fill [Hfill] (3,5) rectangle (4,6);
\fill [Vfill] (6,3) rectangle (8,4);
\fill [Hfill] (3,3) rectangle (4,4);
\draw (0,1) -- (9,1);
\draw (0,2) -- (9,2);
\draw (0,3) -- (9,3);
\draw (0,4) -- (9,4);
\draw (0,5) -- (9,5);
\draw (0,6) -- (9,6);
\draw (1,0) -- (1,7);
\draw (2,0) -- (2,7);
\draw (3,0) -- (3,7);
\draw (4,0) -- (4,7);
\draw (5,0) -- (5,7);
\draw (6,0) -- (6,7);
\draw (7,0) -- (7,7);
\draw (8,0) -- (8,7);
\draw (1.5,1.5) node {$\varnothing$};
\draw (1.5,3.5) node {$\varnothing$};
\draw (1.5,5.5) node {$\varnothing$};
\draw (7.5,5.5) node {$\varnothing$};
\draw (7.5,1.5) node {$\varnothing$};
\draw (0.5,0.5) node {$\varnothing$};
\draw (8.5,6.5) node {$\varnothing$};
\Hpoint{1}{5};
\draw (0.5,4.75) node {{\tiny $i_{RG}$}};
\Hpoint{0.5}{4};
\draw (0.5,4.3) node {{\tiny $a$}};
\Vpoint{2}{6};
\draw (2,6.5) node {{\tiny $j_{RG}$}};
\Vpoint{3}{6.5};
\draw (2.7,6.5) node {{\tiny $b$}};
\Vpoint{7}{1};
\draw (7,0.6) node {{\tiny $i_{GR}$}};
\Vpoint{6}{0.5};
\draw (6.25,0.35) node {{\tiny $d$}};
\Hpoint{8}{2};
\draw (8.75,1.75) node {{\tiny $j_{GR}$}};
\Hpoint{8.5}{3};
\draw (8.5,2.65) node {{\tiny $c$}};
\Hpoint{4}{5.5};
\draw (4.2,5.8) node {{\tiny $x$}};
\Hpoint{5}{2};
\draw (4.8,2.2) node {{\tiny $y$}};
\draw (8.5,5.5) node {$\varnothing$};
\draw (1.5,0.5) node {$\varnothing$};
\draw (3.5,0.5) node {$\varnothing$};
\draw (0.5,1.5) node {$\varnothing$};
\draw (7.5,6.5) node {$\varnothing$};
\draw (3.5,3.5) node {$A$};
\draw (1.5,4.5) node {$\varnothing$};
\draw (3.5,4.5) node {$\varnothing$};
\draw (7.5,4.5) node {$\varnothing$};
\draw (8.5,4.5) node {$\varnothing$};
\draw (0.5,3.5) node {$\varnothing$};
\draw (2.5,0.5) node {$\varnothing$};
\draw (2.5,1.5) node {$\varnothing$};
\draw (2.5,3.5) node {$\varnothing$};
\draw (2.5,4.5) node {$\varnothing$};
\draw (2.5,5.5) node {$\varnothing$};
\draw (8.5,3.5) node {$\varnothing$};
\draw (3.5,6.5) node {$\varnothing$};
\draw (0.5,2.5) node {$\varnothing$};
\draw (1.5,2.5) node {$\varnothing$};
\draw (2.5,2.5) node {$\varnothing$};
\draw (7.5,2.5) node {$\varnothing$};
\draw (6.5,1.5) node {$\varnothing$};
\draw (6.5,2.5) node {$\varnothing$};
\draw (6.5,4.5) node {$\varnothing$};
\draw (6.5,5.5) node {$\varnothing$};
\draw (6.5,6.5) node {$\varnothing$};
\draw (4.5,0.5) node {$\varnothing$};
\draw (4.5,1.5) node {$\varnothing$};
\draw (4.5,2.5) node {$\varnothing$};
\draw (5.5,4.5) node {$\varnothing$};
\draw (5.5,5.5) node {$\varnothing$};
\draw (5.5,6.5) node {$\varnothing$};
\draw (5.5,0.5) node {$\varnothing$};
\draw (4.5,3.5) node {$*$};
\draw (5.5,3.5) node {$*$};
\draw (4.5,4.5) node {$\varnothing$};
\draw (4.5,6.5) node {$\varnothing$};
\draw (4.5,5.5) node {$\varnothing$};
\draw (3.5,1.5) node {$\varnothing$};
\draw (3.5,2.5) node {$\varnothing$};
\draw (3.5,5.5) node {$2$};
\draw (7,3.5) node {$3$};
\draw (5.5,2) node {$4$};
\zoneRG{1}{6}{1}
\zoneGR{9}{0}{1}
\end{tikzpicture}
\begin{tikzpicture}[scale=.5]
\useasboundingbox (0,-0.5) (10,7.5);
\fill [Hfill] (0,4) rectangle (1,6);
\fill [Vfill] (1,6) rectangle (3,7);
\fill [Hfill] (8,1) rectangle (9,3);
\fill [Vfill] (6,0) rectangle (8,1);
\fill [Hfill] (5,1) rectangle (6,3);
\fill [Hfill] (3,5) rectangle (4,6);
\fill [Vfill] (5,3) rectangle (8,4);
\fill [Hfill] (3,3) rectangle (4,4);
\draw (0,1) -- (9,1);
\draw (0,2) -- (9,2);
\draw (0,3) -- (9,3);
\draw (0,4) -- (9,4);
\draw (0,5) -- (9,5);
\draw (0,6) -- (9,6);
\draw (1,0) -- (1,7);
\draw (2,0) -- (2,7);
\draw (3,0) -- (3,7);
\draw (4,0) -- (4,7);
\draw (5,0) -- (5,7);
\draw (6,0) -- (6,7);
\draw (7,0) -- (7,7);
\draw (8,0) -- (8,7);
\draw (1.5,1.5) node {$\varnothing$};
\draw (1.5,3.5) node {$\varnothing$};
\draw (1.5,5.5) node {$\varnothing$};
\draw (7.5,5.5) node {$\varnothing$};
\draw (7.5,1.5) node {$\varnothing$};
\draw (0.5,0.5) node {$\varnothing$};
\draw (8.5,6.5) node {$\varnothing$};
\Hpoint{1}{5};
\draw (0.5,4.75) node {{\tiny $i_{RG}$}};
\Hpoint{0.5}{4};
\draw (0.5,4.3) node {{\tiny $a$}};
\Vpoint{2}{6};
\draw (2,6.5) node {{\tiny $j_{RG}$}};
\Vpoint{3}{6.5};
\draw (2.7,6.5) node {{\tiny $b$}};
\Vpoint{7}{1};
\draw (7,0.6) node {{\tiny $i_{GR}$}};
\Vpoint{6}{0.5};
\draw (6.25,0.35) node {{\tiny $d$}};
\Hpoint{8}{2};
\draw (8.75,1.75) node {{\tiny $j_{GR}$}};
\Hpoint{8.5}{3};
\draw (8.5,2.65) node {{\tiny $c$}};
\Hpoint{4}{5.5};
\draw (4.2,5.8) node {{\tiny $x$}};
\Hpoint{5}{2};
\draw (4.8,2.2) node {{\tiny $y$}};
\draw (8.5,5.5) node {$\varnothing$};
\draw (1.5,0.5) node {$\varnothing$};
\draw (3.5,0.5) node {$\varnothing$};
\draw (0.5,1.5) node {$\varnothing$};
\draw (7.5,6.5) node {$\varnothing$};
\draw (3.5,3.5) node {$A$};
\draw (1.5,4.5) node {$\varnothing$};
\draw (3.5,4.5) node {$\varnothing$};
\draw (7.5,4.5) node {$\varnothing$};
\draw (8.5,4.5) node {$\varnothing$};
\draw (0.5,3.5) node {$\varnothing$};
\draw (2.5,0.5) node {$\varnothing$};
\draw (2.5,1.5) node {$\varnothing$};
\draw (2.5,3.5) node {$\varnothing$};
\draw (2.5,4.5) node {$\varnothing$};
\draw (2.5,5.5) node {$\varnothing$};
\draw (8.5,3.5) node {$\varnothing$};
\draw (3.5,6.5) node {$\varnothing$};
\draw (0.5,2.5) node {$\varnothing$};
\draw (1.5,2.5) node {$\varnothing$};
\draw (2.5,2.5) node {$\varnothing$};
\draw (7.5,2.5) node {$\varnothing$};
\draw (6.5,1.5) node {$\varnothing$};
\draw (6.5,2.5) node {$\varnothing$};
\draw (6.5,4.5) node {$\varnothing$};
\draw (6.5,5.5) node {$\varnothing$};
\draw (6.5,6.5) node {$\varnothing$};
\draw (4.5,0.5) node {$\varnothing$};
\draw (4.5,1.5) node {$\varnothing$};
\draw (4.5,2.5) node {$\varnothing$};
\draw (5.5,4.5) node {$\varnothing$};
\draw (5.5,5.5) node {$\varnothing$};
\draw (5.5,6.5) node {$\varnothing$};
\draw (5.5,0.5) node {$\varnothing$};
\draw (4.5,3.5) node {$*$};
\draw (4.5,4.5) node {$\varnothing$};
\draw (4.5,6.5) node {$\varnothing$};
\draw (4.5,5.5) node {$\varnothing$};
\draw (3.5,1.5) node {$\varnothing$};
\draw (3.5,2.5) node {$\varnothing$};
\draw (3.5,5.5) node {$2$};
\draw (6.5,3.5) node {$3$};
\draw (5.5,2) node {$4$};
\zoneRG{1}{6}{1}
\zoneGR{9}{0}{1}
\end{tikzpicture}

\caption{Zone $1$ is empty and zone $4$ is not empty\label{fig:zone4nonempty}}
\end{center}
\end{figure}

We apply rule (\rmnum{2}) to $y$ and $c$ and obtain the second diagram.
But ($i_{RG}$, $j_{RG}$) is lowest-right increasing sequence RG, hence there is no point labeled \G in the above-right quadrant of $y$. Hence zone $3$ is empty which is forbidden as $\sigma$ is $\ominus$-indecomposable.

\paragraph{Zone $4$ is empty}

We prove that this case is also not possible.
Suppose that zone $4$ is empty as illustrated in the first diagram of Figure~\ref{fig:RGGR4empty}.

\begin{figure}[H]
\begin{center}
\begin{tikzpicture}[scale=.5]
\useasboundingbox (0,-0.5) (8.75,7.5);
\fill [Hfill] (0,4) rectangle (1,6);
\fill [Vfill] (1,6) rectangle (3,7);
\fill [Hfill] (7,1) rectangle (8,3);
\fill [Vfill] (5,0) rectangle (7,1);
\fill [Hfill] (3,5) rectangle (4,6);
\fill [Vfill] (5,3) rectangle (7,4);
\fill [Hfill] (3,3) rectangle (4,4);
\draw (0,1) -- (8,1);
\draw (0,2) -- (8,2);
\draw (0,3) -- (8,3);
\draw (0,4) -- (8,4);
\draw (0,5) -- (8,5);
\draw (0,6) -- (8,6);
\draw (1,0) -- (1,7);
\draw (2,0) -- (2,7);
\draw (3,0) -- (3,7);
\draw (4,0) -- (4,7);
\draw (5,0) -- (5,7);
\draw (6,0) -- (6,7);
\draw (7,0) -- (7,7);
\draw (1.5,1.5) node {$\varnothing$};
\draw (1.5,3.5) node {$\varnothing$};
\draw (1.5,5.5) node {$\varnothing$};
\draw (6.5,5.5) node {$\varnothing$};
\draw (6.5,1.5) node {$\varnothing$};
\draw (0.5,0.5) node {$\varnothing$};
\draw (7.5,6.5) node {$\varnothing$};
\Hpoint{1}{5};
\draw (0.5,4.75) node {{\tiny $i_{RG}$}};
\Hpoint{0.5}{4};
\draw (0.5,4.3) node {{\tiny $a$}};
\Vpoint{2}{6};
\draw (2,6.5) node {{\tiny $j_{RG}$}};
\Vpoint{3}{6.5};
\draw (2.7,6.5) node {{\tiny $b$}};
\Vpoint{6}{1};
\draw (6,0.6) node {{\tiny $i_{GR}$}};
\Vpoint{5}{0.5};
\draw (5.25,0.35) node {{\tiny $d$}};
\Hpoint{7}{2};
\draw (7.75,1.75) node {{\tiny $j_{GR}$}};
\Hpoint{7.5}{3};
\draw (7.5,2.65) node {{\tiny $c$}};
\Hpoint{4}{5.5};
\draw (4.2,5.8) node {{\tiny $x$}};
\draw (7.5,5.5) node {$\varnothing$};
\draw (1.5,0.5) node {$\varnothing$};
\draw (3.5,0.5) node {$\varnothing$};
\draw (0.5,1.5) node {$\varnothing$};
\draw (6.5,6.5) node {$\varnothing$};
\draw (3.5,3.5) node {$A$};
\draw (1.5,4.5) node {$\varnothing$};
\draw (3.5,4.5) node {$\varnothing$};
\draw (6.5,4.5) node {$\varnothing$};
\draw (7.5,4.5) node {$\varnothing$};
\draw (0.5,3.5) node {$\varnothing$};
\draw (2.5,0.5) node {$\varnothing$};
\draw (2.5,1.5) node {$\varnothing$};
\draw (2.5,3.5) node {$\varnothing$};
\draw (2.5,4.5) node {$\varnothing$};
\draw (2.5,5.5) node {$\varnothing$};
\draw (7.5,3.5) node {$\varnothing$};
\draw (3.5,6.5) node {$\varnothing$};
\draw (0.5,2.5) node {$\varnothing$};
\draw (1.5,2.5) node {$\varnothing$};
\draw (2.5,2.5) node {$\varnothing$};
\draw (6.5,2.5) node {$\varnothing$};
\draw (5.5,1.5) node {$\varnothing$};
\draw (5.5,2.5) node {$\varnothing$};
\draw (5.5,4.5) node {$\varnothing$};
\draw (5.5,5.5) node {$\varnothing$};
\draw (5.5,6.5) node {$\varnothing$};
\draw (4.5,0.5) node {$\varnothing$};
\draw (4.5,3.5) node {$*$};
\draw (4.5,1.5) node {$\varnothing$};
\draw (4.5,2.5) node {$\varnothing$};
\draw (4.5,4.5) node {$\varnothing$};
\draw (4.5,6.5) node {$\varnothing$};
\draw (4.5,5.5) node {$\varnothing$};
\draw (3.5,1.5) node {$\varnothing$};
\draw (3.5,2.5) node {$\varnothing$};
\draw (3.5,5.5) node {$2$};
\draw (6,3.5) node {$3$};
\zoneRG{1}{6}{1}
\zoneGR{8}{0}{1}
\end{tikzpicture}
\begin{tikzpicture}[scale=.5]
\useasboundingbox (0,-0.5) (8.75,8.5);
\fill [Hfill] (0,5) rectangle (1,7);
\fill [Vfill] (1,7) rectangle (3,8);
\fill [Hfill] (7,1) rectangle (8,3);
\fill [Vfill] (5,0) rectangle (7,1);
\fill [Hfill] (3,6) rectangle (4,7);
\fill [Vfill] (5,3) rectangle (7,4);
\fill [Hfill] (3,3) rectangle (4,5);
\fill [Hfill] (4,3) rectangle (5,4);
\draw (0,1) -- (8,1);
\draw (0,2) -- (8,2);
\draw (0,3) -- (8,3);
\draw (0,4) -- (8,4);
\draw (0,5) -- (8,5);
\draw (0,6) -- (8,6);
\draw (0,7) -- (8,7);
\draw (1,0) -- (1,8);
\draw (2,0) -- (2,8);
\draw (3,0) -- (3,8);
\draw (4,0) -- (4,8);
\draw (5,0) -- (5,8);
\draw (6,0) -- (6,8);
\draw (7,0) -- (7,8);
\draw (1.5,1.5) node {$\varnothing$};
\draw (1.5,3.5) node {$\varnothing$};
\draw (1.5,6.5) node {$\varnothing$};
\draw (6.5,6.5) node {$\varnothing$};
\draw (6.5,1.5) node {$\varnothing$};
\draw (0.5,0.5) node {$\varnothing$};
\draw (7.5,7.5) node {$\varnothing$};
\Hpoint{1}{6};
\draw (0.5,5.75) node {{\tiny $i_{RG}$}};
\Hpoint{0.5}{5};
\draw (0.5,5.3) node {{\tiny $a$}};
\Vpoint{2}{7};
\draw (2,7.5) node {{\tiny $j_{RG}$}};
\Vpoint{3}{7.5};
\draw (2.7,7.5) node {{\tiny $b$}};
\Vpoint{6}{1};
\draw (6,0.6) node {{\tiny $i_{GR}$}};
\Vpoint{5}{0.5};
\draw (5.25,0.35) node {{\tiny $d$}};
\Hpoint{7}{2};
\draw (7.75,1.75) node {{\tiny $j_{GR}$}};
\Hpoint{7.5}{3};
\draw (7.5,2.65) node {{\tiny $c$}};
\Hpoint{4}{6.5};
\draw (4.2,6.8) node {{\tiny $x$}};
\Vpoint{6}{4};
\draw (5.8,4.2) node {{\tiny $z$}};
\draw (7.5,6.5) node {$\varnothing$};
\draw (1.5,0.5) node {$\varnothing$};
\draw (3.5,0.5) node {$\varnothing$};
\draw (0.5,1.5) node {$\varnothing$};
\draw (6.5,7.5) node {$\varnothing$};
\draw (3.5,4) node {$A$};
\draw (1.5,5.5) node {$\varnothing$};
\draw (3.5,5.5) node {$\varnothing$};
\draw (6.5,5.5) node {$\varnothing$};
\draw (7.5,5.5) node {$\varnothing$};
\draw (0.5,4.5) node {$\varnothing$};
\draw (1.5,4.5) node {$\varnothing$};
\draw (2.5,4.5) node {$\varnothing$};
\draw (5.5,4.5) node {$\varnothing$};
\draw (6.5,4.5) node {$\varnothing$};
\draw (7.5,4.5) node {$\varnothing$};
\draw (0.5,3.5) node {$\varnothing$};
\draw (2.5,0.5) node {$\varnothing$};
\draw (2.5,1.5) node {$\varnothing$};
\draw (2.5,3.5) node {$\varnothing$};
\draw (2.5,5.5) node {$\varnothing$};
\draw (2.5,6.5) node {$\varnothing$};
\draw (7.5,3.5) node {$\varnothing$};
\draw (3.5,7.5) node {$\varnothing$};
\draw (0.5,2.5) node {$\varnothing$};
\draw (1.5,2.5) node {$\varnothing$};
\draw (2.5,2.5) node {$\varnothing$};
\draw (6.5,2.5) node {$\varnothing$};
\draw (5.5,1.5) node {$\varnothing$};
\draw (5.5,2.5) node {$\varnothing$};
\draw (5.5,5.5) node {$\varnothing$};
\draw (5.5,6.5) node {$\varnothing$};
\draw (5.5,7.5) node {$\varnothing$};
\draw (4.5,0.5) node {$\varnothing$};
\draw (4.5,4.5) node {$*$};
\draw (4.5,1.5) node {$\varnothing$};
\draw (4.5,2.5) node {$\varnothing$};
\draw (4.5,5.5) node {$\varnothing$};
\draw (4.5,6.5) node {$\varnothing$};
\draw (4.5,7.5) node {$\varnothing$};
\draw (3.5,1.5) node {$\varnothing$};
\draw (3.5,2.5) node {$\varnothing$};
\draw (3.5,6.5) node {$2$};
\draw (6,3.5) node {$3$};
\zoneRG{1}{7}{1}
\zoneGR{8}{0}{1}
\end{tikzpicture}
\begin{tikzpicture}[scale=.5]
\useasboundingbox (0,-0.5) (8.75,8.5);
\fill [Hfill] (0,5) rectangle (1,7);
\fill [Vfill] (1,7) rectangle (3,8);
\fill [Hfill] (7,1) rectangle (8,3);
\fill [Vfill] (5,0) rectangle (7,1);
\fill [Hfill] (3,6) rectangle (4,7);
\fill [Vfill] (5,3) rectangle (7,4);
\fill [Hfill] (3,4) rectangle (4,5);
\draw (0,1) -- (8,1);
\draw (0,2) -- (8,2);
\draw (0,3) -- (8,3);
\draw (0,4) -- (8,4);
\draw (0,5) -- (8,5);
\draw (0,6) -- (8,6);
\draw (0,7) -- (8,7);
\draw (1,0) -- (1,8);
\draw (2,0) -- (2,8);
\draw (3,0) -- (3,8);
\draw (4,0) -- (4,8);
\draw (5,0) -- (5,8);
\draw (6,0) -- (6,8);
\draw (7,0) -- (7,8);
\draw (1.5,1.5) node {$\varnothing$};
\draw (1.5,3.5) node {$\varnothing$};
\draw (1.5,6.5) node {$\varnothing$};
\draw (6.5,6.5) node {$\varnothing$};
\draw (6.5,1.5) node {$\varnothing$};
\draw (0.5,0.5) node {$\varnothing$};
\draw (7.5,7.5) node {$\varnothing$};
\Hpoint{1}{6};
\draw (0.5,5.75) node {{\tiny $i_{RG}$}};
\Hpoint{0.5}{5};
\draw (0.5,5.3) node {{\tiny $a$}};
\Vpoint{2}{7};
\draw (2,7.5) node {{\tiny $j_{RG}$}};
\Vpoint{3}{7.5};
\draw (2.7,7.5) node {{\tiny $b$}};
\Vpoint{6}{1};
\draw (6,0.6) node {{\tiny $i_{GR}$}};
\Vpoint{5}{0.5};
\draw (5.25,0.35) node {{\tiny $d$}};
\Hpoint{7}{2};
\draw (7.75,1.75) node {{\tiny $j_{GR}$}};
\Hpoint{7.5}{3};
\draw (7.5,2.65) node {{\tiny $c$}};
\Hpoint{4}{6.5};
\draw (4.2,6.8) node {{\tiny $x$}};
\Vpoint{6}{4};
\draw (5.8,4.2) node {{\tiny $z$}};
\draw (7.5,6.5) node {$\varnothing$};
\draw (1.5,0.5) node {$\varnothing$};
\draw (3.5,0.5) node {$\varnothing$};
\draw (3.5,3.5) node {$\varnothing$};
\draw (4.5,3.5) node {$\varnothing$};
\draw (0.5,1.5) node {$\varnothing$};
\draw (6.5,7.5) node {$\varnothing$};
\draw (3.5,4.5) node {$A$};
\draw (1.5,5.5) node {$\varnothing$};
\draw (3.5,5.5) node {$\varnothing$};
\draw (6.5,5.5) node {$\varnothing$};
\draw (7.5,5.5) node {$\varnothing$};
\draw (0.5,4.5) node {$\varnothing$};
\draw (1.5,4.5) node {$\varnothing$};
\draw (2.5,4.5) node {$\varnothing$};
\draw (5.5,4.5) node {$\varnothing$};
\draw (6.5,4.5) node {$\varnothing$};
\draw (7.5,4.5) node {$\varnothing$};
\draw (0.5,3.5) node {$\varnothing$};
\draw (2.5,0.5) node {$\varnothing$};
\draw (2.5,1.5) node {$\varnothing$};
\draw (2.5,3.5) node {$\varnothing$};
\draw (2.5,5.5) node {$\varnothing$};
\draw (2.5,6.5) node {$\varnothing$};
\draw (7.5,3.5) node {$\varnothing$};
\draw (3.5,7.5) node {$\varnothing$};
\draw (0.5,2.5) node {$\varnothing$};
\draw (1.5,2.5) node {$\varnothing$};
\draw (2.5,2.5) node {$\varnothing$};
\draw (6.5,2.5) node {$\varnothing$};
\draw (5.5,1.5) node {$\varnothing$};
\draw (5.5,2.5) node {$\varnothing$};
\draw (5.5,5.5) node {$\varnothing$};
\draw (5.5,6.5) node {$\varnothing$};
\draw (5.5,7.5) node {$\varnothing$};
\draw (4.5,0.5) node {$\varnothing$};
\draw (4.5,4.5) node {$*$};
\draw (4.5,1.5) node {$\varnothing$};
\draw (4.5,2.5) node {$\varnothing$};
\draw (4.5,5.5) node {$\varnothing$};
\draw (4.5,6.5) node {$\varnothing$};
\draw (4.5,7.5) node {$\varnothing$};
\draw (3.5,1.5) node {$\varnothing$};
\draw (3.5,2.5) node {$\varnothing$};
\draw (3.5,6.5) node {$2$};
\draw (6,3.5) node {$3$};
\zoneRG{1}{7}{1}
\zoneGR{8}{0}{1}
\end{tikzpicture}

\caption{Zone $1$ is empty and zone $4$ is empty.\label{fig:RGGR4empty}}
\end{center}
\end{figure}
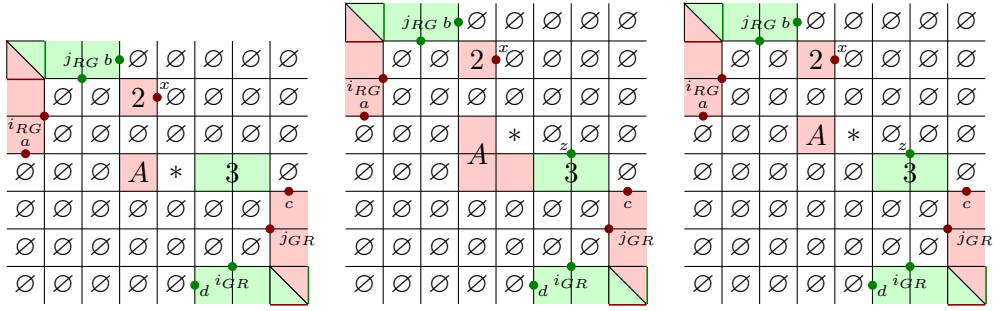

As $\sigma$ is $\ominus$-indecomposable, zone $3$ is non-empty. Let $z$ be the topmost point of zone $3$ (it may be to the left or to the right of $i_{GR}$).
Applying rule (\rmnum{1}) to $z$ and $d$ we obtain the second diagram. 
But ($i_{RG}$, $j_{RG}$) is the lowest right increasing sequence labeled RG, hence there are no point labeled \R in the below-left quadrant of $z$ -- see diagram $3$ --. 
But then $\sigma$ is $\ominus$-decomposable which is forbidden.

\paragraph{Zone $1$ is not empty}

Suppose that zone $1$ of Figure~\ref{fig:ABempty} is non empty. 
Define $x$ as the lowest point of this zone as shown in the first diagram of Figure~\ref{RGGR1nonempty}.

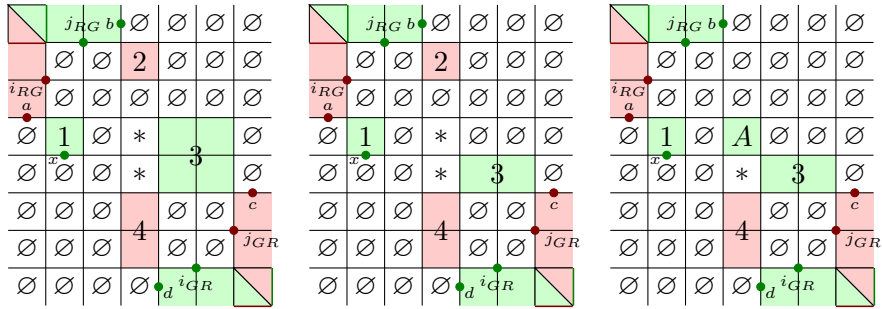
\begin{figure}[H]
\begin{center}
\begin{tikzpicture}[scale=.5]
\useasboundingbox (0,-0.5) (7.75,8.5);
\fill [Hfill] (0,5) rectangle (1,7);
\fill [Vfill] (1,7) rectangle (3,8);
\fill [Hfill] (6,1) rectangle (7,3);
\fill [Vfill] (4,0) rectangle (6,1);
\fill [Hfill] (3,1) rectangle (4,3);
\fill [Vfill] (1,4) rectangle (2,5);
\fill [Hfill] (3,6) rectangle (4,7);
\fill [Vfill] (4,3) rectangle (6,5);
\draw (0,1) -- (7,1);
\draw (0,2) -- (7,2);
\draw (0,3) -- (7,3);
\draw (0,4) -- (7,4);
\draw (0,5) -- (7,5);
\draw (0,6) -- (7,6);
\draw (0,7) -- (7,7);
\draw (1,0) -- (1,8);
\draw (2,0) -- (2,8);
\draw (3,0) -- (3,8);
\draw (4,0) -- (4,8);
\draw (5,0) -- (5,8);
\draw (6,0) -- (6,8);
\draw (1.5,1.5) node {$\varnothing$};
\draw (1.5,6.5) node {$\varnothing$};
\draw (5.5,6.5) node {$\varnothing$};
\draw (5.5,1.5) node {$\varnothing$};
\draw (0.5,0.5) node {$\varnothing$};
\draw (6.5,7.5) node {$\varnothing$};
\Hpoint{1}{6};
\draw (0.5,5.75) node {{\tiny $i_{RG}$}};
\Hpoint{0.5}{5};
\draw (0.5,5.3) node {{\tiny $a$}};
\Vpoint{2}{7};
\draw (2,7.5) node {{\tiny $j_{RG}$}};
\Vpoint{3}{7.5};
\draw (2.7,7.5) node {{\tiny $b$}};
\Vpoint{5}{1};
\draw (5,0.6) node {{\tiny $i_{GR}$}};
\Vpoint{4}{0.5};
\draw (4.25,0.35) node {{\tiny $d$}};
\Hpoint{6}{2};
\draw (6.75,1.75) node {{\tiny $j_{GR}$}};
\Hpoint{6.5}{3};
\draw (6.5,2.65) node {{\tiny $c$}};
\Vpoint{1.5}{4};
\draw (1.2,3.85) node {{\tiny $x$}};
\draw (6.5,6.5) node {$\varnothing$};
\draw (1.5,0.5) node {$\varnothing$};
\draw (3.5,0.5) node {$\varnothing$};
\draw (0.5,1.5) node {$\varnothing$};
\draw (5.5,7.5) node {$\varnothing$};
\draw (3.5,3.5) node {$*$};
\draw (3.5,4.5) node {$*$};
\draw (1.5,3.5) node {$\varnothing$};
\draw (0.5,4.5) node {$\varnothing$};
\draw (2.5,4.5) node {$\varnothing$};
\draw (6.5,4.5) node {$\varnothing$};
\draw (1.5,5.5) node {$\varnothing$};
\draw (3.5,5.5) node {$\varnothing$};
\draw (5.5,5.5) node {$\varnothing$};
\draw (6.5,5.5) node {$\varnothing$};
\draw (0.5,3.5) node {$\varnothing$};
\draw (2.5,0.5) node {$\varnothing$};
\draw (2.5,1.5) node {$\varnothing$};
\draw (2.5,3.5) node {$\varnothing$};
\draw (2.5,5.5) node {$\varnothing$};
\draw (2.5,6.5) node {$\varnothing$};
\draw (6.5,3.5) node {$\varnothing$};
\draw (3.5,7.5) node {$\varnothing$};
\draw (0.5,2.5) node {$\varnothing$};
\draw (1.5,2.5) node {$\varnothing$};
\draw (2.5,2.5) node {$\varnothing$};
\draw (5.5,2.5) node {$\varnothing$};
\draw (4.5,1.5) node {$\varnothing$};
\draw (4.5,2.5) node {$\varnothing$};
\draw (4.5,5.5) node {$\varnothing$};
\draw (4.5,6.5) node {$\varnothing$};
\draw (4.5,7.5) node {$\varnothing$};
\draw (1.5,4.5) node {$1$};
\draw (3.5,6.5) node {$2$};
\draw (5,4) node {$3$};
\draw (3.5,2) node {$4$};
\zoneRG{1}{7}{1}
\zoneGR{7}{0}{1}
\end{tikzpicture}
\begin{tikzpicture}[scale=.5]
\useasboundingbox (0,-0.5) (7.75,8.5);
\fill [Hfill] (0,5) rectangle (1,7);
\fill [Vfill] (1,7) rectangle (3,8);
\fill [Hfill] (6,1) rectangle (7,3);
\fill [Vfill] (4,0) rectangle (6,1);
\fill [Hfill] (3,1) rectangle (4,3);
\fill [Vfill] (1,4) rectangle (2,5);
\fill [Hfill] (3,6) rectangle (4,7);
\fill [Vfill] (4,3) rectangle (6,4);
\draw (0,1) -- (7,1);
\draw (0,2) -- (7,2);
\draw (0,3) -- (7,3);
\draw (0,4) -- (7,4);
\draw (0,5) -- (7,5);
\draw (0,6) -- (7,6);
\draw (0,7) -- (7,7);
\draw (1,0) -- (1,8);
\draw (2,0) -- (2,8);
\draw (3,0) -- (3,8);
\draw (4,0) -- (4,8);
\draw (5,0) -- (5,8);
\draw (6,0) -- (6,8);
\draw (1.5,1.5) node {$\varnothing$};
\draw (1.5,6.5) node {$\varnothing$};
\draw (5.5,6.5) node {$\varnothing$};
\draw (5.5,1.5) node {$\varnothing$};
\draw (0.5,0.5) node {$\varnothing$};
\draw (6.5,7.5) node {$\varnothing$};
\Hpoint{1}{6};
\draw (0.5,5.75) node {{\tiny $i_{RG}$}};
\Hpoint{0.5}{5};
\draw (0.5,5.3) node {{\tiny $a$}};
\Vpoint{2}{7};
\draw (2,7.5) node {{\tiny $j_{RG}$}};
\Vpoint{3}{7.5};
\draw (2.7,7.5) node {{\tiny $b$}};
\Vpoint{5}{1};
\draw (5,0.6) node {{\tiny $i_{GR}$}};
\Vpoint{4}{0.5};
\draw (4.25,0.35) node {{\tiny $d$}};
\Hpoint{6}{2};
\draw (6.75,1.75) node {{\tiny $j_{GR}$}};
\Hpoint{6.5}{3};
\draw (6.5,2.65) node {{\tiny $c$}};
\Vpoint{1.5}{4};
\draw (1.2,3.85) node {{\tiny $x$}};
\draw (6.5,6.5) node {$\varnothing$};
\draw (1.5,0.5) node {$\varnothing$};
\draw (3.5,0.5) node {$\varnothing$};
\draw (0.5,1.5) node {$\varnothing$};
\draw (5.5,7.5) node {$\varnothing$};
\draw (3.5,3.5) node {$*$};
\draw (3.5,4.5) node {$*$};
\draw (1.5,3.5) node {$\varnothing$};
\draw (0.5,4.5) node {$\varnothing$};
\draw (2.5,4.5) node {$\varnothing$};
\draw (6.5,4.5) node {$\varnothing$};
\draw (1.5,5.5) node {$\varnothing$};
\draw (3.5,5.5) node {$\varnothing$};
\draw (5.5,5.5) node {$\varnothing$};
\draw (6.5,5.5) node {$\varnothing$};
\draw (0.5,3.5) node {$\varnothing$};
\draw (2.5,0.5) node {$\varnothing$};
\draw (2.5,1.5) node {$\varnothing$};
\draw (2.5,3.5) node {$\varnothing$};
\draw (2.5,5.5) node {$\varnothing$};
\draw (2.5,6.5) node {$\varnothing$};
\draw (6.5,3.5) node {$\varnothing$};
\draw (3.5,7.5) node {$\varnothing$};
\draw (0.5,2.5) node {$\varnothing$};
\draw (1.5,2.5) node {$\varnothing$};
\draw (2.5,2.5) node {$\varnothing$};
\draw (5.5,2.5) node {$\varnothing$};
\draw (5.5,4.5) node {$\varnothing$};
\draw (4.5,4.5) node {$\varnothing$};
\draw (4.5,1.5) node {$\varnothing$};
\draw (4.5,2.5) node {$\varnothing$};
\draw (4.5,5.5) node {$\varnothing$};
\draw (4.5,6.5) node {$\varnothing$};
\draw (4.5,7.5) node {$\varnothing$};
\draw (1.5,4.5) node {$1$};
\draw (3.5,6.5) node {$2$};
\draw (5,3.5) node {$3$};
\draw (3.5,2) node {$4$};
\zoneRG{1}{7}{1}
\zoneGR{7}{0}{1}
\end{tikzpicture}
\begin{tikzpicture}[scale=.5]
\useasboundingbox (0,-0.5) (7.75,8.5);
\fill [Hfill] (0,5) rectangle (1,7);
\fill [Vfill] (1,7) rectangle (3,8);
\fill [Hfill] (6,1) rectangle (7,3);
\fill [Vfill] (4,0) rectangle (6,1);
\fill [Hfill] (3,1) rectangle (4,3);
\fill [Vfill] (1,4) rectangle (2,5);
\fill [Vfill] (4,3) rectangle (6,4);
\fill [Vfill] (3,4) rectangle (4,5);
\draw (0,1) -- (7,1);
\draw (0,2) -- (7,2);
\draw (0,3) -- (7,3);
\draw (0,4) -- (7,4);
\draw (0,5) -- (7,5);
\draw (0,6) -- (7,6);
\draw (0,7) -- (7,7);
\draw (1,0) -- (1,8);
\draw (2,0) -- (2,8);
\draw (3,0) -- (3,8);
\draw (4,0) -- (4,8);
\draw (5,0) -- (5,8);
\draw (6,0) -- (6,8);
\draw (1.5,1.5) node {$\varnothing$};
\draw (1.5,6.5) node {$\varnothing$};
\draw (5.5,6.5) node {$\varnothing$};
\draw (5.5,1.5) node {$\varnothing$};
\draw (0.5,0.5) node {$\varnothing$};
\draw (6.5,7.5) node {$\varnothing$};
\Hpoint{1}{6};
\draw (0.5,5.75) node {{\tiny $i_{RG}$}};
\Hpoint{0.5}{5};
\draw (0.5,5.3) node {{\tiny $a$}};
\Vpoint{2}{7};
\draw (2,7.5) node {{\tiny $j_{RG}$}};
\Vpoint{3}{7.5};
\draw (2.7,7.5) node {{\tiny $b$}};
\Vpoint{5}{1};
\draw (5,0.6) node {{\tiny $i_{GR}$}};
\Vpoint{4}{0.5};
\draw (4.25,0.35) node {{\tiny $d$}};
\Hpoint{6}{2};
\draw (6.75,1.75) node {{\tiny $j_{GR}$}};
\Hpoint{6.5}{3};
\draw (6.5,2.65) node {{\tiny $c$}};
\Vpoint{1.5}{4};
\draw (1.2,3.85) node {{\tiny $x$}};
\draw (6.5,6.5) node {$\varnothing$};
\draw (1.5,0.5) node {$\varnothing$};
\draw (3.5,0.5) node {$\varnothing$};
\draw (0.5,1.5) node {$\varnothing$};
\draw (5.5,7.5) node {$\varnothing$};
\draw (3.5,3.5) node {$*$};
\draw (3.5,4.5) node {$A$};
\draw (1.5,3.5) node {$\varnothing$};
\draw (0.5,4.5) node {$\varnothing$};
\draw (2.5,4.5) node {$\varnothing$};
\draw (6.5,4.5) node {$\varnothing$};
\draw (1.5,5.5) node {$\varnothing$};
\draw (3.5,5.5) node {$\varnothing$};
\draw (5.5,5.5) node {$\varnothing$};
\draw (6.5,5.5) node {$\varnothing$};
\draw (0.5,3.5) node {$\varnothing$};
\draw (2.5,0.5) node {$\varnothing$};
\draw (2.5,1.5) node {$\varnothing$};
\draw (2.5,3.5) node {$\varnothing$};
\draw (2.5,5.5) node {$\varnothing$};
\draw (2.5,6.5) node {$\varnothing$};
\draw (6.5,3.5) node {$\varnothing$};
\draw (3.5,6.5) node {$\varnothing$};
\draw (3.5,7.5) node {$\varnothing$};
\draw (0.5,2.5) node {$\varnothing$};
\draw (1.5,2.5) node {$\varnothing$};
\draw (2.5,2.5) node {$\varnothing$};
\draw (5.5,2.5) node {$\varnothing$};
\draw (5.5,4.5) node {$\varnothing$};
\draw (4.5,4.5) node {$\varnothing$};
\draw (4.5,1.5) node {$\varnothing$};
\draw (4.5,2.5) node {$\varnothing$};
\draw (4.5,5.5) node {$\varnothing$};
\draw (4.5,6.5) node {$\varnothing$};
\draw (4.5,7.5) node {$\varnothing$};
\draw (1.5,4.5) node {$1$};
\draw (5,3.5) node {$3$};
\draw (3.5,2) node {$4$};
\zoneRG{1}{7}{1}
\zoneGR{7}{0}{1}
\end{tikzpicture}
\caption{Zone $1$ is not empty.\label{RGGR1nonempty}}
\end{center}
\end{figure}

Rule (\rmnum{8}) applied to $x$ and $d$ implies the second diagram. 
Moreover, as ($i_{GR}$, $j_{GR}$) is the leftmost-top increasing sequence labeled GR, all points to the top right of $x$ are in \G, 
leading to the last diagram.

\paragraph{Zone $3$ is not empty}

If zone $3$ is not empty, let $y$ be its topmost point ($y$ may be to the left or to the right of $i_{GR}$) as pictured in Figure~\ref{fig:RGGR3notempty}.

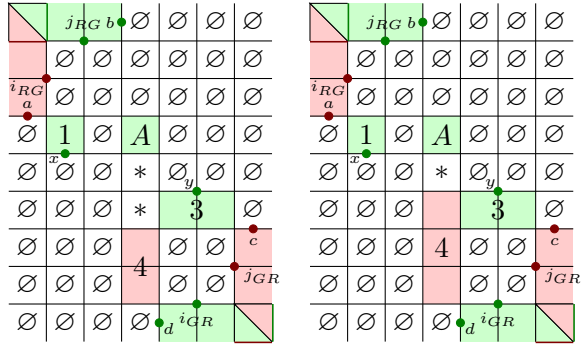
\begin{figure}[H]
\begin{center}
\begin{tikzpicture}[scale=.5]
\useasboundingbox (0,-0.5) (7.75,9.5);
\fill [Hfill] (0,6) rectangle (1,8);
\fill [Vfill] (1,8) rectangle (3,9);
\fill [Hfill] (6,1) rectangle (7,3);
\fill [Vfill] (4,0) rectangle (6,1);
\fill [Hfill] (3,1) rectangle (4,3);
\fill [Vfill] (1,5) rectangle (2,6);
\fill [Vfill] (4,3) rectangle (6,4);
\fill [Vfill] (3,5) rectangle (4,6);
\draw (0,1) -- (7,1);
\draw (0,2) -- (7,2);
\draw (0,3) -- (7,3);
\draw (0,4) -- (7,4);
\draw (0,5) -- (7,5);
\draw (0,6) -- (7,6);
\draw (0,7) -- (7,7);
\draw (0,8) -- (7,8);
\draw (1,0) -- (1,9);
\draw (2,0) -- (2,9);
\draw (3,0) -- (3,9);
\draw (4,0) -- (4,9);
\draw (5,0) -- (5,9);
\draw (6,0) -- (6,9);
\draw (1.5,1.5) node {$\varnothing$};
\draw (1.5,7.5) node {$\varnothing$};
\draw (5.5,7.5) node {$\varnothing$};
\draw (5.5,1.5) node {$\varnothing$};
\draw (0.5,0.5) node {$\varnothing$};
\draw (6.5,8.5) node {$\varnothing$};
\Hpoint{1}{7};
\draw (0.5,6.75) node {{\tiny $i_{RG}$}};
\Hpoint{0.5}{6};
\draw (0.5,6.3) node {{\tiny $a$}};
\Vpoint{2}{8};
\draw (2,8.5) node {{\tiny $j_{RG}$}};
\Vpoint{3}{8.5};
\draw (2.7,8.5) node {{\tiny $b$}};
\Vpoint{5}{1};
\draw (5,0.6) node {{\tiny $i_{GR}$}};
\Vpoint{4}{0.5};
\draw (4.25,0.35) node {{\tiny $d$}};
\Hpoint{6}{2};
\draw (6.75,1.75) node {{\tiny $j_{GR}$}};
\Hpoint{6.5}{3};
\draw (6.5,2.65) node {{\tiny $c$}};
\Vpoint{1.5}{5};
\draw (1.2,4.85) node {{\tiny $x$}};
\Vpoint{5}{4};
\draw (4.8,4.2) node {{\tiny $y$}};
\draw (6.5,7.5) node {$\varnothing$};
\draw (1.5,0.5) node {$\varnothing$};
\draw (3.5,0.5) node {$\varnothing$};
\draw (0.5,1.5) node {$\varnothing$};
\draw (5.5,8.5) node {$\varnothing$};
\draw (3.5,3.5) node {$*$};
\draw (3.5,4.5) node {$*$};
\draw (3.5,5.5) node {$A$};
\draw (1.5,3.5) node {$\varnothing$};
\draw (0.5,4.5) node {$\varnothing$};
\draw (1.5,4.5) node {$\varnothing$};
\draw (2.5,4.5) node {$\varnothing$};
\draw (4.5,4.5) node {$\varnothing$};
\draw (5.5,4.5) node {$\varnothing$};
\draw (6.5,4.5) node {$\varnothing$};
\draw (0.5,5.5) node {$\varnothing$};
\draw (2.5,5.5) node {$\varnothing$};
\draw (6.5,5.5) node {$\varnothing$};
\draw (1.5,6.5) node {$\varnothing$};
\draw (3.5,6.5) node {$\varnothing$};
\draw (5.5,6.5) node {$\varnothing$};
\draw (6.5,6.5) node {$\varnothing$};
\draw (0.5,3.5) node {$\varnothing$};
\draw (2.5,0.5) node {$\varnothing$};
\draw (2.5,1.5) node {$\varnothing$};
\draw (2.5,3.5) node {$\varnothing$};
\draw (2.5,6.5) node {$\varnothing$};
\draw (2.5,7.5) node {$\varnothing$};
\draw (6.5,3.5) node {$\varnothing$};
\draw (3.5,7.5) node {$\varnothing$};
\draw (3.5,8.5) node {$\varnothing$};
\draw (0.5,2.5) node {$\varnothing$};
\draw (1.5,2.5) node {$\varnothing$};
\draw (2.5,2.5) node {$\varnothing$};
\draw (5.5,2.5) node {$\varnothing$};
\draw (5.5,5.5) node {$\varnothing$};
\draw (4.5,5.5) node {$\varnothing$};
\draw (4.5,1.5) node {$\varnothing$};
\draw (4.5,2.5) node {$\varnothing$};
\draw (4.5,6.5) node {$\varnothing$};
\draw (4.5,7.5) node {$\varnothing$};
\draw (4.5,8.5) node {$\varnothing$};
\draw (1.5,5.5) node {$1$};
\draw (5,3.5) node {$3$};
\draw (3.5,2) node {$4$};
\zoneRG{1}{8}{1}
\zoneGR{7}{0}{1}
\end{tikzpicture}
\begin{tikzpicture}[scale=.5]
\useasboundingbox (0,-0.5) (7.75,9.5);
\fill [Hfill] (0,6) rectangle (1,8);
\fill [Vfill] (1,8) rectangle (3,9);
\fill [Hfill] (6,1) rectangle (7,3);
\fill [Vfill] (4,0) rectangle (6,1);
\fill [Hfill] (3,1) rectangle (4,4);
\fill [Vfill] (1,5) rectangle (2,6);
\fill [Vfill] (4,3) rectangle (6,4);
\fill [Vfill] (3,5) rectangle (4,6);
\draw (0,1) -- (7,1);
\draw (0,2) -- (7,2);
\draw (0,3) -- (7,3);
\draw (0,4) -- (7,4);
\draw (0,5) -- (7,5);
\draw (0,6) -- (7,6);
\draw (0,7) -- (7,7);
\draw (0,8) -- (7,8);
\draw (1,0) -- (1,9);
\draw (2,0) -- (2,9);
\draw (3,0) -- (3,9);
\draw (4,0) -- (4,9);
\draw (5,0) -- (5,9);
\draw (6,0) -- (6,9);
\draw (1.5,1.5) node {$\varnothing$};
\draw (1.5,7.5) node {$\varnothing$};
\draw (5.5,7.5) node {$\varnothing$};
\draw (5.5,1.5) node {$\varnothing$};
\draw (0.5,0.5) node {$\varnothing$};
\draw (6.5,8.5) node {$\varnothing$};
\Hpoint{1}{7};
\draw (0.5,6.75) node {{\tiny $i_{RG}$}};
\Hpoint{0.5}{6};
\draw (0.5,6.3) node {{\tiny $a$}};
\Vpoint{2}{8};
\draw (2,8.5) node {{\tiny $j_{RG}$}};
\Vpoint{3}{8.5};
\draw (2.7,8.5) node {{\tiny $b$}};
\Vpoint{5}{1};
\draw (5,0.6) node {{\tiny $i_{GR}$}};
\Vpoint{4}{0.5};
\draw (4.25,0.35) node {{\tiny $d$}};
\Hpoint{6}{2};
\draw (6.75,1.75) node {{\tiny $j_{GR}$}};
\Hpoint{6.5}{3};
\draw (6.5,2.65) node {{\tiny $c$}};
\Vpoint{1.5}{5};
\draw (1.2,4.85) node {{\tiny $x$}};
\Vpoint{5}{4};
\draw (4.8,4.2) node {{\tiny $y$}};
\draw (6.5,7.5) node {$\varnothing$};
\draw (1.5,0.5) node {$\varnothing$};
\draw (3.5,0.5) node {$\varnothing$};
\draw (0.5,1.5) node {$\varnothing$};
\draw (5.5,8.5) node {$\varnothing$};
\draw (3.5,4.5) node {$*$};
\draw (3.5,5.5) node {$A$};
\draw (1.5,3.5) node {$\varnothing$};
\draw (0.5,4.5) node {$\varnothing$};
\draw (1.5,4.5) node {$\varnothing$};
\draw (2.5,4.5) node {$\varnothing$};
\draw (4.5,4.5) node {$\varnothing$};
\draw (5.5,4.5) node {$\varnothing$};
\draw (6.5,4.5) node {$\varnothing$};
\draw (0.5,5.5) node {$\varnothing$};
\draw (2.5,5.5) node {$\varnothing$};
\draw (6.5,5.5) node {$\varnothing$};
\draw (1.5,6.5) node {$\varnothing$};
\draw (3.5,6.5) node {$\varnothing$};
\draw (5.5,6.5) node {$\varnothing$};
\draw (6.5,6.5) node {$\varnothing$};
\draw (0.5,3.5) node {$\varnothing$};
\draw (2.5,0.5) node {$\varnothing$};
\draw (2.5,1.5) node {$\varnothing$};
\draw (2.5,3.5) node {$\varnothing$};
\draw (2.5,6.5) node {$\varnothing$};
\draw (2.5,7.5) node {$\varnothing$};
\draw (6.5,3.5) node {$\varnothing$};
\draw (3.5,7.5) node {$\varnothing$};
\draw (3.5,8.5) node {$\varnothing$};
\draw (0.5,2.5) node {$\varnothing$};
\draw (1.5,2.5) node {$\varnothing$};
\draw (2.5,2.5) node {$\varnothing$};
\draw (5.5,2.5) node {$\varnothing$};
\draw (5.5,5.5) node {$\varnothing$};
\draw (4.5,5.5) node {$\varnothing$};
\draw (4.5,1.5) node {$\varnothing$};
\draw (4.5,2.5) node {$\varnothing$};
\draw (4.5,6.5) node {$\varnothing$};
\draw (4.5,7.5) node {$\varnothing$};
\draw (4.5,8.5) node {$\varnothing$};
\draw (1.5,5.5) node {$1$};
\draw (5,3.5) node {$3$};
\draw (3.5,2.5) node {$4$};
\zoneRG{1}{8}{1}
\zoneGR{7}{0}{1}
\end{tikzpicture}
\caption{Zone $1$ is not empty and zone $3$ is not empty.\label{fig:RGGR3notempty}}
\end{center}
\end{figure}

Rule (\rmnum{1}) applied to $d$ and $y$ gives the second diagram.
But ($i_{RG}$, $j_{RG}$) is the bottom-rightmost increasing sequence $RG$, hence no point in the lower left quadrant of $y$ lies in \R. 
Hence zone $4$ is empty and $\sigma$ is $\ominus$-decomposable which is forbidden.

\paragraph{Zone $3$ is empty}

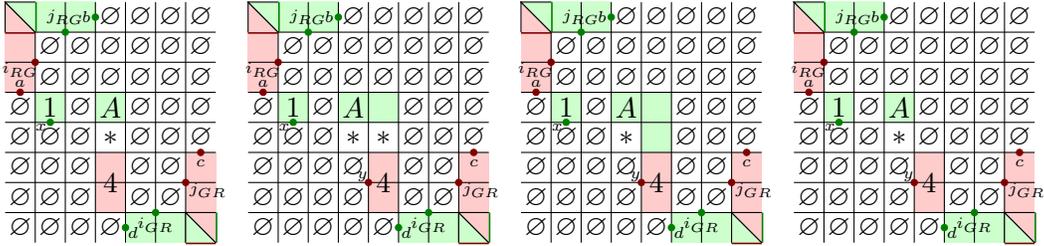
\begin{figure}[H]
\begin{center}
\begin{tikzpicture}[scale=.4]
\useasboundingbox (0,-1) (7.75,8.5);
\fill [Hfill] (0,5) rectangle (1,7);
\fill [Vfill] (1,7) rectangle (3,8);
\fill [Hfill] (6,1) rectangle (7,3);
\fill [Vfill] (4,0) rectangle (6,1);
\fill [Hfill] (3,1) rectangle (4,3);
\fill [Vfill] (1,4) rectangle (2,5);
\fill [Vfill] (3,4) rectangle (4,5);
\draw (0,1) -- (7,1);
\draw (0,2) -- (7,2);
\draw (0,3) -- (7,3);
\draw (0,4) -- (7,4);
\draw (0,5) -- (7,5);
\draw (0,6) -- (7,6);
\draw (0,7) -- (7,7);
\draw (1,0) -- (1,8);
\draw (2,0) -- (2,8);
\draw (3,0) -- (3,8);
\draw (4,0) -- (4,8);
\draw (5,0) -- (5,8);
\draw (6,0) -- (6,8);
\draw (1.5,1.5) node {$\varnothing$};
\draw (1.5,6.5) node {$\varnothing$};
\draw (5.5,6.5) node {$\varnothing$};
\draw (5.5,1.5) node {$\varnothing$};
\draw (0.5,0.5) node {$\varnothing$};
\draw (6.5,7.5) node {$\varnothing$};
\Hpoint{1}{6};
\draw (0.5,5.75) node {{\tiny $i_{RG}$}};
\Hpoint{0.5}{5};
\draw (0.5,5.3) node {{\tiny $a$}};
\Vpoint{2}{7};
\draw (2,7.5) node {{\tiny $j_{RG}$}};
\Vpoint{3}{7.5};
\draw (2.7,7.5) node {{\tiny $b$}};
\Vpoint{5}{1};
\draw (5,0.6) node {{\tiny $i_{GR}$}};
\Vpoint{4}{0.5};
\draw (4.25,0.35) node {{\tiny $d$}};
\Hpoint{6}{2};
\draw (6.75,1.75) node {{\tiny $j_{GR}$}};
\Hpoint{6.5}{3};
\draw (6.5,2.65) node {{\tiny $c$}};
\Vpoint{1.5}{4};
\draw (1.2,3.85) node {{\tiny $x$}};
\draw (6.5,6.5) node {$\varnothing$};
\draw (1.5,0.5) node {$\varnothing$};
\draw (3.5,0.5) node {$\varnothing$};
\draw (0.5,1.5) node {$\varnothing$};
\draw (5.5,7.5) node {$\varnothing$};
\draw (3.5,3.5) node {$*$};
\draw (3.5,4.5) node {$A$};
\draw (1.5,3.5) node {$\varnothing$};
\draw (0.5,4.5) node {$\varnothing$};
\draw (2.5,4.5) node {$\varnothing$};
\draw (6.5,4.5) node {$\varnothing$};
\draw (1.5,5.5) node {$\varnothing$};
\draw (3.5,5.5) node {$\varnothing$};
\draw (5.5,5.5) node {$\varnothing$};
\draw (6.5,5.5) node {$\varnothing$};
\draw (0.5,3.5) node {$\varnothing$};
\draw (2.5,0.5) node {$\varnothing$};
\draw (2.5,1.5) node {$\varnothing$};
\draw (2.5,3.5) node {$\varnothing$};
\draw (2.5,5.5) node {$\varnothing$};
\draw (2.5,6.5) node {$\varnothing$};
\draw (6.5,3.5) node {$\varnothing$};
\draw (3.5,6.5) node {$\varnothing$};
\draw (3.5,7.5) node {$\varnothing$};
\draw (0.5,2.5) node {$\varnothing$};
\draw (1.5,2.5) node {$\varnothing$};
\draw (2.5,2.5) node {$\varnothing$};
\draw (5.5,2.5) node {$\varnothing$};
\draw (5.5,3.5) node {$\varnothing$};
\draw (5.5,4.5) node {$\varnothing$};
\draw (4.5,4.5) node {$\varnothing$};
\draw (4.5,1.5) node {$\varnothing$};
\draw (4.5,2.5) node {$\varnothing$};
\draw (4.5,3.5) node {$\varnothing$};
\draw (4.5,5.5) node {$\varnothing$};
\draw (4.5,6.5) node {$\varnothing$};
\draw (4.5,7.5) node {$\varnothing$};
\draw (1.5,4.5) node {$1$};
\draw (3.5,2) node {$4$};
\zoneRG{1}{7}{1}
\zoneGR{7}{0}{1}
\end{tikzpicture}
\begin{tikzpicture}[scale=.4]
\useasboundingbox (0,-1) (8.75,8.5);
\fill [Hfill] (0,5) rectangle (1,7);
\fill [Vfill] (1,7) rectangle (3,8);
\fill [Hfill] (7,1) rectangle (8,3);
\fill [Vfill] (5,0) rectangle (7,1);
\fill [Hfill] (4,1) rectangle (5,3);
\fill [Vfill] (1,4) rectangle (2,5);
\fill [Vfill] (3,4) rectangle (5,5);
\draw (0,1) -- (8,1);
\draw (0,2) -- (8,2);
\draw (0,3) -- (8,3);
\draw (0,4) -- (8,4);
\draw (0,5) -- (8,5);
\draw (0,6) -- (8,6);
\draw (0,7) -- (8,7);
\draw (1,0) -- (1,8);
\draw (2,0) -- (2,8);
\draw (3,0) -- (3,8);
\draw (4,0) -- (4,8);
\draw (5,0) -- (5,8);
\draw (6,0) -- (6,8);
\draw (7,0) -- (7,8);
\draw (1.5,1.5) node {$\varnothing$};
\draw (1.5,6.5) node {$\varnothing$};
\draw (6.5,6.5) node {$\varnothing$};
\draw (6.5,1.5) node {$\varnothing$};
\draw (0.5,0.5) node {$\varnothing$};
\draw (7.5,7.5) node {$\varnothing$};
\Hpoint{1}{6};
\draw (0.5,5.75) node {{\tiny $i_{RG}$}};
\Hpoint{0.5}{5};
\draw (0.5,5.3) node {{\tiny $a$}};
\Vpoint{2}{7};
\draw (2,7.5) node {{\tiny $j_{RG}$}};
\Vpoint{3}{7.5};
\draw (2.7,7.5) node {{\tiny $b$}};
\Vpoint{6}{1};
\draw (6,0.6) node {{\tiny $i_{GR}$}};
\Vpoint{5}{0.5};
\draw (5.25,0.35) node {{\tiny $d$}};
\Hpoint{7}{2};
\draw (7.75,1.75) node {{\tiny $j_{GR}$}};
\Hpoint{7.5}{3};
\draw (7.5,2.65) node {{\tiny $c$}};
\Vpoint{1.5}{4};
\draw (1.2,3.85) node {{\tiny $x$}};
\Hpoint{4}{2};
\draw (3.8,2.2) node {{\tiny $y$}};
\draw (7.5,6.5) node {$\varnothing$};
\draw (1.5,0.5) node {$\varnothing$};
\draw (3.5,0.5) node {$\varnothing$};
\draw (0.5,1.5) node {$\varnothing$};
\draw (6.5,7.5) node {$\varnothing$};
\draw (3.5,3.5) node {$*$};
\draw (3.5,4.5) node {$A$};
\draw (1.5,3.5) node {$\varnothing$};
\draw (0.5,4.5) node {$\varnothing$};
\draw (2.5,4.5) node {$\varnothing$};
\draw (7.5,4.5) node {$\varnothing$};
\draw (1.5,5.5) node {$\varnothing$};
\draw (3.5,5.5) node {$\varnothing$};
\draw (6.5,5.5) node {$\varnothing$};
\draw (7.5,5.5) node {$\varnothing$};
\draw (0.5,3.5) node {$\varnothing$};
\draw (2.5,0.5) node {$\varnothing$};
\draw (2.5,1.5) node {$\varnothing$};
\draw (2.5,3.5) node {$\varnothing$};
\draw (2.5,5.5) node {$\varnothing$};
\draw (2.5,6.5) node {$\varnothing$};
\draw (7.5,3.5) node {$\varnothing$};
\draw (3.5,6.5) node {$\varnothing$};
\draw (3.5,7.5) node {$\varnothing$};
\draw (0.5,2.5) node {$\varnothing$};
\draw (1.5,2.5) node {$\varnothing$};
\draw (2.5,2.5) node {$\varnothing$};
\draw (6.5,2.5) node {$\varnothing$};
\draw (6.5,3.5) node {$\varnothing$};
\draw (6.5,4.5) node {$\varnothing$};
\draw (5.5,4.5) node {$\varnothing$};
\draw (5.5,1.5) node {$\varnothing$};
\draw (5.5,2.5) node {$\varnothing$};
\draw (5.5,3.5) node {$\varnothing$};
\draw (5.5,5.5) node {$\varnothing$};
\draw (5.5,6.5) node {$\varnothing$};
\draw (5.5,7.5) node {$\varnothing$};
\draw (4.5,0.5) node {$\varnothing$};
\draw (3.5,1.5) node {$\varnothing$};
\draw (3.5,2.5) node {$\varnothing$};
\draw (4.5,3.5) node {$*$};
\draw (4.5,5.5) node {$\varnothing$};
\draw (4.5,6.5) node {$\varnothing$};
\draw (4.5,7.5) node {$\varnothing$};
\draw (1.5,4.5) node {$1$};
\draw (4.5,2) node {$4$};
\zoneRG{1}{7}{1}
\zoneGR{8}{0}{1}
\end{tikzpicture}
\begin{tikzpicture}[scale=.4]
\useasboundingbox (0,-1) (8.75,8.5);
\fill [Hfill] (0,5) rectangle (1,7);
\fill [Vfill] (1,7) rectangle (3,8);
\fill [Hfill] (7,1) rectangle (8,3);
\fill [Vfill] (5,0) rectangle (7,1);
\fill [Hfill] (4,1) rectangle (5,3);
\fill [Vfill] (1,4) rectangle (2,5);
\fill [Vfill] (3,4) rectangle (5,5);
\fill [Vfill] (4,3) rectangle (5,4);
\draw (0,1) -- (8,1);
\draw (0,2) -- (8,2);
\draw (0,3) -- (8,3);
\draw (0,4) -- (8,4);
\draw (0,5) -- (8,5);
\draw (0,6) -- (8,6);
\draw (0,7) -- (8,7);
\draw (1,0) -- (1,8);
\draw (2,0) -- (2,8);
\draw (3,0) -- (3,8);
\draw (4,0) -- (4,8);
\draw (5,0) -- (5,8);
\draw (6,0) -- (6,8);
\draw (7,0) -- (7,8);
\draw (1.5,1.5) node {$\varnothing$};
\draw (1.5,6.5) node {$\varnothing$};
\draw (6.5,6.5) node {$\varnothing$};
\draw (6.5,1.5) node {$\varnothing$};
\draw (0.5,0.5) node {$\varnothing$};
\draw (7.5,7.5) node {$\varnothing$};
\Hpoint{1}{6};
\draw (0.5,5.75) node {{\tiny $i_{RG}$}};
\Hpoint{0.5}{5};
\draw (0.5,5.3) node {{\tiny $a$}};
\Vpoint{2}{7};
\draw (2,7.5) node {{\tiny $j_{RG}$}};
\Vpoint{3}{7.5};
\draw (2.7,7.5) node {{\tiny $b$}};
\Vpoint{6}{1};
\draw (6,0.6) node {{\tiny $i_{GR}$}};
\Vpoint{5}{0.5};
\draw (5.25,0.35) node {{\tiny $d$}};
\Hpoint{7}{2};
\draw (7.75,1.75) node {{\tiny $j_{GR}$}};
\Hpoint{7.5}{3};
\draw (7.5,2.65) node {{\tiny $c$}};
\Vpoint{1.5}{4};
\draw (1.2,3.85) node {{\tiny $x$}};
\Hpoint{4}{2};
\draw (3.8,2.2) node {{\tiny $y$}};
\draw (7.5,6.5) node {$\varnothing$};
\draw (1.5,0.5) node {$\varnothing$};
\draw (3.5,0.5) node {$\varnothing$};
\draw (0.5,1.5) node {$\varnothing$};
\draw (6.5,7.5) node {$\varnothing$};
\draw (3.5,3.5) node {$*$};
\draw (3.5,4.5) node {$A$};
\draw (1.5,3.5) node {$\varnothing$};
\draw (0.5,4.5) node {$\varnothing$};
\draw (2.5,4.5) node {$\varnothing$};
\draw (7.5,4.5) node {$\varnothing$};
\draw (1.5,5.5) node {$\varnothing$};
\draw (3.5,5.5) node {$\varnothing$};
\draw (6.5,5.5) node {$\varnothing$};
\draw (7.5,5.5) node {$\varnothing$};
\draw (0.5,3.5) node {$\varnothing$};
\draw (2.5,0.5) node {$\varnothing$};
\draw (2.5,1.5) node {$\varnothing$};
\draw (2.5,3.5) node {$\varnothing$};
\draw (2.5,5.5) node {$\varnothing$};
\draw (2.5,6.5) node {$\varnothing$};
\draw (7.5,3.5) node {$\varnothing$};
\draw (3.5,6.5) node {$\varnothing$};
\draw (3.5,7.5) node {$\varnothing$};
\draw (0.5,2.5) node {$\varnothing$};
\draw (1.5,2.5) node {$\varnothing$};
\draw (2.5,2.5) node {$\varnothing$};
\draw (6.5,2.5) node {$\varnothing$};
\draw (6.5,3.5) node {$\varnothing$};
\draw (6.5,4.5) node {$\varnothing$};
\draw (5.5,4.5) node {$\varnothing$};
\draw (5.5,1.5) node {$\varnothing$};
\draw (5.5,2.5) node {$\varnothing$};
\draw (5.5,3.5) node {$\varnothing$};
\draw (5.5,5.5) node {$\varnothing$};
\draw (5.5,6.5) node {$\varnothing$};
\draw (5.5,7.5) node {$\varnothing$};
\draw (4.5,0.5) node {$\varnothing$};
\draw (3.5,1.5) node {$\varnothing$};
\draw (3.5,2.5) node {$\varnothing$};
\draw (4.5,5.5) node {$\varnothing$};
\draw (4.5,6.5) node {$\varnothing$};
\draw (4.5,7.5) node {$\varnothing$};
\draw (1.5,4.5) node {$1$};
\draw (4.5,2) node {$4$};
\zoneRG{1}{7}{1}
\zoneGR{8}{0}{1}
\end{tikzpicture}
\begin{tikzpicture}[scale=.4]
\useasboundingbox (0,-1) (7.75,8.5);
\fill [Hfill] (0,5) rectangle (1,7);
\fill [Vfill] (1,7) rectangle (3,8);
\fill [Hfill] (7,1) rectangle (8,3);
\fill [Vfill] (5,0) rectangle (7,1);
\fill [Hfill] (4,1) rectangle (5,3);
\fill [Vfill] (1,4) rectangle (2,5);
\fill [Vfill] (3,4) rectangle (4,5);
\draw (0,1) -- (8,1);
\draw (0,2) -- (8,2);
\draw (0,3) -- (8,3);
\draw (0,4) -- (8,4);
\draw (0,5) -- (8,5);
\draw (0,6) -- (8,6);
\draw (0,7) -- (8,7);
\draw (1,0) -- (1,8);
\draw (2,0) -- (2,8);
\draw (3,0) -- (3,8);
\draw (4,0) -- (4,8);
\draw (5,0) -- (5,8);
\draw (6,0) -- (6,8);
\draw (7,0) -- (7,8);
\draw (1.5,1.5) node {$\varnothing$};
\draw (1.5,6.5) node {$\varnothing$};
\draw (6.5,6.5) node {$\varnothing$};
\draw (6.5,1.5) node {$\varnothing$};
\draw (0.5,0.5) node {$\varnothing$};
\draw (7.5,7.5) node {$\varnothing$};
\Hpoint{1}{6};
\draw (0.5,5.75) node {{\tiny $i_{RG}$}};
\Hpoint{0.5}{5};
\draw (0.5,5.3) node {{\tiny $a$}};
\Vpoint{2}{7};
\draw (2,7.5) node {{\tiny $j_{RG}$}};
\Vpoint{3}{7.5};
\draw (2.7,7.5) node {{\tiny $b$}};
\Vpoint{6}{1};
\draw (6,0.6) node {{\tiny $i_{GR}$}};
\Vpoint{5}{0.5};
\draw (5.25,0.35) node {{\tiny $d$}};
\Hpoint{7}{2};
\draw (7.75,1.75) node {{\tiny $j_{GR}$}};
\Hpoint{7.5}{3};
\draw (7.5,2.65) node {{\tiny $c$}};
\Vpoint{1.5}{4};
\draw (1.2,3.85) node {{\tiny $x$}};
\Hpoint{4}{2};
\draw (3.8,2.2) node {{\tiny $y$}};
\draw (7.5,6.5) node {$\varnothing$};
\draw (1.5,0.5) node {$\varnothing$};
\draw (3.5,0.5) node {$\varnothing$};
\draw (0.5,1.5) node {$\varnothing$};
\draw (6.5,7.5) node {$\varnothing$};
\draw (3.5,3.5) node {$*$};
\draw (3.5,4.5) node {$A$};
\draw (1.5,3.5) node {$\varnothing$};
\draw (0.5,4.5) node {$\varnothing$};
\draw (2.5,4.5) node {$\varnothing$};
\draw (7.5,4.5) node {$\varnothing$};
\draw (1.5,5.5) node {$\varnothing$};
\draw (3.5,5.5) node {$\varnothing$};
\draw (6.5,5.5) node {$\varnothing$};
\draw (7.5,5.5) node {$\varnothing$};
\draw (0.5,3.5) node {$\varnothing$};
\draw (2.5,0.5) node {$\varnothing$};
\draw (2.5,1.5) node {$\varnothing$};
\draw (2.5,3.5) node {$\varnothing$};
\draw (2.5,5.5) node {$\varnothing$};
\draw (2.5,6.5) node {$\varnothing$};
\draw (7.5,3.5) node {$\varnothing$};
\draw (3.5,6.5) node {$\varnothing$};
\draw (3.5,7.5) node {$\varnothing$};
\draw (0.5,2.5) node {$\varnothing$};
\draw (1.5,2.5) node {$\varnothing$};
\draw (2.5,2.5) node {$\varnothing$};
\draw (6.5,2.5) node {$\varnothing$};
\draw (6.5,3.5) node {$\varnothing$};
\draw (6.5,4.5) node {$\varnothing$};
\draw (5.5,4.5) node {$\varnothing$};
\draw (5.5,1.5) node {$\varnothing$};
\draw (5.5,2.5) node {$\varnothing$};
\draw (5.5,3.5) node {$\varnothing$};
\draw (5.5,5.5) node {$\varnothing$};
\draw (5.5,6.5) node {$\varnothing$};
\draw (5.5,7.5) node {$\varnothing$};
\draw (4.5,0.5) node {$\varnothing$};
\draw (3.5,1.5) node {$\varnothing$};
\draw (3.5,2.5) node {$\varnothing$};
\draw (4.5,3.5) node {$\varnothing$};
\draw (4.5,4.5) node {$\varnothing$};
\draw (4.5,5.5) node {$\varnothing$};
\draw (4.5,6.5) node {$\varnothing$};
\draw (4.5,7.5) node {$\varnothing$};
\draw (1.5,4.5) node {$1$};
\draw (4.5,2) node {$4$};
\zoneRG{1}{7}{1}
\zoneGR{8}{0}{1}
\end{tikzpicture}

\caption{Zone $1$ is not empty and zone $3$ is empty\label{fig:RGGR3empty}}
\end{center}
\end{figure}
Figure~\ref{fig:RGGR3empty} illustrates the proof. As $\sigma$ is $\ominus$-indecomposable, zone $4$ is not empty. 
Let $y$ be the leftmost point inside zone $4$ -- either above or under $j_{GR}$ -- as depicted in the second diagram.
Rule (\rmnum{2}) applied to $y$ and $c$ leads to the third diagram.
But ($i_{RG}$, $j_{RG}$) is the bottom-rightmost increasing sequence RG, hence no point of \G lies in the top-right quadrant of $y$ leading to the fourth diagram.
So $\sigma$ is $\ominus$-decomposable which is forbidden.

This ends the cases study, proving that zone $A$ and $B$ cannot be both empty.

\end{pf}

\begin{defn}\label{def:C_*}
Let $\sigma$ be a permutation and $i,j,k,\ell$ four indices of $\sigma$ such that $\sigma_i \sigma_j$ and $\sigma_k \sigma_\ell$ are ascents.
We define the partial bicoloring $C_*(\sigma, i,j,k,\ell)$ of $\sigma$ as follows.

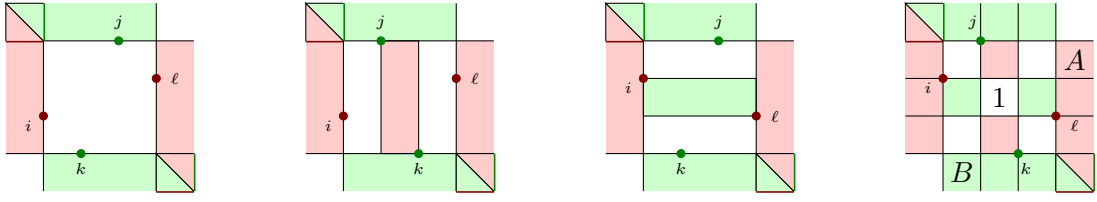
\begin{tikzpicture}[scale=.5]%3
\useasboundingbox (0,-1) (7.75,6);
\fill [Hfill] (0,1) rectangle (1,4);
\fill [Vfill] (1,4) rectangle (4,5);
\fill [Hfill] (4,1) rectangle (5,4);
\fill [Vfill] (1,0) rectangle (4,1);
\draw (0,1) -- (5,1);
\draw (0,4) -- (5,4);
\draw (1,0) -- (1,5);
\draw (4,0) -- (4,5);
\Hpoint{1}{2};
\draw (0.6,1.75) node {{\tiny $i$}};
\Vpoint{3}{4};
\draw (3,4.5) node {{\tiny $j$}};
\Vpoint{2}{1};
\draw (2,0.6) node {{\tiny $k$}};
\Hpoint{4}{3};
\draw (4.5,3) node {{\tiny $\ell$}};
\zoneRG{1}{4}{1}
\zoneGR{5}{0}{1}
\end{tikzpicture}
\begin{tikzpicture}[scale=.5]%4
\useasboundingbox (0,-1) (7.75,6);
\fill [Hfill] (0,1) rectangle (1,4);
\fill [Vfill] (1,4) rectangle (4,5);
\fill [Hfill] (4,1) rectangle (5,4);
\fill [Vfill] (1,0) rectangle (4,1);
\draw [Hfill] (2,1) rectangle (3,4);
\draw (0,1) -- (5,1);
\draw (0,4) -- (5,4);
\draw (1,0) -- (1,5);
\draw (4,0) -- (4,5);
\Hpoint{1}{2};
\draw (0.6,1.75) node {{\tiny $i$}};
\Vpoint{2}{4};
\draw (2,4.5) node {{\tiny $j$}};
\Vpoint{3}{1};
\draw (3,0.6) node {{\tiny $k$}};
\Hpoint{4}{3};
\draw (4.5,3) node {{\tiny $\ell$}};
\zoneRG{1}{4}{1}
\zoneGR{5}{0}{1}
\end{tikzpicture}
\begin{tikzpicture}[scale=.5]%2
\useasboundingbox (0,-1) (7.75,6);
\fill [Hfill] (0,1) rectangle (1,4);
\fill [Vfill] (1,4) rectangle (4,5);
\fill [Hfill] (4,1) rectangle (5,4);
\fill [Vfill] (1,0) rectangle (4,1);
\draw [Vfill] (1,2) rectangle (4,3);
\draw (0,1) -- (5,1);
\draw (0,4) -- (5,4);
\draw (1,0) -- (1,5);
\draw (4,0) -- (4,5);
\Hpoint{1}{3};
\draw (0.6,2.75) node {{\tiny $i$}};
\Vpoint{3}{4};
\draw (3,4.5) node {{\tiny $j$}};
\Vpoint{2}{1};
\draw (2,0.6) node {{\tiny $k$}};
\Hpoint{4}{2};
\draw (4.5,2) node {{\tiny $\ell$}};
\zoneRG{1}{4}{1}
\zoneGR{5}{0}{1}
\end{tikzpicture}
\begin{tikzpicture}[scale=.5]%1
\useasboundingbox (0,-1) (6,6);
\fill [Hfill] (0,1) rectangle (1,4);
\fill [Vfill] (1,4) rectangle (4,5);
\fill [Hfill] (4,1) rectangle (5,4);
\fill [Vfill] (1,0) rectangle (4,1);
\fill [Hfill] (2,1) rectangle (3,2);
\fill [Vfill] (1,2) rectangle (2,3);
\fill [Hfill] (2,3) rectangle (3,4);
\fill [Vfill] (3,2) rectangle (4,3);
\draw (0,1) -- (5,1);
\draw (0,2) -- (5,2);
\draw (0,3) -- (5,3);
\draw (0,4) -- (5,4);
\draw (1,0) -- (1,5);
\draw (2,0) -- (2,5);
\draw (3,0) -- (3,5);
\draw (4,0) -- (4,5);
\Hpoint{1}{3};
\draw (0.6,2.75) node {{\tiny $i$}};
\Vpoint{2}{4};
\draw (1.8,4.5) node {{\tiny $j$}};
\Vpoint{3}{1};
\draw (3.2,0.6) node {{\tiny $k$}};
\Hpoint{4}{2};
\draw (4.5,1.75) node {{\tiny $\ell$}};
\draw (4.5,3.5) node {$A$};
\draw (1.5,0.5) node {$B$};
\draw (2.5,2.5) node {{$1$}};
\zoneRG{1}{4}{1}
\zoneGR{5}{0}{1}
\end{tikzpicture}
If $\sigma_i$, $\sigma_\ell$, $\sigma_k$ and $\sigma_j$ have a relative position corresponding to one of the above diagrams, 
then we define $C_*(\sigma, i,j,k,\ell)$ as the partial bicoloring of $\sigma$ having the corresponding shape,
where in the first diagram points of zone $1$ are in \R if zone $A$ is nonempty, 
in \G if zones $B$ is nonempty, and have no color otherwise.

Otherwise $C_*(\sigma, i,j,k,\ell)$ is the partial coloring with no point colored.
\end{defn}

\begin{prop}\label{prop:C_*}
Let $\sigma$ be a $\ominus$-indecomposable permutation and $c$ a valid coloring of $\sigma$ such that there exist increasing sequences RG and increasing sequences GR in $c$. 
Then $c = C_*(\sigma, i_{RG}, j_{RG}, i_{GR}, j_{GR})$.
\end{prop}

\begin{pf}
This is a consequence of Proposition~\ref{prop:diagrammes*} and Definition~\ref{def:C_*}, noticing that if there exists 
a point $x$ in zone $A$, then applying rule (\rmnum{2}) to $\ell$ and $x$, all points in zone $1$ belong to \R, and
if there exists a point $y$ in zone $B$, then applying rule (\rmnum{1}) to $y$ and $k$, all points in zone $1$ belong to \G.
\end{pf}

\subsection{A first polynomial algorithm}

\begin{algorithm}[H]
% \SetAlgoLined
  \KwData{$\sigma$ a $\ominus$-indecomposable permutation (whose size is denoted $n$).}
  \KwResult{The set $E$ of valid colorings of $\sigma$}

  \For{$c$ bicoloring of $\sigma$ being one of\\ 
\hspace*{.3cm} $c$ is unicolor \R\\
\hspace*{.3cm} $c$ is unicolor \G \\
\hspace*{.3cm} $c = C_{GR}(\sigma, i,j)$ or $C_{RG}(\sigma, i,j)$ for $i \in [1..n]$ and $j \in [i..n]$ s.t. $\sigma_j > \sigma_i$\\
\hspace*{.3cm} $c = C_*(\sigma, i,j,k,\ell)$ for $i\in [1..n]$ and $j \in [i..n]$ s.t. $\sigma_j > \sigma_i$ and\\
\hspace*{3.4cm} for $k \in [1..n]$ and $\ell \in [k..n]$ s.t. $\sigma_\ell > \sigma_k$\\ 
}{
 	If all points of $\sigma$ are colored and $c$ is valid then add $c$ to $E$\;
     }
  \caption{ColoringIndecomposable1$(\sigma)$\label{alg:Valid-colorings-indec_naif}}
\end{algorithm}

\begin{prop}
Algorithm~\ref{alg:Valid-colorings-indec_naif} compute in time ${\mathcal O}(n^5)$ the set of valid colorings of any $\ominus$-indecomposable permutation $\sigma$.
\end{prop}

\begin{pf}
Let $\sigma$ be a $\ominus$-indecomposable permutation of size $n$ and $c$ a valid coloring of $\sigma$.
Then from Propositions~\ref{prop:monochromatic}, \ref{prop:C_GR}, \ref{prop:C_RG} and \ref{prop:C_*}, 
$c$ is either monochromatic, or $C_{GR}(\sigma, i,j)$ or $C_{RG}(\sigma, i,j)$ for some $i \in [1..n]$ and some $j \in [i..n]$ such that $\sigma_j > \sigma_i$, 
or $c = C_*(\sigma, i,j,k,\ell)$ for some $i \in [1..n]$, some $j \in [i..n]$ such that $\sigma_j > \sigma_i$, 
some $k \in [1..n]$ and some $\ell \in [k..n]$ such that $\sigma_\ell > \sigma_k$.
Thus $c$ is computed by Algorithm~\ref{alg:Valid-colorings-indec_naif} and added to $E$ as it is valid.
Conversely, each coloring added to $E$ is a valid bicoloring of $\sigma$.

Now consider the complexity of Algorithm~\ref{alg:Valid-colorings-indec_naif}.
There are ${\mathcal O}(n^{4})$ colorings computed. 
Indeed there are two monochromatic colorings, ${\mathcal O}(n^{2})$ colorings $C_{GR}(\sigma, i,j)$ or $C_{RG}(\sigma, i,j)$
and ${\mathcal O}(n^{4})$ colorings $C_*(\sigma, i,j,k,\ell)$.
Moreover the coloring is computed in linear time and checking if the coloring is valid is done in linear time 
using Proposition~\ref{prop:Check-Valid-linear}.
Hence Algorithm~\ref{alg:Valid-colorings-indec_naif} runs in time ${\mathcal O}(n^{5})$.
\end{pf}

\section{An optimal algorithm} \label{sec:optimalAlgo}

\subsection{Rooting colorings} \label{ssec:root}

In this section we show how each diagram of Propositions~\ref{prop:diagrammesGR}, \ref{prop:diagrammesRG} and \ref{prop:diagrammes*} can be rooted in a given point such that each point $i_{GR},i_{RG},j_{GR}$ and $j_{RG}$ can be deduced from this one. 
Moreover, given a diagram we show how we can assign colors to points of the permutations lying in a colored zone of the diagram in linear time.

\begin{defn}\label{def:C_numero}
Let $\sigma$ be a permutation and $s \in [1..|\sigma|]$. We set \vspace{0.2 cm}\\
$C_1(\sigma, s) = C_{GR}(\sigma, s, t)$ where $t = \min \{k \mid k > s \text{ and } \sigma_k > \sigma_s\}$ \\
$C_2(\sigma, s) = C_{GR}(\sigma, t, s)$ where $t$ is such that $\sigma_t = \max \{\sigma_k \mid k < s \text{ and } \sigma_k < \sigma_s\}$\\
$C_3(\sigma, s) = C_{RG}(\sigma, s, t)$ where $t$ is such that $\sigma_t = \min \{\sigma_k \mid k > s \text{ and } \sigma_k > \sigma_s\}$\\
$C_4(\sigma, s) = C_{RG}(\sigma, t, s)$ where $t = \max \{k \mid k < s \text{ and } \sigma_k < \sigma_s\}$\\

\noindent
\begin{minipage}{.75\textwidth}\setlength{\parindent}{2em}
\noindent $C_5(\sigma, s) = C_*(\sigma, p,q,t,s)$ with \\
 \indent $t = \max \{k \mid k < u \text{ and } \sigma_k < \sigma_s\}$\\
 \indent with $u = \max \{k \mid k < s \text{ and } \sigma_k > \sigma_s\}$,\\
 \indent $p = \max \{k \mid k < t \text{ and } \sigma_t < \sigma_k < \sigma_s\}$ and \\
 \indent  $q$  such that $\sigma_q = \min \{\sigma_k \mid t < k \leq u \}$\\
\end{minipage}
\begin{minipage}{.25\textwidth}
\begin{tikzpicture}[scale=.5]
\useasboundingbox (1,0) (7.25,5);
\fill [Grisfill] (0,1) rectangle (1,5);
\fill [Grisfill] (1,3) rectangle (2,4);
\fill [Grisfill] (1,4) rectangle (4,5);
\fill [Grisfill] (4,1) rectangle (5,3);
\fill [Grisfill] (5,0) rectangle (6,5);
\fill [Grisfill] (0,0) rectangle (5,1);
\draw [dotted] (0,1) -- (6,1);
\draw [dotted] (1,3) -- (5,3);
\draw [dotted] (1,4) -- (5,4);
\draw [dotted] (1,1) -- (1,4);
\draw [dotted] (2,1) -- (2,4);
\draw [dotted] (4,1) -- (4,5);
\draw [dotted] (5,1) -- (5,5);
\draw (1.5,2) node {$\varnothing$};
\draw (3,2) node {$\varnothing$};
\draw (3,3.5) node {$\varnothing$};
\draw (4.5,3.5) node {$\varnothing$};
\draw (4.5,4.5) node {$\varnothing$};
\draw (1,2) [fill] circle (3pt);
\draw (0.5,1.8) node {{\tiny $p$}};
\draw (3,4) [fill] circle (3pt);
\draw (2.7,4.4) node {{\tiny $q$}};
\draw (4,4.5) [fill] circle (3pt);
\draw (4,4.5) node[left] {{\tiny $u$}};
\draw (2,1) [fill] circle (3pt);
\draw (2,0.6) node {{\tiny $t$}};
\draw (5,3) [fill] circle (3pt);
\draw (5.5,3) node {{\tiny $s$}};
\end{tikzpicture}
\end{minipage}
\begin{minipage}{.75\textwidth}\setlength{\parindent}{2em}
\noindent $C_6(\sigma, s) = C_*(\sigma, p,q,s,t)$ with \\
 \indent $t = \min \{k \mid k > s \text{ and } \sigma_k > \sigma_s\}$,\\
 \indent $u = \max \{k \mid k < t \text{ and } \sigma_k > \sigma_t\}$,\\
 \indent $p = \max \{k \mid k < u \text{ and } \sigma_s < \sigma_k < \sigma_t\}$ and\\
 \indent $q$  such that $\sigma_q = \min \{\sigma_k \mid \sigma_k > \sigma_t \text{ and } p < k \leq u \}$\\
\end{minipage}
\begin{minipage}{.25\textwidth}
\begin{tikzpicture}[scale=.5]
\useasboundingbox (1,0) (5,5);
\fill [Grisfill] (0,1) rectangle (1,5);
\fill [Grisfill] (1,4) rectangle (3,5);
\fill [Grisfill] (5,1) rectangle (6,5);
\fill [Grisfill] (0,0) rectangle (6,1);
\fill [Grisfill] (3,1) rectangle (4,3);
\draw [dotted] (0,1) -- (6,1);
\draw [dotted] (1,3) -- (5,3);
\draw [dotted] (1,4) -- (3,4);
\draw [dotted] (1,1) -- (1,4);
\draw [dotted] (3,1) -- (3,5);
\draw [dotted] (4,1) -- (4,5);
\draw [dotted] (5,1) -- (5,5);
\draw (2,2) node {$\varnothing$};
\draw (2,3.5) node {$\varnothing$};
\draw (4.5,2) node {$\varnothing$};
\draw (3.5,4) node {$\varnothing$};
\draw (4.5,4) node {$\varnothing$};
\draw (1,2) [fill] circle (3pt);
\draw (0.5,1.75) node {{\tiny $p$}};
\draw (2,4) [fill] circle (3pt);
\draw (2,4.5) node {{\tiny $q$}};
\draw (4,1) [fill] circle (3pt);
\draw (4,0.6) node {{\tiny $s$}};
\draw (5,3) [fill] circle (3pt);
\draw (5.5,3) node {{\tiny $t$}};
\draw (3,4.5) [fill] circle (3pt);
\draw (3,4.5) node[right] {\tiny $u$};
\end{tikzpicture}
\end{minipage}
\begin{minipage}{.75\textwidth}\setlength{\parindent}{2em}
\noindent $C_7(\sigma, s) = C_*(\sigma, q,p,t,s)$ with \\
 \indent  $t$  such that $\sigma_t = \max \{\sigma_k \mid k < s \text{ and } \sigma_k < \sigma_s\}$,\\
 \indent  $u$ such that $\sigma_u = \min \{\sigma_k \mid k < t \text{ and } \sigma_k > \sigma_t\}$,\\
 \indent $p$  such that $\sigma_p = \min \{\sigma_k \mid \sigma_k > \sigma_u \text{ and } t < k < s \}$ and\\
 \indent $q = \max \{k \mid k < p \text{ and } \sigma_u \leq \sigma_k < \sigma_p\}$\\
\end{minipage}
\begin{minipage}{.25\textwidth}
\begin{tikzpicture}[scale=.5]
\useasboundingbox (1,0) (6.25,6);
\fill [Grisfill] (0,3) rectangle (1,6);
\fill [Grisfill] (1,5) rectangle (4,6);
\fill [Grisfill] (4,0) rectangle (5,6);
\fill [Grisfill] (0,0) rectangle (4,1);
\fill [Grisfill] (2,2) rectangle (4,3);
\draw (0.5,3) [fill] circle (3pt);
\draw (0.5,3) node[below] {\tiny $u$};
\draw [dotted] (0,1) -- (4,1);
\draw [dotted] (0,2) -- (4,2);
\draw [dotted] (0,3) -- (4,3);
\draw [dotted] (0,5) -- (4,5);
\draw [dotted] (1,3) -- (1,5);
\draw [dotted] (2,1) -- (2,5);
\draw [dotted] (4,0) -- (4,6);
\draw (1,1.5) node {$\varnothing$};
\draw (1,2.5) node {$\varnothing$};
\draw (1.5,4) node {$\varnothing$};
\draw (3,4) node {$\varnothing$};
\draw (3,1.5) node {$\varnothing$};
\draw (1,4) [fill] circle (3pt);
\draw (0.6,4.25) node {{\tiny $q$}};
\draw (3,5) [fill] circle (3pt);
\draw (3,5.5) node {{\tiny $p$}};
\draw (2,1) [fill] circle (3pt);
\draw (2,0.6) node {{\tiny $t$}};
\draw (4,2) [fill] circle (3pt);
\draw (4.5,2) node {{\tiny $s$}};
\end{tikzpicture}
\end{minipage}
\begin{minipage}{.75\textwidth}\setlength{\parindent}{2em}
\noindent $C_8(\sigma, s) = C_*(\sigma, p,q,s,t)$ with \\
 \indent $t = \min \{k \mid k > s \text{ and } \sigma_k > \sigma_s\}$,\\
 \indent $u$ such that  $\sigma_u = \max \{\sigma_k \mid \sigma_k > \sigma_t \text{ and } k > t \}$,\\
 \indent $v = \max \{k \mid k < u \text{ and } \sigma_k > \sigma_u\}$,\\
 \indent $p = \max \{k \mid k < v \text{ and } \sigma_t < \sigma_k < \sigma_u\}$ and \\
 \indent $q$  such that $\sigma_q = \min \{\sigma_k \mid \sigma_k > \sigma_u \text{ and } p < k \leq v \}$\\
\end{minipage}
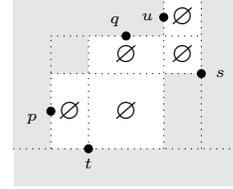
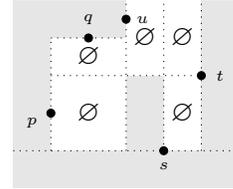
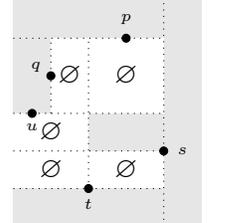
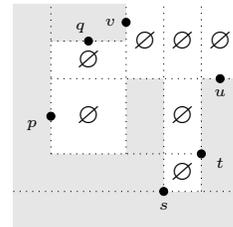
\begin{minipage}{.25\textwidth}
\begin{tikzpicture}[scale=.5]
\useasboundingbox (1,0) (6,6);
\fill [Grisfill] (0,1) rectangle (1,6);
\fill [Grisfill] (1,1) rectangle (3,2);
\fill [Grisfill] (1,5) rectangle (3,6);
\fill [Grisfill] (5,1) rectangle (6,4);
\fill [Grisfill] (0,0) rectangle (6,1);
\fill [Grisfill] (3,1) rectangle (4,4);
\draw [dotted] (0,1) -- (6,1);
\draw [dotted] (1,2) -- (5,2);
\draw [dotted] (1,4) -- (6,4);
\draw [dotted] (1,5) -- (3,5);
\draw [dotted] (1,2) -- (1,6);
\draw [dotted] (3,2) -- (3,6);
\draw [dotted] (4,1) -- (4,6);
\draw [dotted] (5,1) -- (5,6);
\draw (2,3) node {$\varnothing$};
\draw (2,4.5) node {$\varnothing$};
\draw (3.5,5) node {$\varnothing$};
\draw (4.5,1.5) node {$\varnothing$};
\draw (4.5,3) node {$\varnothing$};
\draw (4.5,5) node {$\varnothing$};
\draw (5.5,5) node {$\varnothing$};
\draw (1,3) [fill] circle (3pt);
\draw (0.5,2.75) node {{\tiny $p$}};
\draw (2,5) [fill] circle (3pt);
\draw (2,5.35) node[xshift=-3pt] {{\tiny $q$}};
\draw (4,1) [fill] circle (3pt);
\draw (5.5,4) [fill] circle (3pt);
\draw (5.5,4) node[below] {\tiny $u$};
\draw (4,0.6) node {{\tiny $s$}};
\draw (5,2) [fill] circle (3pt);
\draw (5.5,1.75) node {{\tiny $t$}};
\draw (3,5.5) [fill] circle (3pt);
\draw (3,5.5) node[left] {\tiny $v$};
\end{tikzpicture}
\end{minipage}
\begin{minipage}{.75\textwidth}\setlength{\parindent}{2em}
\noindent $C_9(\sigma, s) = C_*(\sigma, q,p,t,s)$ with \\
 \indent  $t$ such that $\sigma_t = \max \{\sigma_k \mid k < s \text{ and } \sigma_k < \sigma_s\}$,\\
 \indent  $u = \min \{k \mid k < t \text{ and } \sigma_k < \sigma_t\}$,\\
 \indent  $v$ such that $\sigma_v = \min \{\sigma_k \mid k < u \text{ and } \sigma_k > \sigma_u\}$,\\
 \indent  $p$ such that $\sigma_p = \min \{\sigma_k \mid \sigma_k > \sigma_v \text{ and } u < k < t \}$ and\\
 \indent  $q = \max \{k \mid k < p \text{ and } \sigma_v \leq \sigma_k < \sigma_p\}$
\end{minipage}
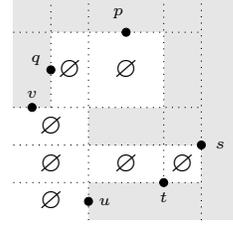
\begin{minipage}{.25\textwidth}
\begin{tikzpicture}[scale=.5]
\useasboundingbox (1,0) (6,6);
\fill [Grisfill] (0,3) rectangle (1,6);
\fill [Grisfill] (4,3) rectangle (5,5);
\fill [Grisfill] (1,5) rectangle (5,6);
\fill [Grisfill] (5,0) rectangle (6,6);
\fill [Grisfill] (2,0) rectangle (5,1);
\fill [Grisfill] (2,2) rectangle (5,3);
\draw (2,0.5) [fill] circle (3pt);
\draw (2,0.5) node[right] {\tiny $u$};
\draw [dotted] (0,1) -- (5,1);
\draw [dotted] (0,2) -- (5,2);
\draw [dotted] (0,3) -- (5,3);
\draw [dotted] (0,5) -- (5,5);
\draw [dotted] (1,3) -- (1,6);
\draw [dotted] (2,0) -- (2,6);
\draw [dotted] (4,1) -- (4,6);
\draw [dotted] (5,0) -- (5,6);
\draw (1,0.5) node {$\varnothing$};
\draw (1,1.5) node {$\varnothing$};
\draw (1,2.5) node {$\varnothing$};
\draw (1.5,4) node {$\varnothing$};
\draw (3,4) node {$\varnothing$};
\draw (3,1.5) node {$\varnothing$};
\draw (4.5,1.5) node {$\varnothing$};
\draw (1,4) [fill] circle (3pt);
\draw (0.6,4.25) node {{\tiny $q$}};
\draw (3,5) [fill] circle (3pt);
\draw (3,5.5) node[xshift=-3pt] {{\tiny $p$}};
\draw (4,1) [fill] circle (3pt);
\draw (4,0.6) node {{\tiny $t$}};
\draw (5,2) [fill] circle (3pt);
\draw (5.5,2) node {{\tiny $s$}};
\draw (0.5,3) [fill] circle (3pt);
\draw (0.5,3.35) node {\tiny $v$};
\end{tikzpicture}
\end{minipage}
\end{defn}

\begin{prop}\label{prop:valid=C_numero}
Let $\sigma$ be a $\ominus$-indecomposable permutation and $c$ a valid coloring of $\sigma$ which is not monochromatic.
Then there exists $s \in [1..|\sigma|]$ and $m \in [1..9]$ such that $c = C_m(\sigma, s)$.
\end{prop}

\begin{pf}
As $c$ is not monochromatic, then from Proposition~\ref{prop:monochromatic} $\sigma$ has at least a pattern $12$ which is not monochromatic.

If there is no increasing subsequence RG in $c$ then there is at least an increasing sequence GR in $c$. 
Thus from Propostion~\ref{prop:C_GR}, $c = C_{GR}(\sigma, i_{GR}, j_{GR})$. 
Moreover, $c$ has one of the $u$ shapes described in Proposition~\ref{prop:diagrammesGR}.
If the shape of $c$ is one of the two first shapes, then $j_{GR}$ is the leftmost point in the upper-right quadrant of $i_{GR}$
and $c = C_1(\sigma, i_{GR})$.
Otherwise the shape of $c$ is the third one and $i_{GR}$ is the topmost point in the bottom-left quadrant of $j_{GR}$
thus $c = C_2(\sigma, j_{GR})$.

If there is an increasing subsequence RG in $c$ but no increasing sequence GR, 
then from Proposition~\ref{prop:C_RG} $c = C_{RG}(\sigma, i_{RG}, j_{RG})$.
Moreover $c$ has one of the $5$ shapes described in Proposition~\ref{prop:diagrammesRG}.
If the shape of $c$ is one of the three first shapes, then $j_{GR}$ is the lowest point in the upper-right quadrant of $i_{GR}$
and $c = C_3(\sigma, i_{GR})$.
Otherwise the shape of $c$ is one of the two last shapes and $i_{GR}$ is the rightmost point in the bottom-left quadrant of $j_{GR}$
thus $c = C_4(\sigma, j_{GR})$.

If there is an increasing subsequence RG and an increasing sequence GR in $c$, 
then from Proposition~\ref{prop:C_*} $c = C_*(\sigma, i_{RG}, j_{RG}, i_{GR}, j_{GR})$.
Moreover $c$ has one of the $5$ shapes described in Proposition~\ref{prop:diagrammes*}.

If the shape of $c$ is the first one, let $u$ be the rightmost point in the top left quadrant of $j_{GR}$ (maybe $u = j_{RG}$).
Then applying rule (\rmnum{2}) to $i_{GR}$ and $u$, $c$ has the following shape:

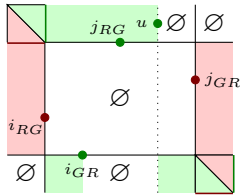
\begin{minipage}{.2\textwidth}
\begin{tikzpicture}[scale=.5]
\useasboundingbox (1,-0.5) (7.25,5.5);
\fill [Hfill] (0,1) rectangle (1,4);
\fill [Vfill] (1,4) rectangle (4,5);
\fill [Hfill] (5,1) rectangle (6,4);
\fill [Vfill] (1,0) rectangle (2,1);
\fill [Vfill] (4,0) rectangle (5,1);
\draw (0,1) -- (6,1);
\draw (0,4) -- (6,4);
\draw (1,0) -- (1,5);
\draw (5,0) -- (5,5);
\draw [dotted] (4,0) -- (4,5);
\draw (3,2.5) node {$\varnothing$};
\draw (0.5,0.5) node {$\varnothing$};
\draw (3,0.5) node {$\varnothing$};
\draw (5.5,4.5) node {$\varnothing$};
\draw (4.5,4.5) node {$\varnothing$};
\Hpoint{1}{2};
\draw (0.5,1.75) node {{\tiny $i_{RG}$}};
\Vpoint{3}{4};
\draw (2.7,4.4) node {{\tiny $j_{RG}$}};
\Vpoint{4}{4.5};
\draw (4,4.5) node[left] {{\tiny $u$}};
\Vpoint{2}{1};
\draw (2,0.6) node {{\tiny $i_{GR}$}};
\Hpoint{5}{3};
\draw (5.75,3) node {{\tiny $j_{GR}$}};
\zoneRG{1}{4}{1}
\zoneGR{6}{0}{1}
\end{tikzpicture}
\end{minipage}
\begin{minipage}{.75\textwidth}
Thus $i_{GR}$ is the rightmost point on the left of $u$ below $j_{GR}$. 
Moreover $i_{RG}$ is the rightmost point on the topleft quadrant of $i_{GR}$ below $j_{GR}$.
Finally $j_{RG}$ is the lowest point on the right of $i_{GR}$ and on the left of $u$.
Hence $c = C_5(\sigma, j_{GR})$.
\end{minipage}

If the shape of $c$ is the second one of Proposition~\ref{prop:diagrammes*}, 
let $u$ be the rightmost point in the top right quadrant of $j_{RG}$ (maybe $u = j_{RG}$).
From rule (\rmnum{8}) applied to $j_{RG}$ and $i_{GR}$, $u < i_{GR}$.
Then applying rule (\rmnum{2}) to $i_{RG}$ and $j_{GR}$ 
and applying rule (\rmnum{1}) to $j_{RG}$ and $u$ if $u \neq j_{RG}$, $c$ has the following shape:

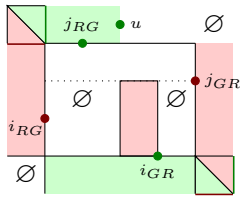
\begin{minipage}{.2\textwidth}
\begin{tikzpicture}[scale=.5]
\useasboundingbox (1,-0.5) (5,5.5);
\fill [Hfill] (0,1) rectangle (1,4);
\fill [Vfill] (1,4) rectangle (3,5);
\fill [Hfill] (5,1) rectangle (6,4);
\fill [Vfill] (1,0) rectangle (5,1);
\draw [Hfill] (3,1) rectangle (4,3);
\draw (0,1) -- (5,1);
\draw (0,4) -- (5,4);
\draw (1,0) -- (1,5);
\draw (5,0) -- (5,4);
\draw [dotted] (1,3) -- (5,3);
\draw (2,2.5) node {$\varnothing$};
\draw (4.5,2.5) node {$\varnothing$};
\draw (0.5,0.5) node {$\varnothing$};
\draw (5.5,4.5) node {$\varnothing$};
\Hpoint{1}{2};
\draw (0.5,1.75) node {{\tiny $i_{RG}$}};
\Vpoint{2}{4};
\draw (2,4.5) node {{\tiny $j_{RG}$}};
\Vpoint{4}{1};
\draw (4,0.6) node {{\tiny $i_{GR}$}};
\Hpoint{5}{3};
\draw (5.75,3) node {{\tiny $j_{GR}$}};
\zoneRG{1}{4}{1}
\zoneGR{6}{0}{1}
\Vpoint{3}{4.5}
\draw (3,4.5) node[right] {\tiny $u$};
\end{tikzpicture}
\end{minipage}
\begin{minipage}{.75\textwidth}
Thus $j_{GR}$ is the leftmost point in the upper right quadrant of $i_{GR}$ and
$u$ is the rightmost point in the upper left quadrant of $j_{GR}$.
Moreover $i_{RG}$ is the rightmost point to the left of $u$, below $j_{GR}$ and above $i_{GR}$. 
Finally, $j_{RG}$ is the lowest point in the upper left quadrant of $j_{GR}$ and to the right of $i_{RG}$. 
Hence $c = C_6(\sigma, i_{GR})$.
\end{minipage}

If the shape of $c$ is the third one of Proposition~\ref{prop:diagrammes*}, 
let $u$ be the lowest point in the lower left quadrant of $i_{RG}$ (maybe $u = i_{RG}$).
From rule (\rmnum{7}) applied to $i_{RG}$ and $j_{GR}$, $\sigma_u > \sigma_{j_{GR}}$.
Then applying rule (\rmnum{1}) to $i_{GR}$ and $j_{RG}$ 
and applying rule (\rmnum{2}) to $u$ and $i_{RG}$ if $u \neq i_{RG}$, $c$ has the following shape: 

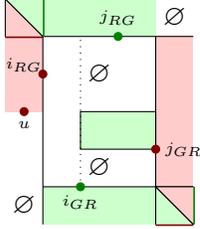
\begin{minipage}{.15\textwidth}
\begin{tikzpicture}[scale=.5]
\useasboundingbox (1,-1) (6.25,6);
\fill [Hfill] (0,3) rectangle (1,5);
\fill [Vfill] (1,5) rectangle (4,6);
\fill [Hfill] (4,1) rectangle (5,5);
\fill [Vfill] (1,0) rectangle (4,1);
\Hpoint{0.5}{3};
\draw (0.5,3) node[below] {\tiny $u$};
\draw [Vfill] (2,2) rectangle (4,3);
\draw (1,1) -- (5,1);
\draw (0,5) -- (5,5);
\draw (1,0) -- (1,5);
\draw (4,0) -- (4,5);
\draw (2.5,4) node {$\varnothing$};
\draw (2.5,1.5) node {$\varnothing$};
\draw (0.5,0.5) node {$\varnothing$};
\draw (4.5,5.5) node {$\varnothing$};
\Hpoint{1}{4};
\draw (0.5,4.25) node {{\tiny $i_{RG}$}};
\Vpoint{3}{5};
\draw (3,5.5) node {{\tiny $j_{RG}$}};
\Vpoint{2}{1};
\draw (2,0.6) node {{\tiny $i_{GR}$}};
\Hpoint{4}{2};
\draw (4.75,2) node {{\tiny $j_{GR}$}};
\zoneRG{1}{5}{1}
\zoneGR{5}{0}{1}
\draw [dotted] (2,1) -- (2,5);
\end{tikzpicture}
\end{minipage}
\begin{minipage}{.8\textwidth}
Thus $i_{GR}$ is the topmost point in the lower left quadrant of $j_{GR}$ and 
$u$ is the lowest point in the upper left quadrant of $i_{GR}$. 
Moreover $j_{RG}$ is the lowest point above $u$, to the right of $i_{GR}$ and to the left of $j_{GR}$. 
Finally, $i_{RG}$ is the rightmost point to the lower left of $j_{RG}$ and above $u$.
Hence $c = C_7(\sigma, j_{GR})$
\end{minipage}

If the shape of $c$ is the fourth one of Proposition~\ref{prop:diagrammes*},
let $u$ be the topmost point to the upright quadrant of $j_{GR}$
and $v$ be the rightmost point to the top-right quadrant of $j_{RG}$ (maybe $v = j_{RG}$).
Note that $u$ is above $i_{RG}$ as $u$ is above $x$ ($u$ is the topmost point) which is above $i_{RG}$.
Then applying rule (\rmnum{2}) to $i_{RG}$ and $u$
and applying rule (\rmnum{3}) to $j_{RG}$ and $v$ if $v \neq j_{RG}$, $c$ has the following shape:

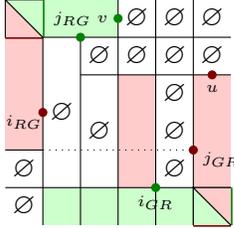
\begin{minipage}{.2\textwidth}
\begin{tikzpicture}[scale=.5]
\useasboundingbox (1,-1) (6,6);
\fill [Hfill] (0,2) rectangle (1,5);
\draw (0.5,1.5) node {$\varnothing$};
\draw (4.5,5.5) node {$\varnothing$};
\draw (5.5,5.5) node {$\varnothing$};
\draw (3.5,5.5) node {$\varnothing$};
\draw (5.5,4.5) node {$\varnothing$};
\fill [Vfill] (1,5) rectangle (3,6);
\fill [Hfill] (5,1) rectangle (6,4);
\fill [Vfill] (1,0) rectangle (5,1);
\fill [Hfill] (3,1) rectangle (4,4);
\draw (0,1) -- (6,1);
\draw (0,2) -- (1,2);
\draw (2,4) -- (6,4);
\draw (0,5) -- (6,5);
\draw (1,0) -- (1,6);
\draw (2,0) -- (2,5);
\draw (3,0) -- (3,6);
\draw (4,0) -- (4,6);
\draw (5,0) -- (5,6);
\draw [dotted] (1,2) -- (5,2);
\draw (2.5,2.5) node {$\varnothing$};
\draw (2.5,4.5) node {$\varnothing$};
\draw (3.5,4.5) node {$\varnothing$};
\draw (1.5,3) node {$\varnothing$};
\draw (4.5,3.5) node {$\varnothing$};
\draw (4.5,1.5) node {$\varnothing$};
\draw (0.5,0.5) node {$\varnothing$};
\draw (4.5,4.5) node {$\varnothing$};
\draw (4.5,2.5) node {$\varnothing$};
\Hpoint{1}{3};
\draw (0.5,2.75) node {{\tiny $i_{RG}$}};
\Vpoint{2}{5};
\draw (2,5.5) node[xshift=-3pt] {{\tiny $j_{RG}$}};
\Vpoint{4}{1};
\Hpoint{5.5}{4};
\draw (5.5,4) node[below] {\tiny $u$};
\draw (4,0.6) node {{\tiny $i_{GR}$}};
\Hpoint{5}{2};
\draw (5.75,1.75) node {{\tiny $j_{GR}$}};
\zoneRG{1}{5}{1}
\zoneGR{6}{0}{1}
\Vpoint{3}{5.5};
\draw (3,5.5) node[left] {\tiny $v$};
\end{tikzpicture}
\end{minipage}
\begin{minipage}{.75\textwidth}
Thus $j_{GR}$ is the leftmost point in the up right quadrant of $i_{GR}$. 
Point $u$ is the topmost point in the upper right quadrant of $j_{GR}$. 
Point $v$ is the rightmost point in the upper left quadrant of $u$. 
Then $i_{RG}$ is the rightmost point to the left of $v$, below $u$ and above $i_{GR}$. 
At last, $j_{RG}$ is the lowest point above $u$, to the right of $i_{RG}$ and to the left of $v$.
Hence $c = C_8(\sigma, i_{GR})$
\end{minipage}

If the shape of $c$ is the last one of Proposition~\ref{prop:diagrammes*},
let $u$ be the leftmost point in the lower left quadrant of $i_{GR}$
and $v$ be the lowest point in the lower left quadrant of $i_{RG}$ (maybe $v = i_{RG}$).
Note that $u$ is to the left of $j_{RG}$ as it is to the left of $y$ ($u$ is the leftmost point) and $y$ is to the left of $j_{RG}$. 
Then applying rule (\rmnum{1}) to $u$ and $j_{RG}$
and applying rule (\rmnum{2}) to $v$ and $i_{RG}$ if $v \neq i_{RG}$, $c$ has the following shape:

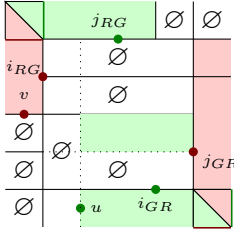
\begin{minipage}{.2\textwidth}
\begin{tikzpicture}[scale=.5]
\useasboundingbox (1,-1) (6,6);
\fill [Hfill] (0,3) rectangle (1,5);
\draw (0.5,1.5) node {$\varnothing$};
\draw (0.5,2.5) node {$\varnothing$};
\draw (4.5,5.5) node {$\varnothing$};
\draw (5.5,5.5) node {$\varnothing$};
\draw (0,3) -- (5,3);
\fill [Vfill] (1,5) rectangle (4,6);
\fill [Hfill] (5,1) rectangle (6,5);
\fill [Vfill] (2,0) rectangle (5,1);
\fill [Vfill] (2,2) rectangle (5,3);
\Vpoint{2}{0.5};
\draw (2,0.5) node[right] {\tiny $u$};
\draw (0,1) -- (6,1);
\draw (0,2) -- (1,2);
\draw (1,4) -- (6,4);
\draw (0,5) -- (6,5);
\draw (1,0) -- (1,6);
\draw[dotted] (2,0) -- (2,5);
\draw (4,5) -- (4,6);
\draw (5,0) -- (5,6);
\draw [dotted] (1,2) -- (5,2);
\Hpoint{0.5}{3};
\draw (0.5,3.5) node {\tiny $v$};
\draw (3,4.5) node {$\varnothing$};
\draw (3,3.5) node {$\varnothing$};
\draw (3,1.5) node {$\varnothing$};
\draw (0.5,0.5) node {$\varnothing$};
\draw (1.5,2) node {$\varnothing$};
\Hpoint{1}{4};
\draw (0.5,4.25) node {{\tiny $i_{RG}$}};
\Vpoint{3}{5};
\draw (3,5.5) node[xshift=-3pt] {{\tiny $j_{RG}$}};
\Vpoint{4}{1};
\draw (4,0.6) node {{\tiny $i_{GR}$}};
\Hpoint{5}{2};
\draw (5.75,1.75) node {{\tiny $j_{GR}$}};
\zoneRG{1}{5}{1}
\zoneGR{6}{0}{1}
\end{tikzpicture}
\end{minipage}
\begin{minipage}{.75\textwidth}
Thus $i_{GR}$ is the topmost point in the lower left quadrant of $j_{GR}$ and 
$u$ is the leftmost point in the lower left quadrant of $i_{GR}$. 
Moreover $v$ is the lowest point in the upper left quadrant of $u$ and 
$j_{RG}$ is the lowest point above $v$ and to the right of $u$ and to the left of $i_{GR}$.
Finally, $i_{RG}$ is the rightmost point in the lower left quadrant of $j_{RG}$ and above $v$. 
Hence $c = C_9(\sigma, j_{GR})$
\end{minipage}
\end{pf}

\begin{prop}\label{prop:linearTestC_numero}
Let $\sigma$ be a permutation, $s \in [1..|\sigma|]$ and $m \in [1..9]$. 
Then we can compute $C_m(\sigma, s)$, test whether all points of $\sigma$ are colored and check whether $C_m(\sigma, s)$ is valid
in linear time w.r.t. $|\sigma|$.
\end{prop}

\begin{pf}
Theorems~\ref{thm:RGIncreasing} and \ref{thm:GRIncreasing} and \ref{prop:Check-Valid-linear}
\end{pf}

\subsection{Algorithm and linear number of sortings for $\ominus$-indecomposable permutations}

\begin{algorithm}[H]
%  \SetAlgoLined
  \KwData{$\sigma$ a $\ominus$-indecomposable permutation}
  \KwResult{The set $E$ of valid colorings of $\sigma$}

  \For{$c$ bicoloring of $\sigma$ unicolor \R or unicolor \G}{
 	If $c$ is valid then add $c$ to $E$\;
    }
  \For{$s$ from $1$ to $|\sigma|$}{
      \For{$m$ from $1$ to $9$}{
        $c = C_m(\sigma, s)$\;
 	If all points of $\sigma$ are colored and $c$ is valid then add $c$ to $E$\;
      }
    }
  \caption{ColoringIndecOptimal$(\sigma)$}\label{alg:Valid-colorings-indec_optimal}
\end{algorithm}

Given any point $s$ in the permutation the Algorithm decides if the permutation can be colored in each possible case depicted in Propositions~\ref{prop:diagrammesGR},\ref{prop:diagrammesRG} and \ref{prop:diagrammes*}. 
Note that diagrams of Propositions~\ref{prop:diagrammesGR},\ref{prop:diagrammesRG} and \ref{prop:diagrammes*} depend on $v$ points $i_{GR},i_{RG},j_{GR},j_{RG}$. 
Indeed, we prove in section~\ref{ssec:root} that any diagram can be rooted in one point -- say $i_{RG}$ for example -- 
and from this points, we can find in linear time any other points -- $i_{GR},j_{GR},j_{RG}$ for instance --. 
Then, we color the permutations with respect to the different zones defined in the diagram. 
In this process, some points may be uncolored, meaning that they lie in empty zone of the diagram hence have to be rejected. 
At last, we have a coloring according to diagram and we have to check that this coloring is valid.

\begin{thm}\label{thm:Algo_ominus-indec_optimal}
A $\ominus$-indecomposable permutation of size $n$ has at most $9n+2$ valid colorings. 
Those colorings can be computed using Algorithm~\ref{alg:Valid-colorings-indec_optimal} 
in time ${\mathcal O}(n^{2})$ which is optimal.
\end{thm}

\begin{pf}
This is a direct consequence of Propositions~\ref{prop:valid=C_numero} and \ref{prop:linearTestC_numero}, 
except for the optimality.
Proposition~\ref{prop:nbColoringIdentity} below implies that 
the size of the set of valid colorings of the identity of size $n$ is $2n^2$, proving the optimality. 
\end{pf}

\begin{prop}\label{prop:nbColoringIdentity}
For all $n$ the identity of size $n$ has exactly $2n$ valid colorings.
%(and $n$ pushall sorting stack words).
\end{prop}

%\marginpar{prouver parenthese? Ou l'enlever?} TODO : rajouter la parenthese et sa preuve en version finale

\begin{pf}
Let $\sigma$ be the identity of size $n$. 
For all $k$ between $1$ and $n$ let $C_{RG}^k$ (resp. $C_{GR}^k$) be the coloring of $\sigma$ such that for all $i$, $\sigma$ is in \R (resp. \G) if $i \leq k$ and in \G (resp. \R) otherwise.
Then it is straightforward to check using Proposition~\ref{prop:rulesR8} that $C_{RG}^k$ (resp. $C_{GR}^k$) is a valid coloring of $\sigma$. 
Conversely if $c$ is a valid coloring of the identity, rules $(\rmnum{3})$ and $(\rmnum{4})$ of $\mathcal{R}_8$ imply that there are at most one pair of consecutive points whoses colors are different. So $c$ is some $C_{RG}^k$ or some $C_{GR}^k$.
\end{pf}

The property of having a linear number of sortings is not a special case of the identity.
Indeed there are some simple permutations that also have a linear number of sortings, as shown in the next proposition.

\begin{prop}
Permutations $\sigma^{(n)} = (2n-1)(2n-3)(2n)(2n-5)(2n-2)(2n-7)(2n-4)\ldots 5\,8\,3\,6\,1\,4\,2$ of size $2n$ have at least $2n-3$ valid colorings. 
\end{prop}
\begin{proof}
To prove the result, we exhibit $2n-3$ colorings. 
We look at set of four points of $\sigma$ whose indices (resp. values) are consecutive and which form a pattern $2\,4\,1\,3$ (resp. $3\,1\,4\,2$). 
Notice that they can be taken to be $\{i_{RG}, i_{GR},j_{RG},j_{GR}\}$ in a valid coloring of $\sigma$ respecting to the last (resp. third) diagram of Proposition~\ref{prop:diagrammes*}, as shown in the figure below. 
This way we obtain $2n-3$ valid colorings of $\sigma$.

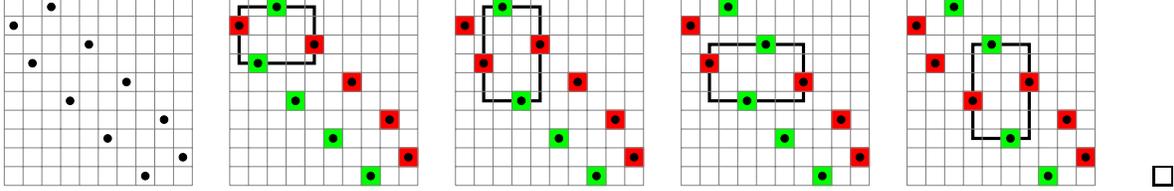
\noindent\begin{tikzpicture}[scale=.25]
\permutation{9,7,10,5,8,3,6,1,4,2}
\end{tikzpicture}
\hfill
\begin{tikzpicture}[scale=.25]
\draw[very thick] (1.5,7.5) rectangle +(4,3);
\draw[Vfill,V] (2,7) rectangle +(1,1); 
\draw[Hfill,H] (1,9) rectangle +(1,1); 
\draw[Vfill,V] (4,5) rectangle +(1,1); 
\draw[Hfill,H] (7,6) rectangle +(1,1); 
\draw[Hfill,H] (9,4) rectangle +(1,1); 
\draw[Hfill,H] (10,2) rectangle +(1,1); 
\draw[Vfill,V] (3,10) rectangle +(1,1); 
\draw[Hfill,H] (5,8) rectangle +(1,1); 
\draw[Vfill,V] (6,3) rectangle +(1,1); 
\draw[Vfill,V] (8,1) rectangle +(1,1); 

\permutation{9,7,10,5,8,3,6,1,4,2}
\end{tikzpicture}
\hfill
\begin{tikzpicture}[scale=.25]
\draw[very thick] (2.5,5.5) rectangle +(3,5);
\draw[Hfill,H] (2,7) rectangle +(1,1); 
\draw[Hfill,H] (1,9) rectangle +(1,1); 
\draw[Vfill,V] (4,5) rectangle +(1,1); 
\draw[Hfill,H] (7,6) rectangle +(1,1); 
\draw[Hfill,H] (9,4) rectangle +(1,1); 
\draw[Hfill,H] (10,2) rectangle +(1,1); 
\draw[Vfill,V] (3,10) rectangle +(1,1); 
\draw[Hfill,H] (5,8) rectangle +(1,1); 
\draw[Vfill,V] (6,3) rectangle +(1,1); 
\draw[Vfill,V] (8,1) rectangle +(1,1); 

\permutation{9,7,10,5,8,3,6,1,4,2}
\end{tikzpicture}
\hfill
\begin{tikzpicture}[scale=.25]
\draw[very thick] (2.5,5.5) rectangle +(5,3);
\draw[Hfill,H] (2,7) rectangle +(1,1); 
\draw[Hfill,H] (1,9) rectangle +(1,1); 
\draw[Vfill,V] (4,5) rectangle +(1,1); 
\draw[Hfill,H] (7,6) rectangle +(1,1); 
\draw[Hfill,H] (9,4) rectangle +(1,1); 
\draw[Hfill,H] (10,2) rectangle +(1,1); 
\draw[Vfill,V] (3,10) rectangle +(1,1); 
\draw[Vfill,V] (5,8) rectangle +(1,1); 
\draw[Vfill,V] (6,3) rectangle +(1,1); 
\draw[Vfill,V] (8,1) rectangle +(1,1); 

\permutation{9,7,10,5,8,3,6,1,4,2}
\end{tikzpicture}
\hfill
\begin{tikzpicture}[scale=.25]
\draw[very thick] (4.5,3.5) rectangle +(3,5);
\draw[Hfill,H] (2,7) rectangle +(1,1); 
\draw[Hfill,H] (1,9) rectangle +(1,1); 
\draw[Hfill,H] (4,5) rectangle +(1,1); 
\draw[Hfill,H] (7,6) rectangle +(1,1); 
\draw[Hfill,H] (9,4) rectangle +(1,1); 
\draw[Hfill,H] (10,2) rectangle +(1,1); 
\draw[Vfill,V] (3,10) rectangle +(1,1); 
\draw[Vfill,V] (5,8) rectangle +(1,1); 
\draw[Vfill,V] (6,3) rectangle +(1,1); 
\draw[Vfill,V] (8,1) rectangle +(1,1); 

\permutation{9,7,10,5,8,3,6,1,4,2}
\end{tikzpicture}
\end{proof}

\subsection{Final algorithm}

Recall first that if a permutation is $\ominus$-decomposable, then it is \pushall if and only if 
each of the block of its decomposition is \pushall and that we can just push elements of the first block according to any sorting procedure of it, then elements of the second and so on, before popping out all the elements. 
This means that the different colorings for a $\ominus$-decomposable permutation is the product of all colorings for each block.

\begin{prop}\label{prop:Col=cartesianProduct}
Let $\sigma$ be a permutation and $Col(\sigma)$ the set of valid colorings of $\sigma$.
If $\sigma = \ominus[\pi_1, \dots ,\pi_k]$ then the map $c \rightarrow (c|\pi_1, \dots, c|\pi_k)$
is a bijection from $Col(\sigma)$ into $Col(\pi_1) \times \dots \times Col(\pi_k)$.
\end{prop}

\begin{pf}
Let $c$ be a valid coloring of $\sigma$, then $c$ avoids patterns \RRR, \GGG, \RRG and \GGR. 
Thus for all $i$, $c|\pi_i$ avoids patterns \RRR, \GGG, \RRG and \GGR hence is a valid coloring of $\pi_i$.
Conversely let $c_i \in Col(\pi_i)$ for all $i$. 
Then coloring points of $\sigma$ according to $(c_1, \dots c_k)$ 
({\em i.e.} according to $c_1$ for the $|\pi_1|$ first points of $\sigma$, 
according to $c_2$ for the $|\pi_2|$ following points and so on)
leads to a coloring $c$ of $\sigma$ which is valid.
Indeed assume that $c$ is not valid.
Then $c$ has a pattern \RRR, \GGG, \RRG or \GGR.
Let $p$ be such a pattern.
Then $p$ is not inside a block $\pi_i$ as $c_i$ is a valid coloring for all $i$.
If all points of $p$ are in different blocks $\pi_i$ then $p$ is $321$ which is excluded.
Thus there are one point of $p$ in a block $\pi_i$ and two points of $p$ in a block $\pi_j$.
If $i<j$ then $p$ begins with its greatest point, which is excluded as $p$ is \RRR, \GGG, \RRG or \GGR.
If $i>j$ then $p$ ends with its smallest point, which is excluded as $p$ is \RRR, \GGG, \RRG or \GGR.
As a consequence such a pattern $p$ does not exists and $c \in Col(\sigma)$, concluding the proof.
\end{pf}

\begin{algorithm}[H]
%  \SetAlgoLined
  \KwData{$\sigma$ a permutation}
  \KwResult{A linear description of the set $Col(\sigma)$ of valid colorings of $\sigma$}

  Compute the $\ominus$-decomposition of $\sigma$: $\sigma = \ominus[\pi_1, \dots ,\pi_k]$ with $\pi_i$ $\ominus$-indecomposable\;
  \For{$i$ from $1$ to $k$}{
 	Compute $Col(\pi_i)$ thanks to Algorithm~\ref{alg:Valid-colorings-indec_optimal}\;
      }
  Return $(Col(\pi_1), \dots, Col(\pi_k))$\;

  \caption{Colorings$(\sigma)$}\label{alg:colorings_general}
\end{algorithm}

\begin{prop}\label{prop:algo5}
Let $\sigma$ be a permutation of size $n$.
Then Algorithm~\ref{alg:colorings_general} gives a linear description of $Col(\sigma)$ in time ${\mathcal O}(n^{2})$.
\end{prop}

\begin{proof} 
The algorithm computes the $\ominus$-decomposition of $\sigma$: $\sigma = \ominus[\pi_1, \dots ,\pi_k]$ with $\pi_i$ $\ominus$-indecomposable.
This is done in linear time.
If $k=1$ then $\sigma$ is $\ominus$-indecomposable and $Col(\sigma) = Col(\pi_1)$.
We concludes thanks to Theorem~\ref{thm:Algo_ominus-indec_optimal}.
If $k>1$ then from Proposition~\ref{prop:Col=cartesianProduct}, $Col(\sigma) \approx Col(\pi_1) \times \dots \times Col(\pi_k)$.
For all $i$, $Col(\pi_i)$ has a size is smaller than $9|\pi_1|$ and is computed in $\mathcal{O}(|\pi|^2)$.
We concludes the proof noticing that $9|\pi_1| + \dots + 9|\pi_k| = 9|\sigma|$ and $|\pi_1|^2 + \dots +|\pi_k|^2 \leq |\sigma|^2$.
\end{proof}

% TODO : Dire que $D_n$ a $2^n$ coloriages (et representation de Col(D_n) = n ensembles de taille 2)

% TODO : borne sup sur le nb de coloriages ? Ou laisser ça pour l'article suivant

\begin{thm}
Using Algorithm \ref{alg:colorings_general},  we can decide in time ${\mathcal O}(n^{2})$ whether a permutation $\sigma$ of size $n$ is \pushall.
\end{thm}

\begin{proof}
By Theorem~\ref{thm:equivalenceColoringPushall}, a permutation $\sigma$ is \pushall if and only if it admits a valid coloring. 
Thus all we need is to test whether each set $Col(\pi_{i})$ returned by Algorithm \ref{alg:colorings_general} is non-empty 
with $\sigma  = \ominus[\pi_{1},\ldots,\pi_{k}]$ being the $\ominus$-decomposition of $\sigma$, and we conclude	 using Proposition \ref{prop:algo5}.
\end{proof}

\section{Conclusion}\label{sec:conclusion}

This article defines a new restriction of $2$-stacks sorting, namely $2$-stacks pushall sorting.
We characterize every possible pushall sorting of a permutation by means of a bi-coloring of the permutation.
Then we give an $\mathcal O(n^2)$ algorithm which computes a linear representation of all pushall sortings of a given permutation,
which thus decides if a permutation is \pushall.
We proove that this complexity is optimal.

More studies remain to be done on $2$-stacks pushall sorting.
First, a simpler mathematical %(rather than the current algorithmic)
characterization of \pushall permutations would be interesting.
Then, we could study more in depth the number of pushall sortings of a given permutation.
More generally it would be nice to compute the generating function of \pushall permutations,
or at least asymptotic bounds on this function.
But most importantly, this result is a step to the solve the general $2$-stack sorting,
which we do in a forthcoming article.

\bibliographystyle{plain}
\bibliography{biblio}
\end{document}